\documentclass[onefignum,onetabnum]{siamonline190516}

\headers{Risk-dependent centrality in economic and financial networks}{P. Bartesaghi, M. Benzi, G. P. Clemente, R. Grassi and E. Estrada}

\title{Risk-dependent centrality in economic and financial networks
}

\author{Paolo Bartesaghi\thanks{Department of Statistics and Quantitative Methods, University of Milano - Bicocca, Via Bicocca degli Arcimboldi 8, 20126, Milano, Italy
		(\email{paolo.bartesaghi@unimib.it}, \email{rosanna.grassi@unimib.it}).}
	\and 
	Michele Benzi\thanks{Scuola Normale Superiore, Piazza dei Cavalieri 7, 56126, Pisa, Italy
		(\email{michele.benzi@sns.it}).}
	\and Gian Paolo Clemente\thanks{Department of Mathematics for Economics, Financial and Actuarial Sciences, Università Cattolica del Sacro Cuore, Largo Gemelli 1, 20123, Milano
		(\email{gianpaolo.clemente@unicatt.it}).}
	\and Rosanna Grassi\footnotemark[2]
	\and Ernesto Estrada\thanks{Institute of Mathematics and Applications (IUMA), Universidad de Zaragoza, Pedro Cerbuna 12, E-50009 Zaragoza, Spain; ARAID Foundation, Government of Aragón, 50018 Zaragoza, Spain
		(\email{estrada66@unizar.es}).}
}

\usepackage{amsopn}

\usepackage{lipsum}
\usepackage{amsfonts}
\usepackage{graphicx}
\usepackage{epstopdf}
\usepackage{algorithmic}
\usepackage{color}
\usepackage{array}
\usepackage{float}
\usepackage{mathrsfs}
\usepackage{mathtools}
\usepackage{multirow}
\usepackage{amstext}
\usepackage{amsmath}
\usepackage{amssymb}

\def\1{\rm 1\!\!{\mathbb I}}

\usepackage{subfig}
\makeatother

\usepackage{enumitem}
\setlist[enumerate]{leftmargin=.5in}
\setlist[itemize]{leftmargin=.5in}

\newsiamremark{remark}{Remark}
\newsiamremark{hypothesis}{Hypothesis}
\crefname{hypothesis}{Hypothesis}{Hypotheses}
\newsiamthm{claim}{Claim}

\begin{document}
	
	\maketitle
	
	\begin{abstract}
		Node centrality is one of the most important and widely used concepts
		in the study of complex networks. Here, we extend the paradigm of
		node centrality in financial and economic networks to consider the
		changes of node ``importance'' produced not only by the variation
		of the topology of the system but also as a consequence of the external   levels of risk to which the network as a whole is submitted. Starting   from the ``Susceptible-Infected'' (SI) model of epidemics and its relation to the communicability functions of networks we develop a series of risk-dependent centralities for nodes in (financial and economic) networks. We analyze here some of the most important mathematical properties of these risk-dependent centrality measures. In particular, we study the newly observed phenomenon of ranking interlacement, by means of which two entities may interlace their ranking positions in terms of risk in the network as a consequence of the change in the external conditions only, i.e., without any change in the topology. We test the risk-dependent centralities by studying two real-world systems: the network generated by collecting assets of the S\&P 100 and the corporate board network of the US top companies, according to Forbes in 1999. We found that a high position in the ranking of the analyzed financial companies according to their risk-dependent centrality corresponds to companies more sensitive to the external market variations during the periods of crisis.
	\end{abstract}
	
	\begin{keywords}
		Communicability, Complex networks, SI models, Centrality measures, Risk propagation
	\end{keywords}
	
	\begin{AMS}
		05C82, 90B15, 91G80
	\end{AMS}

\section{Introduction}

Modern economic and financial systems are characterized by a vast
collection of interacting agents \cite{Allen_Babus_2009,Battiston_DelliGatti 2012,Bongini, Furusawa_Konishi_2007,  Goyal_Moraga-Gonzalez,Kirman_1997}.
In economic systems, for instance, the interdependence among entities
characterizes the trade and exchange of goods in non-anonymous markets
as well as in risk sharing agreements in developing countries \cite{Allen_Babus_2009}.
In this framework, the agents' interaction is responsible for the
nature of the relations between the individual behavior and the aggregate behavior \cite{Kirman_1997}.

The human factor which underlies these economic and financial systems
is also characterized by the interconnectivity. The existence of networks of interpersonal relations has been empirically observed to constitute a fundamental factor in shaping the inter-institution networks, or in accounting for the networks of risk-sharing agreements \cite{Fafchamp_Gupert_2007,Fafchamp_Lund_2003},
the formation of buyer-seller networks \cite{Corominas-Bosch_2004, Gale_Kariv_2007, Kranton_Minehart_2001},
product adoption decisions \cite{Economides_1996, Katz_Shapiro}, diffusive
processes \cite{Galeotti_Goyal_2007,Golub_Jackson_2007,Lopez-pintado 2008},
industrial organization \cite{Goyal_Moraga-Gonzalez}, trade agreements
\cite{Furusawa_Konishi_2007} and even for the existence of interbank
networks \cite{Allen_Babus_2009}. This is not surprising as humans
are responsible for the execution of deals between the institutions
to which they belong to \cite{Cohen_Frazzini_Malloy_2008,Kramarz_Thesmar_2013,Nguyen_Dang_2007}.

From a mathematical perspective all these interdependencies between
economic and financial entities can be captured by the formal concept of
network, in which nodes represent the entities (individuals, firms,
countries, etc.) and edges account for the relations between such
entities, ranging from social relations to trade agreements \cite{Estrada book}.
Hence, it is possible to use the tools of network theory to analyze
the structure, the evolution and the dynamic processes that take
place on these systems. On one side, researchers have studied the
topological properties of these networks (sometimes called \textit{static
	analysis}), which do not assume mechanisms of transmission of effects
through the economic and financial entities \cite{Gai_2011,Kanno_2015,Tirado_2012}.
Among such studies, it is frequent to find analyses of clusters formed
by groups of institutions, as well as the centrality of individual nodes in the networks \cite{Battiston_Puliga2012,Boss_2004,Puhr_2012}.
Specifically, centrality measures (see Chapter 5 in \cite{Estrada book} for a detailed analysis) are topological characterizations of the nodes and their neighborhood in a network. 
In the analysis of financial and economic networks, the use of centrality measures is not so effective,
as the classical ones provide a static view of the network and even other measures based on dynamic processes, such as random walks based centralities \cite{random walk centrality}, do not capture the changing conditions to which these networks could be submitted in relatively short periods of time. As an illustrative example, let us consider a hypothetical interbank network for which we are interested in analyzing the risk-dependent
exposure of the various entities of the system. Any centrality measure will point out a specific and static ranking of the nodes. However, a bank which is very central at a low-level of external risk is not necessarily central when such level of external risk increases, and vice versa.
On the other hand, the propagation of shocks
through these networks is considered and it is usually known as \textit{dynamic analysis} \cite{Amini, Cocco_2009,Cont,  Elsinger_2006,contagion_4, Haldane_May, Glass}.
In these studies, a specific way of transmission of these shocks through
the network is assumed -- as in the case of ``Susceptible-Infected'' and ``Susceptible-Infected-Recovered''
epidemiological models  \cite{Lee,Mei,Contagion on networks} -- and then a systemic risk analysis 
is based on the contagion effects observed through such models.

In this work we develop a mathematical model to account
	for the risk exposure of an entity in a networked (economic or financial)
	system. This model is based on the relation between the Susceptible-Infected
	epidemiological model and the so-called communicability functions
	of a network \cite{Communicability}. Using this connection we derive new centrality indices that quantify the level of risk at which an entity is exposed to as a function of the global external level of risk. 
	Our approach takes advantage of the benefits of both static and dynamic analyses. Indeed, unlike the standard approaches followed in the literature, these risk-dependent centralities are not static indices, as most of centrality indices are, but they vary with the change of the external global risk level at which the system is submitted to.
	More importantly, the ranking of the nodes in these networks also depends on this global external
	level of risk. This means that an entity -- a node in the network -- which
	is at low (high) level of risk under external conditions can be at
	high (low) level under different conditions. 
	
	We test our model by
	using two different systems, a network of assets based on the daily
	returns of the components of the S\&P 100 for the period ranging from
	January 2001 to December of 2017 and a network representing the interconnection
	between companies in the US top corporates according to Forbes in
1999. In the first case we extract the essential information about
asset correlations through the minimum spanning tree. We measure how
the centrality of the assets changes at different values of the
external risk. What emerges is a high volatility in the rankings during
the financial crisis of 2007-2008, when the node centrality proves
to be more sensitive to the external risk. In the case of the corporate
network we analyze a sample of significant companies, looking
	for a correlation between the shareholders value creation (SVC) and
	their behavior during and after the crisis period at which data were
	collected. We find that a remarkable increase in their risk-centrality
	ranking during a crisis corresponds to a less resilient reaction to
	the external market turmoil.

The paper is structured as follows. In Subsection \ref{sec:rel} we recall the main literature about the use of epidemiological models for modeling financial contagion and we motivate the choice of a  Susceptible-Infected model. The necessary mathematical preliminaries are given in Section \ref{sec:prelim}. Therefore, we describe a Susceptible-Infected (SI) model on a financial network (Section \ref{sec:model}) and we define the risk-dependent centrality proving some mathematical properties (Section \ref{sec:RDcent}). We perform numerical analyses of the proposed centrality for random networks (Section \ref{sec:Numerical}), then we apply the proposed measure to real-world financial networks  (Section \ref{sec:RWfin}) and we analyze the ranking interlacement problem (Section \ref{sec:Ranking}). Section \ref{sec:COV} remarks how the proposed model could provide additional insights in the analysis of the economic and financial impacts of the crisis related to the diffusion of the new coronavirus named SARS-COV-2. 
Conclusions follow in Section \ref{sec:conclusions}.

\subsection{Related literature and motivations}\label{sec:rel}
The process in which one financial institution spreads negative effects
to another institution resembles very much the propagation of epidemics
on networks \cite{May_2008,Haldane_May}. The fact that such processes
are known as ``financial contagion'' already captures part of these
similarities. Then, it is not strange that epidemiological models
are frequently used to capture the subtleties of financial contagion
processes. There are many of such compartmental models in epidemiology,
but the most widely used for modeling financial contagion are the Susceptible-Infected-Recovered
(SIR) \cite{contagion_9,SIR_2,SIR_3,Kost,SIR_5,SIR_6} and the Susceptible-Infected-Susceptible
(SIS) \cite{SIS_1,SIS_2} ones. They are not only used to model financial
contagion per se, but also for the propagation of rumors and innovations
of interest for financial institutions \cite{SIR_rumors,SIR_rumors-1}.
These models are well-suited in depicting financial contagion because
they do not require arbitrary assumption on loss rates and balance
sheets. As remarked by Toivanen \cite{contagion_8}, they capture the
psychological aspects of contagion process ``by relating a bank\textquoteright s
relative financial strength with the perceived counterparty risk and
expectations''.

The previously mentioned SIS/SIR models and their variants are mainly
used in studying the dynamics of contagion in a system in a post-mortem
way. As it is well-known, both SIS and SIR models are characterized
by the presence of a threshold $\tau$, which is defined as the reciprocal
of the principal eigenvalue $\lambda_{1}$ of the adjacency matrix.
The below-the-threshold or above-the-threshold behavior of the spreading
process depends on whether the effective infection rate is less than
or greater than such a threshold. Below the threshold, we have the
extinction of the contagion and above the threshold a non-zero fraction
of infected nodes persists in the network even over a wide range of
timescales. The effective infection rate depends on both the infection
rate per link $\gamma$ and on the curing or recovering rate $\delta$.
For instance in Figure \ref{SISModel}(a) we illustrate the evolution
of a contagion dynamics for an Erd\H{o}s-Rényi graph with 100 nodes
and connection probability 0.1 by using the SIS model. The principal
eigenvalue of the adjacency matrix is $\lambda_{1}\approx10.71$ so that the epidemic
threshold is $\tau\approx0.093$. The infectivity rate per link is
0.002 for both curves and the initial infection probability is 0.2
(20 nodes over 100 initially infected). The curing rate is 0.001
for the dashed red line (epidemic) and 0.04 for the solid blue line
(extinction). Then, the effective infection rate is $2>0.093$ for
the the dashed red line (epidemic) and $0.05<0.093$ for the solid
blue line (extinction).

\begin{figure}[H]
	\begin{centering}
	\subfloat[]{\begin{centering}
			\includegraphics[width=0.45\textwidth]{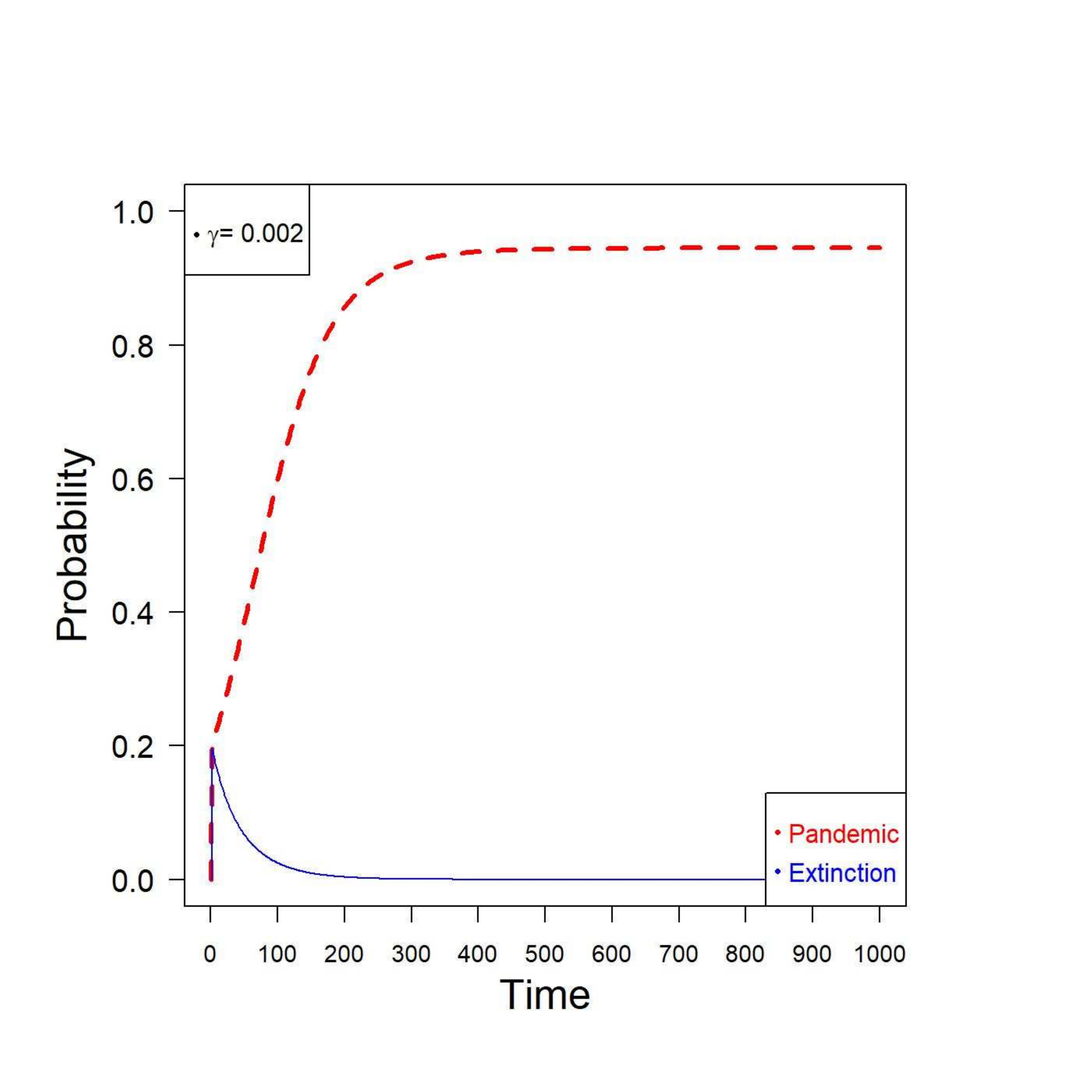}
			\par\end{centering}
	}\subfloat[]{\begin{centering}
			\includegraphics[width=0.45\textwidth]{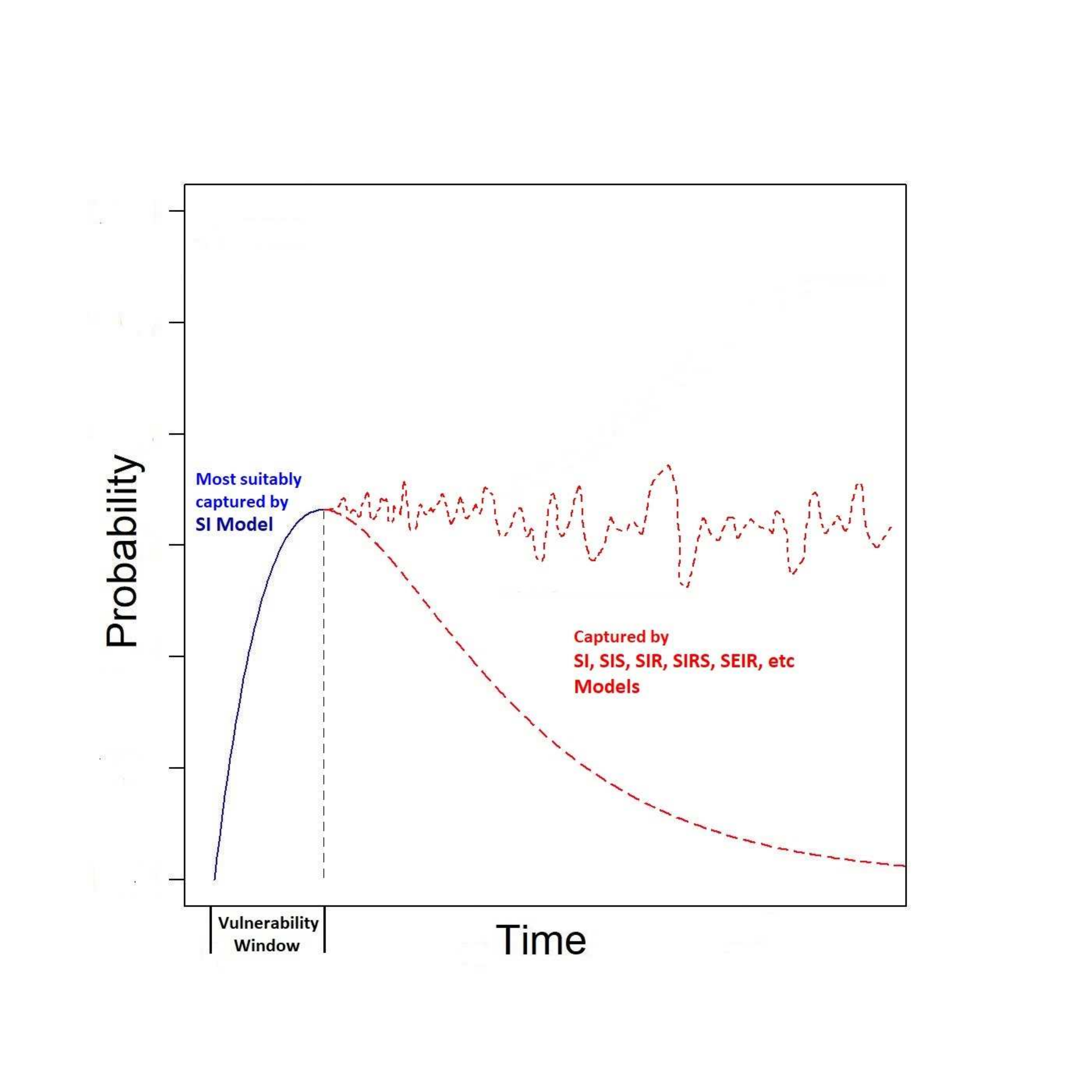}
			\par\end{centering}
	}
    \par\end{centering}
		\caption{(a) Evolution of a SIS dynamics over an Erd\H{o}s-Rényi graph (b)
		Evolution of contagion before and after the window of vulnerability.
		Adapted from Lee et al.\cite{lee_2019}}
	
	\label{SISModel}
\end{figure}

In this work we are interested in the very early signals that the
system can provide for alerting about a propagation of a financial
contagion. In this case it is very important to consider the window
of vulnerability between the time the contagion phenomenon is firstly
recognized and the time an action is taken to face the infection.
This window could be arbitrarily wide. In any real condition, there
is a non negligible time interval in which a recovery tool is not
available yet and the recovering rate is equal to zero. It has been
recently shown by Lee et al. \cite{lee_2019} that, within
this window, the spreading phenomenon is better described by a SI
model than by any other model with a non-null recovering rate, e.g.,
SIR and SIS (see Figure \ref{SISModel}(b)). In this framework, a key point is to predict
the ``most at risk'' nodes in the network. Therefore, we are interested
in the early times of the epidemic where it is possible to limit or
avoid the distress propagation by introducing specific measures on
the risky nodes in the network.

Moreover, in order to be effective in reducing the spreading phenomenon,
the curing rate has to be large enough. More precisely, since $\lambda_{1}>\max\left(\bar{k},\sqrt{k_{\max}}\right)$ (being $\bar{k}$ the mean degree and $k_{\max}$ the maximum degree),
the curing rate has to be at least $\delta>\gamma\cdot\max\left(\bar{k},\sqrt{k_{\max}}\right)$
to get a below-the-threshold behavior. But for a big real network
$\sqrt{k_{\max}}$ can be very large, even if the mean degree
is small. This implies that $\delta$ has to be significantly bigger than $\gamma$,
or in other words, the infection significantly weaker than the self-recovering
process. This fact could be totally unlikely in a real contagion process
on a real network and it makes the use of the SIS or SIR model extremely
unrealistic as remarked by Lee et al. \cite{lee_2019}. Even
when $\delta$ is small and the node infection process is dominant, the corresponding
epidemic dynamic is better captured by the SI model. From an application
point of view, this is possibly true over a wide range of timescales
under constrained environments where applying massive action to limit
contagion is practically infeasible.

As mentioned before, the detection of risky nodes in a network could be relevant for limiting the risk propagation effects (see, e.g., \cite{Glick, Battiston_Gatti}). Hence, centrality of a given institution as best spreader node in a contagion process has been widely explored (see, for instance, \cite{Lee}) in order to identify the most dangerous crisis epicenter. \color{black}
 The idea of best spreader node has also been studied in \cite{Mieghem2} in terms of topological centralities, which was previously investigated under the name of \lq\lq vibrational centrality\rq\rq (see, e.g., \cite{Estrada_Hatano_2010}). 
Centralities have been also used as measures to assess contagion in the interbank market.  In this framework, Dimitros and Vasileios \cite{Dimitros} recommended the use of well-established centrality measures as a way to identify the most important variables in a network. Battiston et al.  \cite{Battiston_Puliga2012} introduce DebtRank, a centrality measure that accounts for distress in one or more banks, based on the possibility of losses occurring prior to default. The concept that some banks might be too central to fail originates from this work (see, e.g., \cite{Battiston_Puliga2012}).

\section{Preliminaries}
\label{sec:prelim}
Here we use indistinctly the terms graphs and networks. Most of the
network theoretic concepts defined hereafter can be found in \cite{Estrada book}.
A \textit{graph} $\Gamma=(V,E)$ is defined by a set of $n$
nodes (vertices) $V$ and a set of $m$ edges $E=\{(u,v)|u,v\in V\}$
between the nodes. $(u, u) \in E$ is a loop starting and ending in $u$.
The \textit{degree} of a node, denoted by $k_{u}$,
is the number of edges incident to $u$ in $\Gamma$. The \textit{adjacency
	matrix} of the graph $A=\left(A_{uv}\right)_{n\times n}$ with entries
$A_{uv}=1$ if $\left(u,v\right)\in E$ or zero otherwise. We consider here simple graphs, i.e. without loops and multiedges.
The theoretical model will be developed for unweighted networks, we also recall here the definition of weighted graphs, as we consider in the paper two empirical real examples for which the network is weighted.
A weighted \textit{graph} $\Gamma'=(V,E,W)$ is a graph in which
$w_{uv}\in W$ is a positive number assigned to the corresponding
edge $\left(u,v\right)\in E$. In this case the sum of the weights
for all edges incident to a node is known as the \textit{weighted
	degree} or \textit{strength}. We consider here only undirected networks,
such that $\left(u,v\right)\in E$ implies that $\left(v,u\right)\in E$.
In this case the matrix $A$ can be expressed as $A=U\varLambda U^{T}$
where $U=\left[\vec{\psi}_{1}\cdots\vec{\psi}_{n}\right]$ is an orthogonal
matrix of the eigenvectors of $A$ and $\varLambda$ is the diagonal
matrix of eigenvalues $\lambda_{1}\geq\lambda_{2}\geq\cdots\geq\lambda_{n}$.
The entries of $\vec{\psi}_j$ are denoted by ${\psi}_{j,1},\ldots, {\psi}_{j,n}$.

\textcolor{black}{An important quantity for studying communication
	processes in networks is the communicability function
	\cite{Communicability}, defined for a pair of nodes $u$
	and $v$ as}

\textcolor{black}{
	\[
	G_{uv}=\sum_{k=0}^{\infty}\frac{\left(A^{k}\right)_{uv}}{k!}=\left(\exp\left(A\right)\right)_{uv}=\sum_{j=1}^{n}e^{\lambda_{j}}\mathbf{\mathbf{\psi}}_{j,u}\mathbf{\psi}_{j,v}.
	\]
}

\textcolor{black}{It counts the total number of walks starting at
	node $u$ and ending at node $v$, weighted in decreasing order of
	their length by a factor of $\frac{1}{k!}$. A }\textit{\textcolor{black}{walk}}\textcolor{black}{{}
	of length $k$ in $\Gamma$ is a set of nodes $i_{1},i_{2},\ldots,i_{k},i_{k+1}$
	such that for all $1\leq l\leq k$, $(i_{l},i_{l+1})\in E$. A }\textit{\textcolor{black}{closed
		walk}}\textcolor{black}{{} is a walk for which $i_{1}=i_{k+1}$. Therefore,
	$G_{uv}$ is considering shorter walks as more influential than longer
	ones. The matrix exponential is an example of a general class of matrix
	functions which are expressible as}

\textcolor{black}{
	\begin{equation}
	\big(f(A)\big)_{uv}=\sum_{k=0}^{\infty}c_{k}\big(A^{k}\big)_{uv},
	\end{equation}
	where $c_{k}$ are coefficients giving more weight to the shorter
	than to the longer walks, and making the series converge. The term
	$G_{uu}$, which counts the number of closed walks starting at the
	node $u$ giving more weight to the shorter than to the longer ones,
	is known as the subgraph centrality of the node $u$.}

We also consider here a Susceptible-Infected (SI) model over an undirected network. Each susceptible node
becomes infected at the infection rate $\gamma$ \textit{per link} times the number
of infected neighboring nodes. Let $t^*$ be the instant in which a node $i$ is infected. Node $i$ remains in this state $\forall t\geq t^*$ and does not come back susceptible.
Let us introduce a random variable $X_{i}(t)$ denoting the state of a node $i$ at time $t$ 

\begin{alignat}{1}
X_{i}\left(t\right) & =\left\{ \begin{array}{c}
1\\
0
\end{array}\begin{array}{c}
\textnormal{if \ensuremath{t\geq t^*}}\\
\textnormal{otherwise}
\end{array}\right.
\end{alignat}

Then we define

\begin{equation}
x_{i}(t)=P[X_{i}(t)=1]=\mathbb{E}[X_{i}(t)]\in[0,1],
\end{equation}

which is the probability that node $i$ is infected at time $t$. In other words, node $i$ is healthy at time $t$ with probability $1-x_i(t)$. 
For the whole network, we define the vector of probabilities:

\begin{equation}
\vec{x}(t)=[x_{1}(t),\ldots,x_{n}(t)]^{T}.
\end{equation}

\section{Model}
\label{sec:model}
Let us consider a SI model on a financial network. The nodes of a graph $\Gamma=\left(V,E\right)$ represent financial institutions and the edges connecting them represent an interaction that can transmit a ``disease'' from one institution
to another. A node can be susceptible and then get infected from a
nearest neighbor or it is infected and can transmit the infection
to other susceptible nodes. Let $\gamma$ be the infection rate and
let $x_{i}\left(t\right)$ be the probability that node $i$ get infected
at time $t$ from any infected nearest neighbor. Then,

\begin{equation}
	\dfrac{dx_{i}\left(t\right)}{dt}=\vec{\dot{x}}\left(t\right)=\gamma\left[1-x_{i}\left(t\right)\right]\sum_{j=1}^{n}A_{ij}x_{j}\left(t\right)\label{eq:SI_original}
\end{equation}
which in matrix-vector form becomes:

\begin{equation}
	\vec{\dot{x}}\left(t\right)=\gamma\left[1-\textnormal{diag}\left(\vec{x}\left(t\right)\right)\right]A\vec{x}\left(t\right),\label{nonlineq}
\end{equation}
with initial condition $\vec{x}\left(0\right)=\vec{x}_{0}.$

It is well-known that on a strongly connected network\footnote{A graph  $\Gamma = (V,E)$ is strongly connected if and only if for each pair of nodes $i,j\in V$ there is a directed walk starting at $i$ and ending at $j$, and a directed walk starting at $j$ and ending at $i$}\cite{Mei}:
\begin{enumerate}
	\item if $\vec{x}_{0}\in[0,1]^{n}$ then $\vec{x}(t)\in[0,1]^{n}$ for all
	$t>0$; 
	\item $\vec{x}(t)$ is monotonically non-decreasing in $t$; 
	\item there are two equilibrium points: $\vec{x}=\vec{0}$, i.e.
	no epidemic, and $\vec{x}=\vec{1}$ (the vector of all ones), i.e. full contagion; 
	\item the linearization of the model around the point $\vec{0}$ is given
	by
	
	\begin{equation}
	\vec{\dot{x}}(t)=\gamma A\,\vec{x}(t)\label{eq:model}
	\end{equation}
	
	and it is exponentially unstable; in fact, since, in a non-empty undirected graph, $A$ has at least one positive eigenvalue, any solution component in the direction of the corresponding eigenvector grows unboundedly as $t$ increases;
	\item each trajectory with $\vec{x}_{0}\neq \vec{0}$ converges asymptotically
	to $\vec{x}=\vec{1}$, i.e. the epidemic spreads monotonically
	to the entire network.
\end{enumerate}
In particular, the linearized problem comes from the following observation.
It can be checked that

\begin{equation}
	\dot{x}_{i}(t)=\gamma[1-x_{i}(t)]\sum_{j=1}^{n}A_{ij}x_{j}(t)\leq\gamma\sum_{j=1}^{n}A_{ij}x_{j}(t)
\end{equation}
or

\begin{equation}
	\vec{\dot{x}}(t)\leq\gamma A\,\vec{x}(t),
\end{equation}
$\forall i$ and $\forall t$. Then, we can use the linear dynamical system

\begin{equation}
\vec{\dot{x}}^{\star}\left(t\right)=\gamma A\vec{x}^{\star}\left(t\right),
\end{equation}

as an upper-bound for the original non-linear dynamical system, that has been used in the literature (see \cite{Mei}) as an approximation of the exact problem. One of its main advantages is that it can be solved analytically and its solution $\vec{x}^{\star}\left(t\right)$ can be written as:

\begin{equation}\label{eq:sol}
	\vec{x}^{\star}(t)=e^{\gamma tA}\vec{x}^{\star}_{0},
\end{equation}
which using the spectral decomposition of $A$ can be written as

\begin{equation}
	\vec{x}^{\star}(t)=\sum_{j=1}^{n} e^{\gamma t\lambda_{j}}\vec{\psi}_{j}\vec{\psi}_{j}^{T}\vec{x}^{\star}_{0}.
\end{equation}

This solution to the linearized model is affected by the following
main problems:
\begin{enumerate}
	\item $\vec{x}^{\star}(t)$ grows quickly without bound, in spite of the fact that
	$\vec{x}^{\star}(t)$ is a vector of probabilities which should not exceed the unit; 
	\item $\vec{x}^{\star}(t)$ is an accurate solution to the nonlinear SI problem
	only if $t\to0$ and $\vec{x}^{\star}_{0}\to0$.
\end{enumerate}

The mathematical properties of the linear dynamical system \ref{eq:model} as well as of the solution \ref{eq:sol} have been extensively studied by Mugnolo in \cite{Mugnolo}. We direct the reader to this reference for the details.

Hereafter we will follow the recent work of Lee et al. \cite{lee_2019}, who proposed the following change of variable to avoid the aforementioned problems with the solution of the linearized SI model:

\begin{equation}
	y_{i}\left(t\right)\coloneqq-\log\left(1-x_{i}\left(t\right)\right),
\end{equation}
which is an increasing convex function. Then, as $1-x_{i}(t)$ is the
probability that node $i$ is not infected at a given time $t$, the new
variable $y_{i}\left(t\right)$ can be interpreted as the information content of the node $i$ or surprise of not being infected (see, e.g., \cite{cover_2006}). According to \cite{lee_2019}, the SI model \cref{eq:SI_original} can be now written as

\begin{equation}
	\dfrac{dy_{i}\left(t\right)}{dt}=\dot{y}_{i}\left(t\right)=\gamma\sum_{j=1}^{n}A_{ij}x_{i}\left(t\right)
	\label{eq:310}
\end{equation}
or

\begin{equation}
	\vec{\dot{y}}\left(t\right)=\gamma A\,\vec{x}(t).
\end{equation}

The approximate solution to the SI model provided by \cite{lee_2019} is then given by

\begin{equation}
	\vec{x}(t)=\vec{1}-e^{-\vec{y}\left(t\right)},
\end{equation}
where $e^{-\vec{y}\left(t\right)}$ is the vector in which the $i$th
entry is $e^{-y_{i}\left(t\right)}$ and

\begin{align}
	 \vec{y}\left(t\right)=  \qquad & 	e^{\gamma tA\textnormal{diag}\left(\vec{1}-\vec{x}_{0}\right)}\left[-\log\left(1-\vec{x}_{0}\right)\right] \nonumber \\
       & +\sum_{j=0}^{\infty}\frac{\left(\gamma t\right)^{j+1}}{\left(j+1\right)!}\left[A{\rm diag}\big(\vec{1}-\vec{x}_{0}\big)\right]^{j}A\left(\vec{x}_{0}+\big(\vec{1}-\vec{x}_{0}\big)\log\big(\vec{1}-\vec{x}_{0}\big)\right).
\end{align}

As stressed by \cite{lee_2019}, the interesting case of the dynamics is when $\vec{x}_{0}<\vec{1}$, in which case the solution simplifies to

\begin{equation}
	\vec{y}(t)=\vec{y}_{0}+\left[e^{\gamma tA{\rm diag}\left(\vec{1}-\vec{x}_{0}\right)}- I \right]\cdot\textnormal{diag}\left(\vec{1}-\vec{x}_{0}\right)^{-1}\vec{x}_{0}.
\end{equation}

Now, we can make the further assumption that the initial probabilities of being infected are equal for every node, i.e. that at the beginning every node has the same probability $\beta$ to be infected and to be the one from which the epidemic starts. This means that we are asking for
\begin{equation}
	x_{0i}=\beta=\frac{c}{n},\ \forall i=1,\ldots,n
\end{equation}
for some scalar constant $c$. In this case ${\rm diag}\left(\vec{1}-\vec{x}_{0}\right)=\left(1-\frac{c}{n}\right) I =(1-\beta) I$.
If we set $\alpha=1-\beta$, the approximate solution of the SI on the network becomes:
\begin{equation}
	\vec{y}(t)=\vec{y}_{0}+\frac{1-\alpha}{\alpha}\left[e^{\alpha\gamma tA}-I\right]\vec{1}.
\end{equation}
and since $\vec{y}_{0}=\left(-\log\alpha\right)\vec{1}$,

\begin{equation}
	\vec{y}(t)=\left(\frac{1}{\alpha}-1\right)e^{\alpha\gamma tA\,}\vec{1}-\left(\log\alpha+\frac{1-\alpha}{\alpha}\right)\vec{1}.\label{solution}
\end{equation}

The component $(e^{\alpha \gamma t A}\vec{1})_{i}$ is called total communicability of node $i$ and it will be denoted by ${\mathscr{R}}_{i}$. Hence, component-wise we have:%

\begin{equation}
	y_{i}\left(t\right)=\left(\frac{1}{\alpha}-1\right)\mathscr{R}_{i}-\left(\log\alpha+\frac{1-\alpha}{\alpha}\right).
\end{equation}

Keeping in mind that $-\log\alpha=y_{i}(0)$ and $\alpha=1-\beta$, we can write the previous equation also as

\begin{equation}
	\Delta y_{i}\left(t\right)=y_{i}\left(t\right)-y_{i}\left(0\right)=\frac{\beta}{\alpha}(\mathscr{R}_{i}-1),
\end{equation}

which means that $\mathscr{R}_{i}-1$ at time $t$ is proportional to the variation in the information content of node $i$ from time 0 to time $t$. Finally, the probability of node $i$ of being infected at time $t$ can be expressed in terms of ${\mathscr R}_{i}$ as
\begin{equation}
x_{i}(t)=1-(1-\beta)e^{-\frac{\beta}{1-\beta}({\mathscr R}_{i}-1)}.\label{probability}
\end{equation}

When the parameter $\beta$ is fixed, the number of infected nodes depends only on the term $e^{\alpha \gamma t A}\vec{1}$ and then on the total communicabilities $\mathscr{R}_{i}$. It is worth noticing that the probability given by \ref{probability} for a node $i$ represents an upper bound for the exact solution of the SI model. Hence, in this way we do not underestimate the contagion probabilities. Let us consider, for instance, the time evolution of an infection propagation on an Erd\H{o}s-Rényi network with 100 nodes and edge density $\delta=0.1$. Results are illustrated in \cref{Figure_1} for two different values of the infectivity rate, $\gamma=0.001$ (left) and $\gamma=0.002$ (right). The dashed red lines represent the mean probability that a node is infected at time $t$ as given by equation \ref{probability}. The solid blue lines represent the same probability as given by the exact solution of the Kermack-McKendrick SI model with the same mean degree. In both plots, the initial probability is $\beta=0.01$.

\begin{figure}[H]
	\centering
	\subfloat[]{\includegraphics[width=0.45\textwidth]{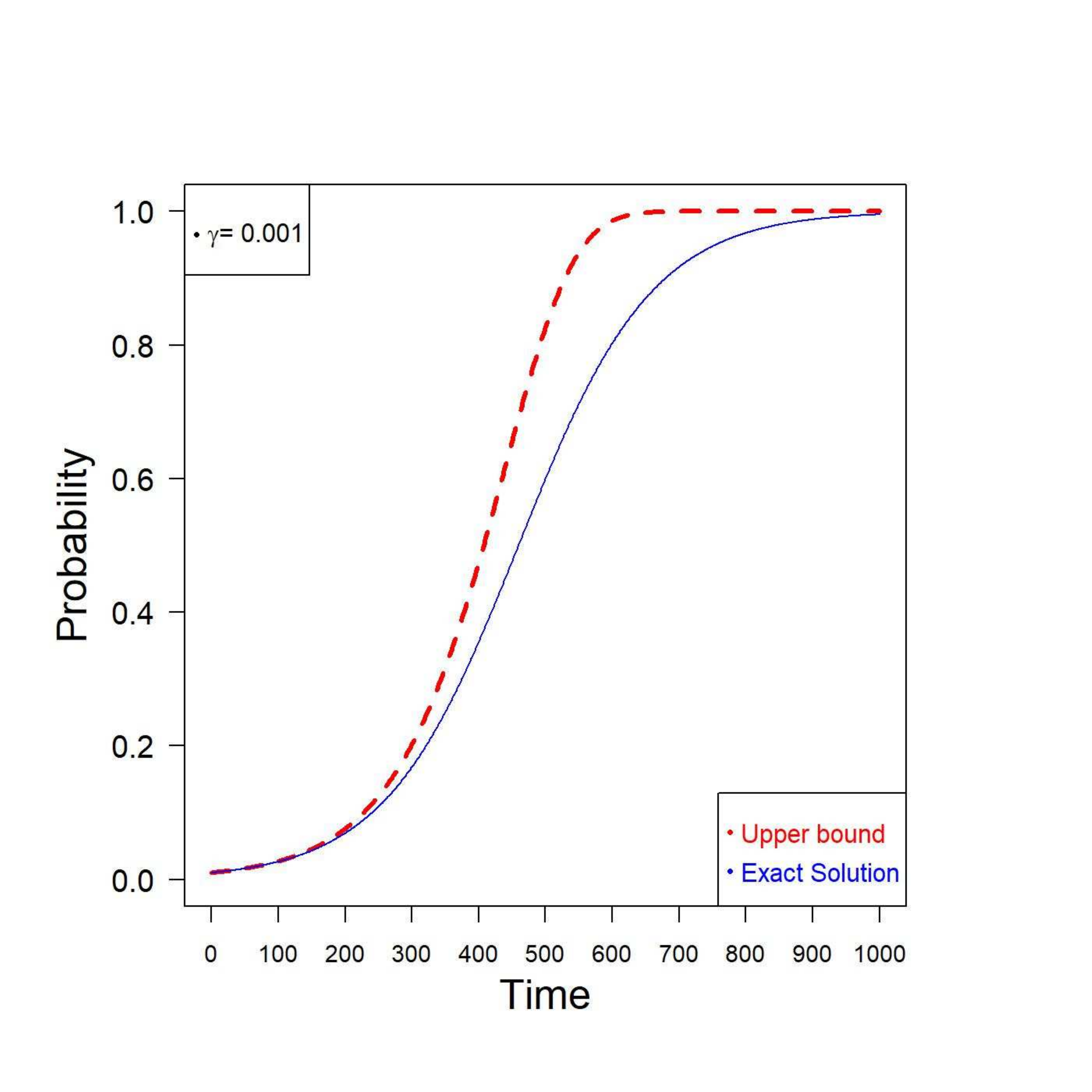}
		
	}\subfloat[]{\includegraphics[width=0.45\textwidth]{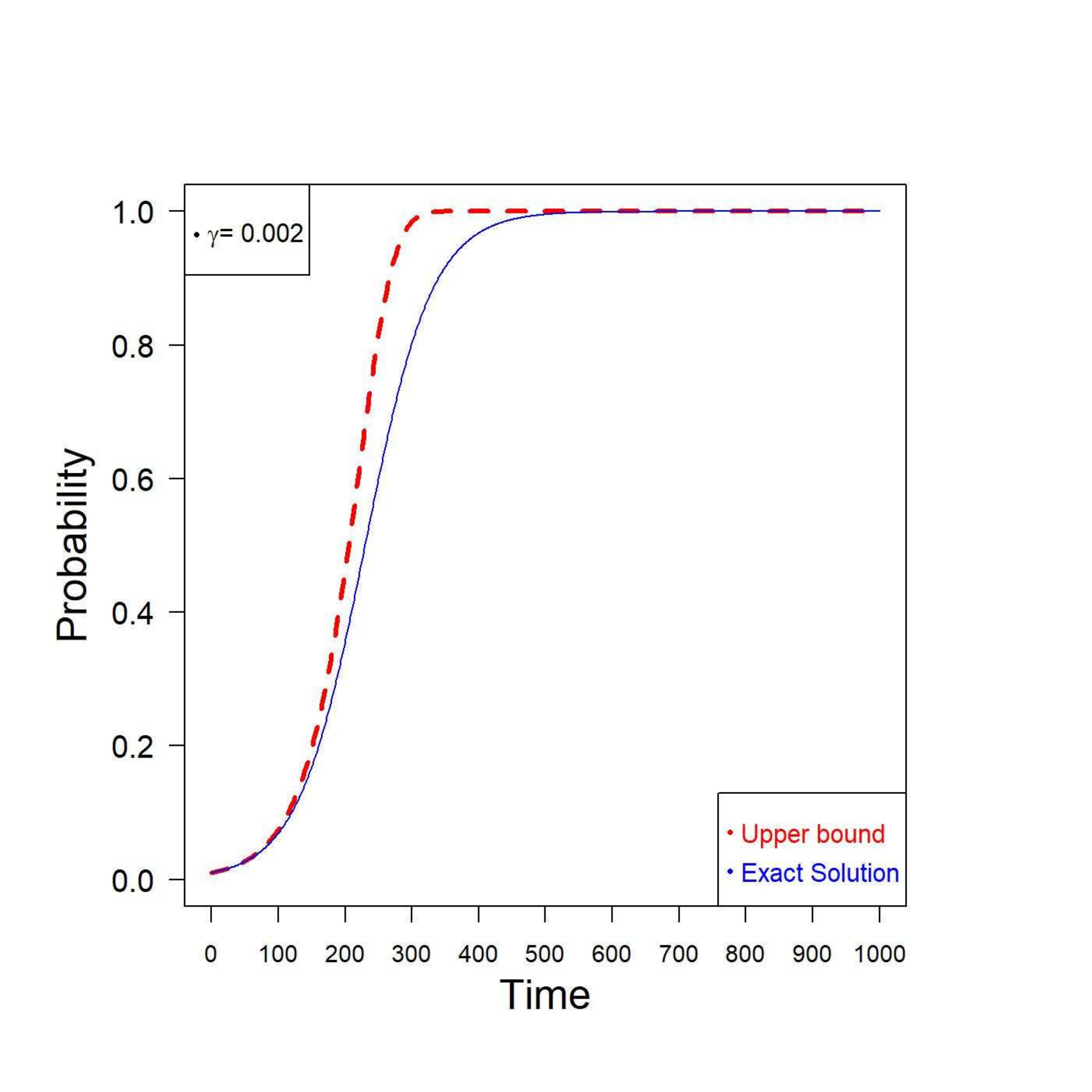}
		
	}
	\caption{Simulation of the progression of a SI epidemics on an Erd\H{o}s-Rényi
	network with 100 nodes and edge density $\delta=0.1$. The parameters used in the model are: $\beta=0.01$ and $\gamma=0.001$ (left) and $\gamma=0.002$ (right). Dashed (red) lines represent the upper bound given by \ref{probability}; solid (blue) lines represent the value of the same probability in a Kermack-McKendrick SI model with the same mean degree $\bar{k}=(n-1)\delta$.}
\label{Figure_1}
\end{figure}

\section{Risk-dependent centrality}
\label{sec:RDcent}
Let us designate $\zeta=\alpha\gamma t$, which determines the level
of risk to which the whole network is submitted at time $t$. For instance, for $\gamma=0$, i.e. $\zeta=0$, there is no risk of infection on the network as a node cannot transmit the disease to a nearest neighbor. This situation
corresponds to the case of isolated nodes (no edges). When $\zeta\rightarrow\infty$
the risk of infection is very high due to the fact that for a fixed
value of $c$ the infectivity is infinite. Therefore, we call $\mathscr{R}_{i}=\left(e^{\zeta A}\vec{1}\right)_{i}$
the risk-dependent centrality of the node $i$. That is, the values
of $\mathscr{R}_{i}$ reflects how central a node is in ``developing''
the epidemics on the network. As the networks considered are
undirected, this centrality accounts
for both the facility with which the node gets infected as well as
the propensity of this node to infect other nodes. The index $\mathscr{R}_{i}$
can be expressed as

\begin{equation}
	\mathscr{R}_{i}=\left[ \left( I+\zeta A+\zeta^{2}\frac{A^{2}}{2!}+\zeta^{3}\frac{A^{3}}{3!}+\cdots \right)\vec{1} \right]_{i},
\end{equation}

which indicates that it counts the number of walks of different lengths,
that have started at the corresponding node, weighted by a factor $\dfrac{\zeta^{k}}{k!}$.
It is straightforward to realize from the definition of the risk-dependent
centrality that it can be split into two contributions. That is, $\mathscr{R}_{i}$ is composed by a weighted sum of all closed walks that start and end at $i$, $\left(e^{\zeta A}\right)_{ii}$ and by the weighted sum of walks that start at the node $i$ and end elsewhere, $\sum_{j\neq i}\left(e^{\zeta A}\right)_{ij}$

\begin{equation}
	\mathscr{R}_{i}=\left(e^{\zeta A}\right)_{ii}+\sum_{j\neq i}\left(e^{\zeta A}\right)_{ij}:=\mathscr{C}_{i}+\mathscr{T}_{i},
\end{equation}
where the first term in the right-hand side represents the circulability
of the disease around a given node and the second one represents the
transmissibility of the disease from the given node to any other in
the network. The circulability is very important because it accounts
for the ways the disease has to become endemic. For instance, a large
circulability for a node $i$ implies that the disease can infect
its nearest neighbors and will keep coming back to $i$ over and over
again in a circular way. We start now by proving some results about
these risk-dependent centralities. The following theorem is a special
case of results found, for instance, in \cite{BK15}.
\begin{theorem}
	\label{thm1} The node ranking given by the risk dependent centralities ${\mathscr{R}}_i(\zeta)$, with $i=1,\ldots ,n$,
	reduces to the ranking given by the degree $k_i$ in the limit as the risk $\zeta \to 0$, and to the ranking given
	by eigenvector centrality as $\zeta \to \infty$.
\end{theorem}

\begin{proof}
	We begin by observing that the ranking of nodes, in terms of their
	risk-dependent centrality, is unaffected if all the centralities ${\mathscr{R}}_{i}$
	are shifted and rescaled by the same amount. That is, the same ranking
	is obtained using either ${\mathscr{R}}_{i}$ or the equivalent measure
\begin{displaymath}	
	\hat{\mathscr{R}}_{i}=\frac{{\mathscr{R}}_{i}-1}{\zeta},
\end{displaymath}
	where $\zeta>0$. Now, we have 
	\begin{equation}
     \hat{\mathscr{R}}_{i}=\left[\left(A+\frac{\zeta}{2!}A^{2}+\cdots\right){\vec 1}\right]_{i}=k_{i}+\frac{\zeta}{2!}(A^{2}{\vec 1})_{i}+O(\zeta^{2}).\label{eq_Thm1_1}
	\end{equation}
	Hence, in the limit of $\zeta\to0$, the ranking given by ${\mathscr{R}}_{i}$
	is identical to degree ranking.
	
	To study the limit for $\zeta$ large we write 
	\begin{equation}
		{\mathscr{R}}_{i}=\left[e^{\zeta A}{\vec 1}\right]_{i}=\sum_{k=1}^{n}e^{\zeta\lambda_{k}}(\psi_{k}^{T}{\vec 1})\psi_{k,i}=e^{\zeta\lambda_{1}}(\psi_{1}^{T}{\vec 1})\psi_{1,i}+\sum_{k=2}^{n}e^{\zeta\lambda_{k}}(\psi_{k}^{T}{\vec 1})\psi_{k,i}.\label{eq_Thm1_2}
	\end{equation}
	We note again that for ranking purposes we can use the equivalent
	measure obtained by dividing all risk-dependent centralities by the
	same quantity, $e^{\zeta\lambda_{1}}(\psi_{1}^{T}{\vec 1})$, which
	is strictly positive. That is, we can use 
	\begin{equation}
		\tilde{\mathscr{R}}_{i}=\psi_{1,i}+\frac{1}{\psi_{1}^{T}{\vec 1}}\sum_{k=2}^{n}e^{\zeta(\lambda_{k}-\lambda_{1})}(\psi_{k}^{T}{\vec 1})\psi_{k,i}.\label{eq_Thm1_3}
	\end{equation}
	Since the network is connected, the Perron--Frobenius Theorem insures
	that $\lambda_{1}>\lambda_{2}\ge\cdots\ge\lambda_{n}$. Hence, each
	term $e^{\zeta(\lambda_{k}-\lambda_{1})}$ for $k=2,\ldots,n$ vanishes
	in the limit as $\zeta\to\infty$, and we see from \cref{eq_Thm1_3}
	that the risk-dependent centrality measure gives the same ranking
	as eigenvector centrality for $\zeta$ large.
\end{proof}
It is interesting to observe that the risk-dependent centrality of
every node also depends on the (strictly positive) quantity 
\[
\psi_{1}^{T}{\vec 1}=\sum_{j=1}^{n}\psi_{1,j},
\]
see Equation \cref{eq_Thm1_2}. The larger this quantity is, the higher is
the risk-dependent centrality of each node. Assuming that the dominant
eigenvector is normalized so as to have Euclidean norm equal to 1,
it is well known that this quantity is always between 1 and $\sqrt{n}$.
The value 1 is never attained for a connected graph. It can only be
approached in the limit as all the eigenvector centrality is concentrated
on one node, say node $i$, where it takes values arbitrarily close
to 1, with the values $\psi_{1,j}$ for all $j\ne i$ taking arbitrarily
small values. An example of this would be the star graph\footnote{We recall that the star graph $S_n$ consists of $n-1$ nodes $v_1 , \ldots , v_{n-1}$, each attached to a central node $v_n$ by an edge.} $S_{n}$ for $n\to\infty$. The maximum value is attained in the case where
all nodes have the same eigenvector centrality: $\psi_{1,1}=\psi_{1,2}=\cdots=\psi_{1,n}$
(i.e., in the case of regular graphs).

Let us return to the decomposition ${\mathscr{R}}_{i}={\mathscr{C}}_{i}+{\mathscr{T}}_{i}$
of the risk-dependent centrality of a node into its two components,
circulability and transmissibility. Similar considerations apply to
these quantities. We summarize them in the following result.
\begin{theorem}
	\label{thm2} The node rankings given by the degree $k_{i}$ and the eigenvector centrality
	are obtained as limiting cases of the risk-dependent circulability ${\mathscr{C}}_{i}(\zeta)$ as the external level of risk $\zeta$ decreases to zero or increases to infinity, 
	respectively. The same is true for the risk dependent transmissibility ${\mathscr{T}}_{i}(\zeta)$.
\end{theorem}

\begin{proof}
	The proof for the circulability is a straightforward adaptation of
	that for the total communicability; see also \cite{BK15}.
	
	We give the details for the transmissibility, which has not been analyzed
	before. We have for $i\ne j$ that 
	\[
	\left(e^{\zeta A}\right)_{ij}=\zeta A_{ij}+\frac{\zeta^{2}}{2!}w_{i,j}^{(2)}+O(\zeta^{3}),
	\]
	where $w_{i,j}^{(2)}$ denotes the number of walks of length two between
	node $i$ and node $j$. Dividing by $\zeta>0$, summing over all
	$j\ne i$ and taking the limit as $\zeta\to0$, we find 
	\[
	\zeta^{-1}{\mathscr{T}}_{i}=\zeta^{-1}\sum_{j\ne i}\left(e^{\zeta A}\right)_{ij}\to\sum_{j\ne i}A_{ij}=k_{i},
	\]
	where we have used the fact that $A_{ii}=0$, for all $i$. Hence,
	transmissibility is equivalent to node degree in the small $\zeta$
	limit. For the large $\zeta$ limit we write 
	\[
	{\mathscr{T}}_{i}=\sum_{j\ne i}\sum_{k=1}^{n}e^{\zeta\lambda_{k}}\psi_{k,i}\psi_{k,j}=e^{\zeta\lambda_{1}}\psi_{1,i}\sum_{j\ne i}\psi_{1,j}+\sum_{k=2}^{n}e^{\zeta\lambda_{k}}\left[\sum_{j\ne i}\psi_{k,i}\psi_{k,j}\right].
	\]
	Dividing by the positive constant $e^{\zeta\lambda_{1}}\sum_{j\ne i}\psi_{1,j}$
	and taking the limit as $\zeta\to\infty$, the second part of the
	right-hand side vanishes and we obtain again the eigenvector centrality
	$\psi_{1,i}$ of node $i$.
\end{proof}

\begin{remark}

A natural question is how rapidly the degree (for $\zeta \to 0$) and eigenvector 
(for $\zeta \to \infty$) centrality limits are approached if the number of nodes $n$
in the network goes to infinity. From the Taylor expansions (see for example
equation \ref{eq_Thm1_1}) we see that the degree limit is reached more slowly if the row sums of $A^2$ grow as $n\to \infty$. In this case, as $n$ increases $\zeta$ must be taken smaller and smaller before the ranking reduces to the one given by the degree.
On the other hand, if the network grows in such a way that the maximum degree of any node remains uniformly bounded, then the rate of convergence is independent of the number $n$ of nodes, at least asymptotically. 

The rate of convergence to the eigenvector centrality ranking is largely determined
by the spectral gap, $\lambda_1 - \lambda_2$. If the gap remains bounded below by a
positive constant as $n\to \infty$, the value of $\zeta$ necessary to reach the eigenvector centrality limit is easily seen to grow at most like $O(\ln n)$, and in practice the rate of convergence is scarcely affected by the size of the network. If, on the other hand, the gap closes as $n\to \infty$, then the rate of convergence to the eigenvector centrality will become arbitrarily slow. The faster the gap closes for $n\to \infty$, the more rapidly the rate of convergence deteriorates.

\end{remark}

We conclude this section with some comments on the measures
${\mathscr{R}}_{i}$, ${\mathscr{C}}_{i}$ and ${\mathscr{T}}_{i}$.
While they all display the same limiting behavior and provide identical
rankings in the small and large $\zeta$ limits, they provide different
insights on the network structure (and therefore on node risk). For
instance, it is well known that subgraph centrality (which is the
same as circulability, see \cite{Estrada05,Estrada book}) can discriminate
between the nodes of certain regular graphs, that is, graphs in which
all the nodes have the same degree. The same holds for transmissibility.
Total communicability, on the other hand, is unable to discriminate
between the nodes of regular graphs (and neither are degree and eigenvector
centrality, of course). These measures are also different from a computational
viewpoint. One advantage of the risk centrality based on total
communicability is that it only requires the computation of the action
of the matrix exponential $e^{\zeta A}$ on the vector ${\vec 1}$.
The entries of the resulting vector can be computed efficiently without
having to compute any entry of $e^{\zeta A}$, see \cite{BK13}. Modern
Krylov-type iterative methods (like those based on the Lanczos or
Arnoldi process) can handle huge networks (with many millions of nodes)
without any difficulty. In contrast, the computation of the circulability
requires the explicit computation of the diagonal entries of $e^{\zeta A}$
(the node transmissibility is then easily obtained by subtracting
the circulability from the total communicability). Although there
are techniques that can handle fairly large graphs (see \cite{BB10}),
these calculations are much more expensive than those for the total
communicability. This limits the size of the networks that they can
be applied to. However, for most financial networks the computation
of the circulability is still feasible. 

A final consideration regards the values assumed by the external risk parameter $\zeta$. Although, in principle, it can vary between $0$ and infinity, for the purposes of most of the applications that follow, it may be sufficient to vary $\zeta$ between $0$ and $1$. The rationale for using the interval $[0,1]$ relies on the fact that, at $\zeta = 1$, the rankings given by ${\mathscr R}_i$ are already stabilizing around those provided by eigenvector centrality and therefore no more interlacings between rankings are possible. As we will show, we typically observe a single point of interlacement and it usually occurs before reaching the value $\zeta = 1$. Furthermore, this choice is equivalent to fix $t=1$ in the epidemic model solution \ref{solution}, and, already as $\zeta$ approaches $1$, all the probabilities involved in that model become completely negligible or equal to $1$.

\section{Risk-dependent centrality on a random network}
\label{sec:Numerical}
For the analysis of real-world (financial and economic) networks it is necessary to investigate how informative the results obtained are with respect to the real system under analysis. This significance is typically addressed by comparing to those properties obtained from network null models. As such null models we consider here Erd\H{o}s-Rényi (ER) random networks $\Gamma_{ER}\left(n,p\right)$ with $n$ nodes
and wiring probability $p$ (see \cite{ER60,ER59}), for which, in this section, we provide
a series of analytical results.
We start by generating a family of simulated ER graphs and discarding simulations for which the obtained graph is not connected.

In particular, we aim at testing how the external risk $\zeta$ and the probability $p$, and hence the graph density. $\delta$, affect the results. For this purpose, we generate $1000$ graphs $\Gamma_{ER}(n;p)$ with $n=100$ at different values of $p$. For each graph, we compute the main measures for alternative values of $\zeta$.
Firstly, we report in \cref{fig:fig1} the behavior of risk-dependent centrality $\mathscr{R}_{i}$, circulability $\mathscr{C}_{i}$ and transmissibility $\mathscr{T}_{i}$ as functions of the density, assuming a fixed high level of external risk, $\zeta=1$. Since the values of $\mathscr{R}_{i}$ are significantly
increasing when the density of the graph increases, we display, in
\cref{fig:fig1}(a), the distributions of the ratio between
the risk-dependent centrality of each node $\mathscr{R}_{i}$ and
its average value ${\mathbb E}(\mathscr{R}_{i})$.

\noindent As might be expected, the centralities of nodes tend to be similar when $\delta\to 1$ and we move towards the complete graph, i.e. we observe a lower variability of the distribution of the ratios. Similar behaviors are also observed for $\mathscr{C}_{i}$
and $\mathscr{T}_{i}$, with an higher volatility for the circulability
(see \cref{fig:fig1}(b) and \cref{fig:fig1}(c)).

In \cref{fig:fig1}(d), we show the distributions of the incidence
of the circulability $\mathscr{C}_{i}$ on the risk-dependent centrality
$\mathscr{R}_{i}$, that is the distribution of the ratio  $\frac{\mathscr{C}_{i}}{\mathscr{R}_{i}}$ again as a function of the density $\delta$. When $\zeta=1$, for all the graphs analyzed, the average value is around $\frac{1}{n}$, implying that the transmissibility has an average incidence of $\frac{n-1}{n}$ on $\mathscr{R}_{i}$.
It is noteworthy to look at the variability of the distributions. When the density is extremely low, i.e. we refer to a very sparse graph, the heterogeneity of the nodes degree affects the ratio $\frac{\mathscr{C}_{i}}{\mathscr{R}_{i}}$. For instance,
when $\delta=0.1$, the circulability of a node ranges
approximately from 0.15\% to 2.5\% of the risk-dependent centrality
for the same node. 
A lower variability is observed for higher densities. For instance, for $\delta=0.5$, the ratio $\frac{\mathscr{C}_{i}}{\mathscr{R}_{i}}$ varies between 0.6\% and 1.3\%. For $\delta=0.95$, we observe a ratio between 0.9\% and 1.15\%.

\begin{figure}[H]
	\subfloat[]{\includegraphics[scale=0.07]{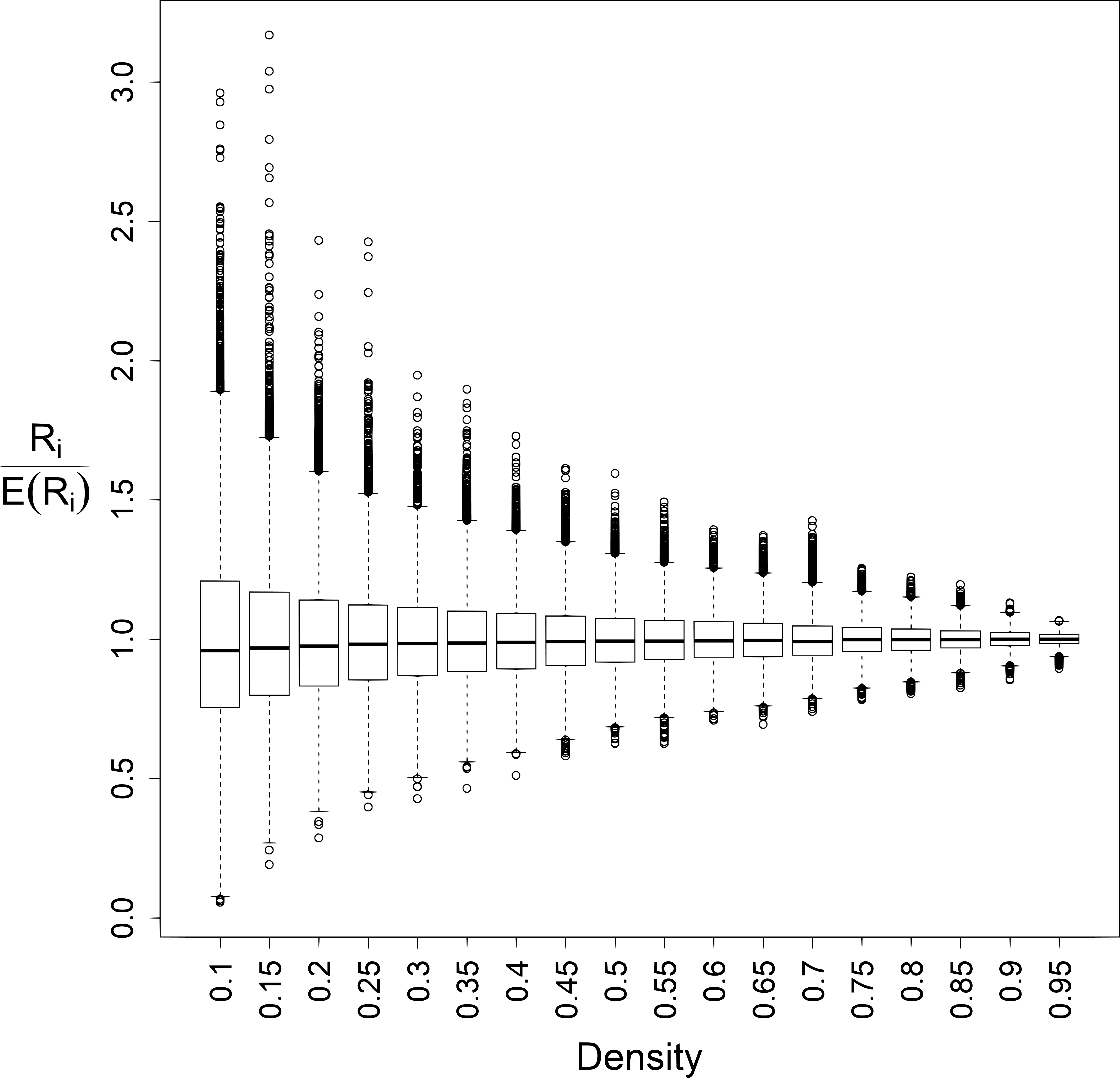}}
	\subfloat[]{\includegraphics[scale=0.07]{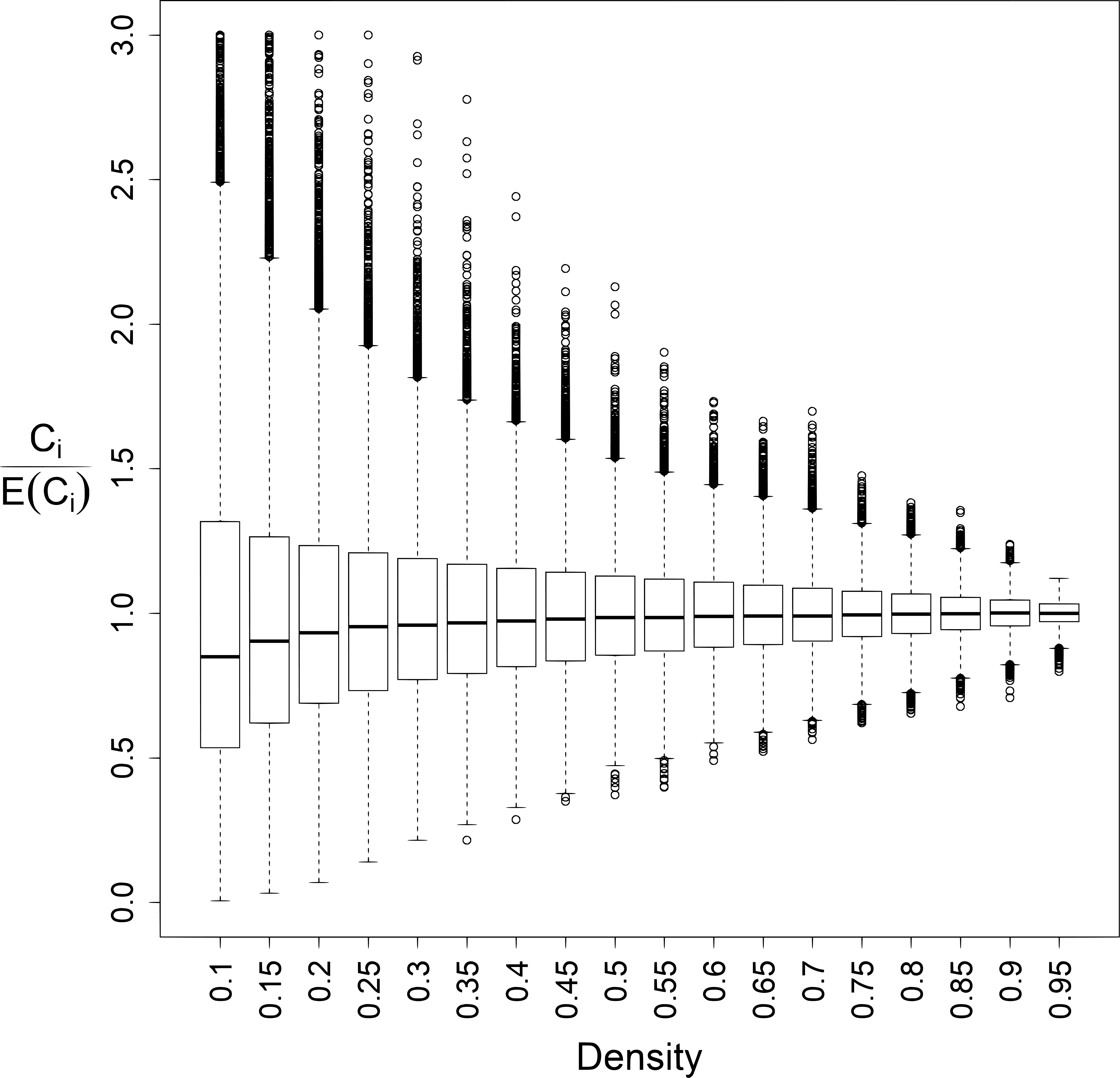}}
	
	\subfloat[]{\includegraphics[scale=0.07]{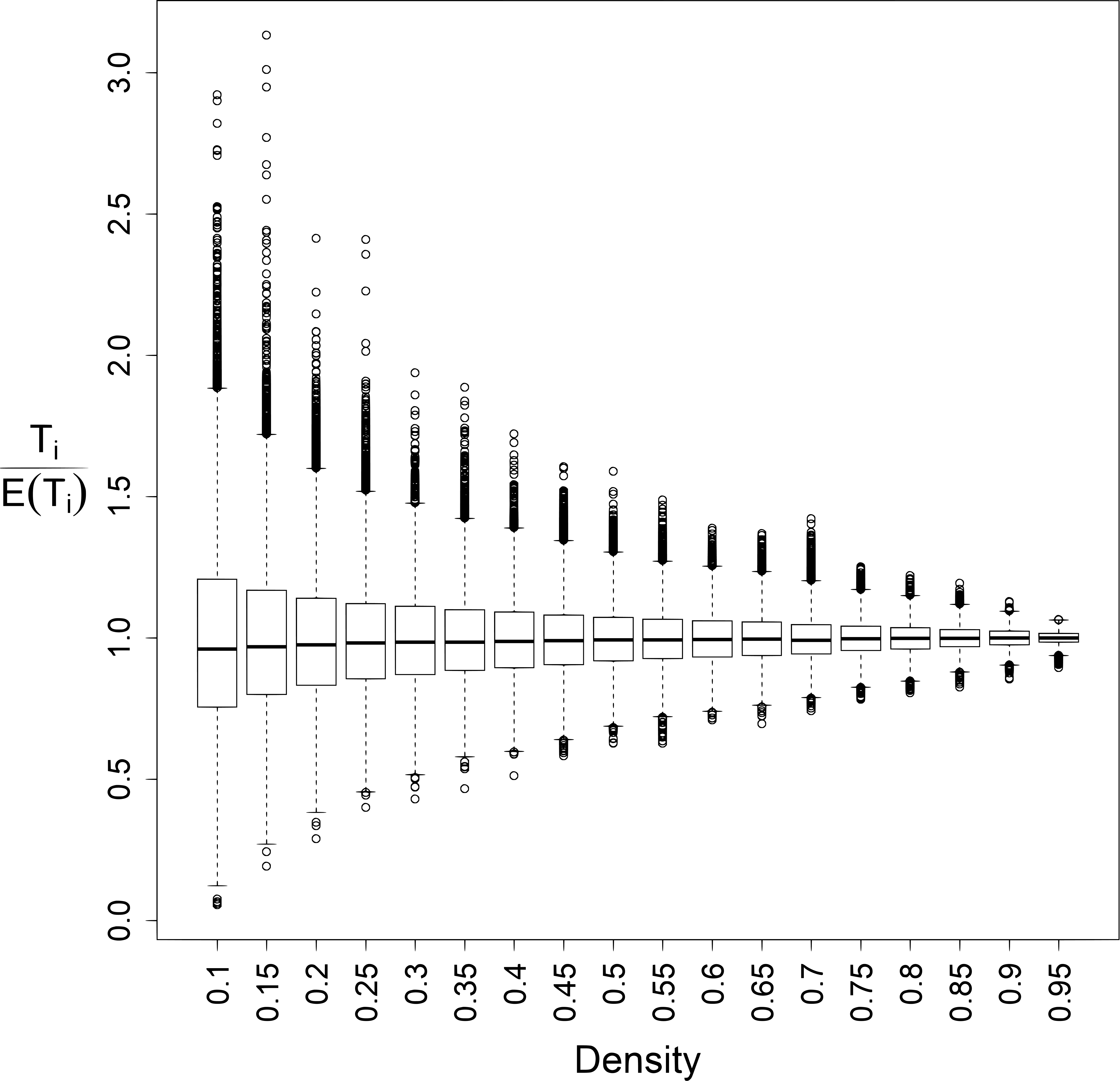}}
    \subfloat[]{\includegraphics[scale=0.07]{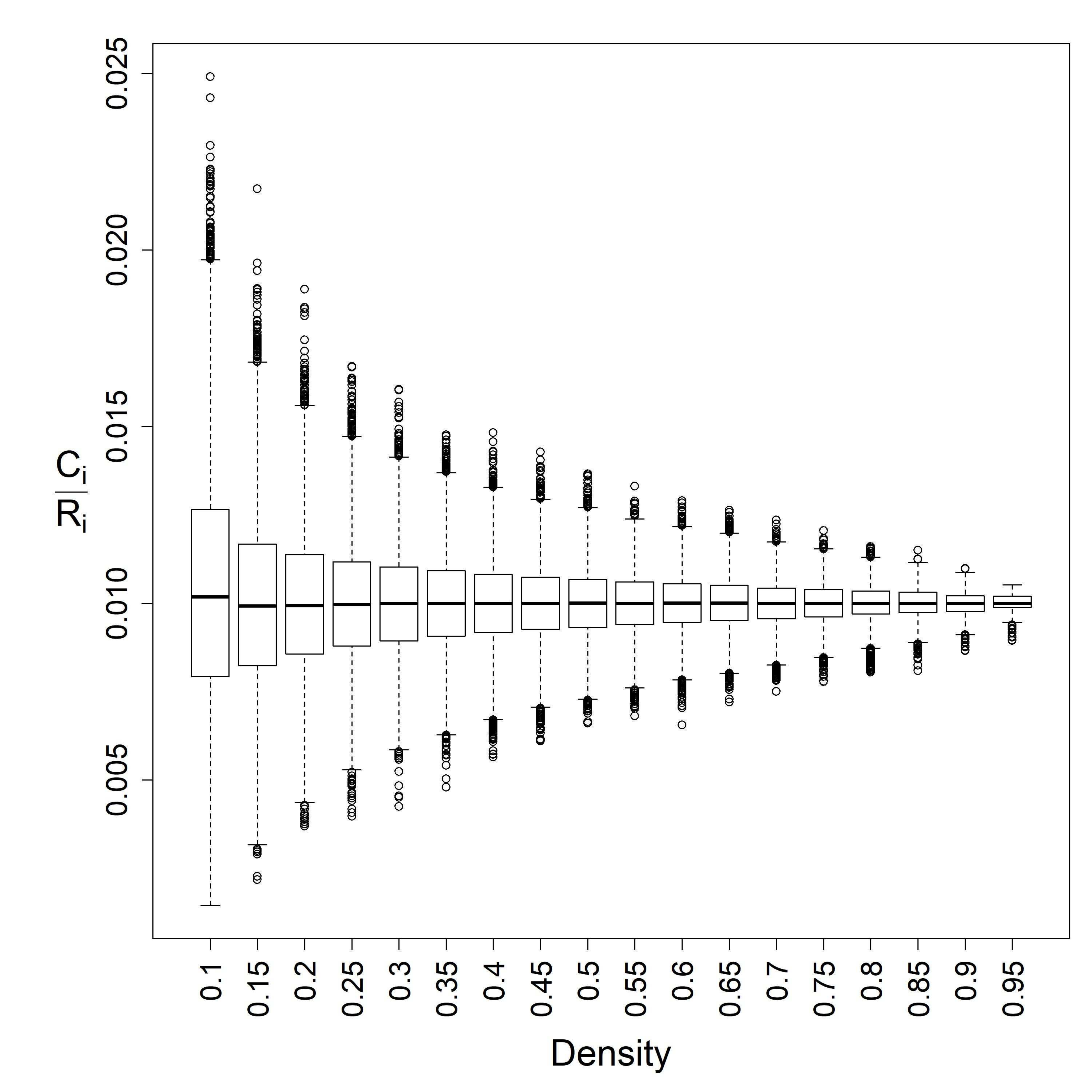}}

\caption{Figure a) displays the distributions of the ratios between the risk-dependent centrality of each node $\mathscr{R}_{i}$ and the average risk-dependent
centrality ${\mathbb E}\left({\mathscr{R}}_{i}\right)$, computed assuming $\zeta=1$.
Figure b) and c) display the analogous distributions for circulability and transmissibility. Figure d) shows the distributions of the ratios
between the circulability ${\mathscr{C}}_{i}$ and the risk-dependent
centrality of each node ${\mathscr{R}}_{i}$, computed assuming $\zeta=1$.
All Figures are based on 1000 randomly generated ER networks $\Gamma_{ER}(n;p)$ with a density varying between 0.1 and 0.9.}
\label{fig:fig1}
\end{figure}

In \cref{fig:fig2}, we show the corresponding behaviors of risk-dependent centrality $\mathscr{R}_{i}$, circulability $\mathscr{C}_{i}$ and transmissibility $\mathscr{T}_{i}$ as functions of the density, but assuming a fixed low level of external risk, $\zeta=0.1$. Again all Figures are based on 1000 randomly generated ER networks $\Gamma_{ER}(n;p)$ with $\delta$ varying between 0.1 and 0.9.

Focusing on the risk-dependent centrality ratio $\frac{\mathscr{R}_{i}}{{\mathbb E}\left({\mathscr{R}}_{i}\right)}$, we observe that the standard deviation between nodes is lower in the low-risk framework ($\zeta=0.1$) than in the high-risk one ($\zeta=1$).
For instance, when the density is equal to $0.1$, the standard deviation
of the ratio moves from $0.20$ for $\zeta=0.1$ to $0.37$ for $\zeta=1$. At a phenomenological level, this behavior can be justified by the fact that differences between nodes tend to be enhanced when the network is highly risk-exposed. \\
Furthermore, the pattern of $\frac{\mathscr{C}_{i}}{{\mathbb E}\left(\mathscr{C}_{i}\right)}$ for $\zeta=0.1$ is very peculiar. In this case, when the network is very sparse, nodes show a similar circulability, while
higher differences are observed when the density is around 0.5. 



\begin{figure}[H]
	\subfloat[]{\includegraphics[scale=0.07]{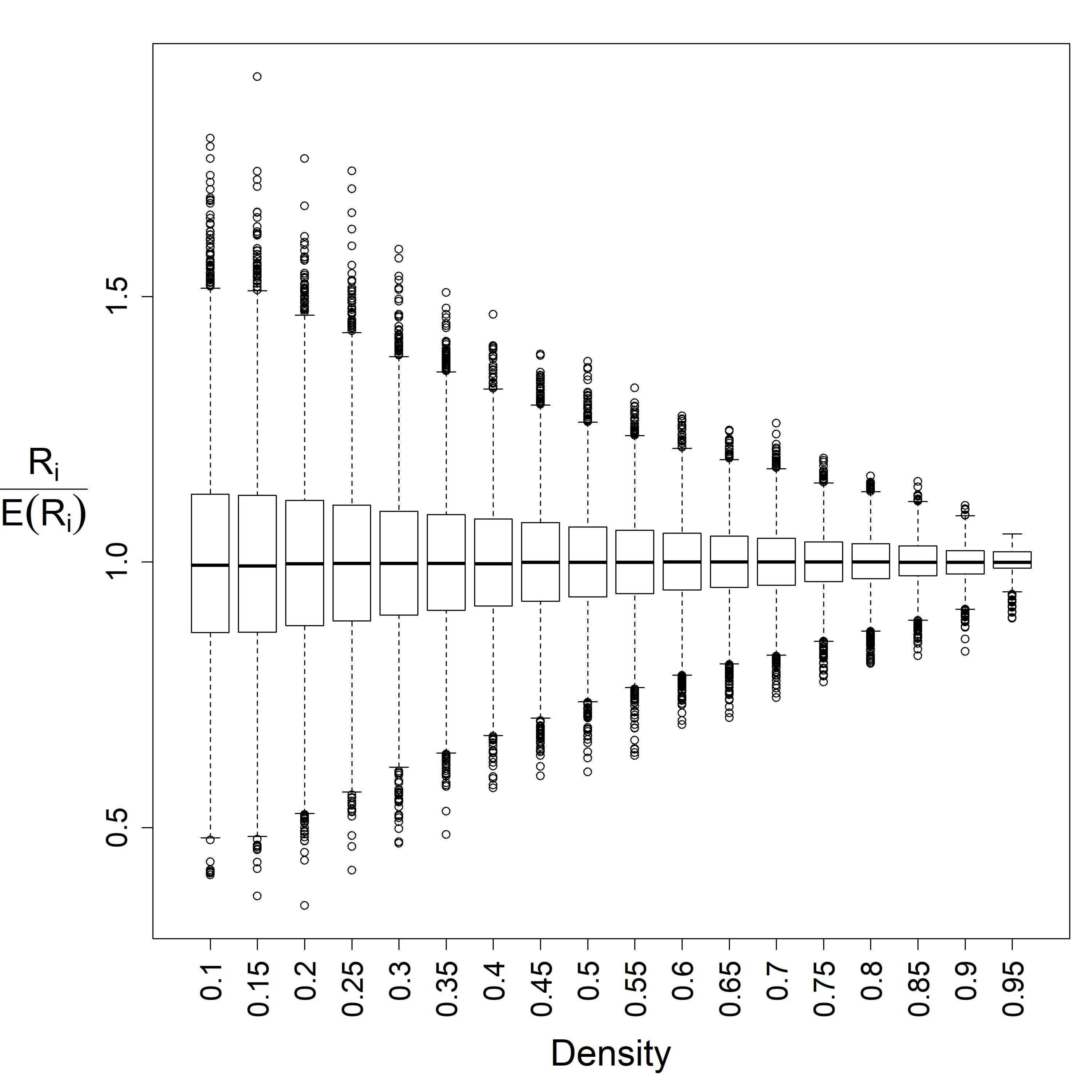}
		
	}\subfloat[]{\includegraphics[scale=0.07]{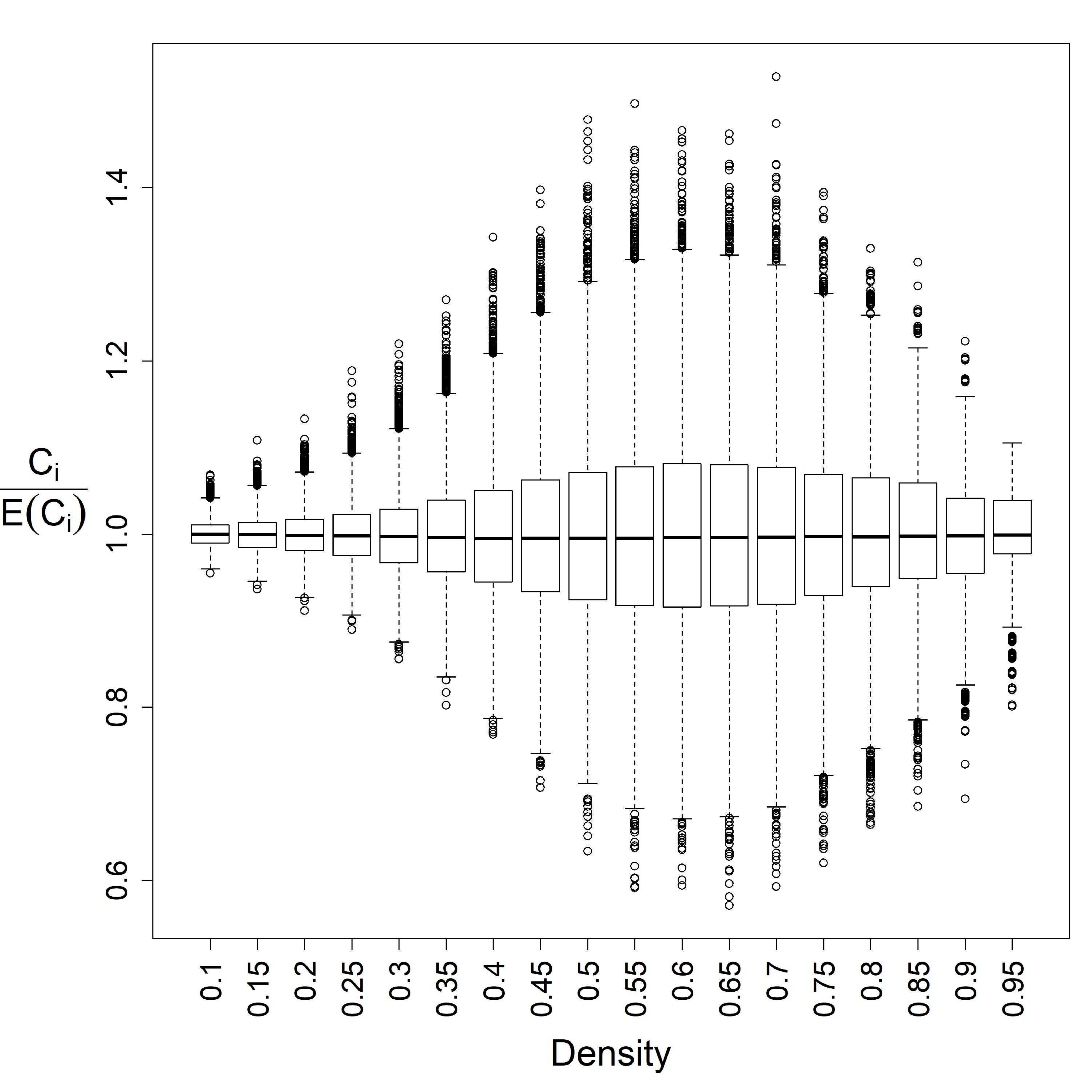}
		
	}
	
	\subfloat[]{\includegraphics[scale=0.07]{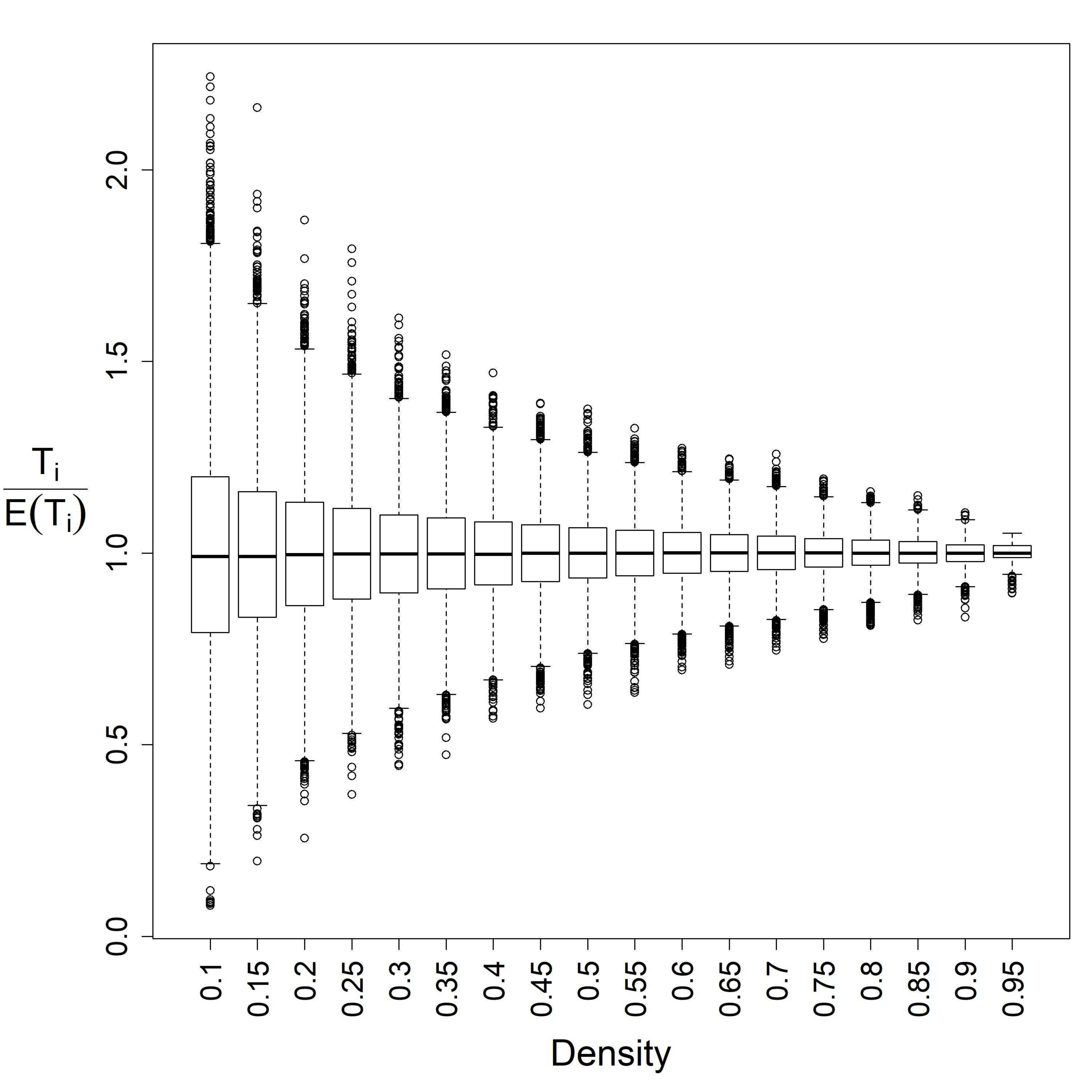}
		
	}\subfloat[]{\includegraphics[scale=0.07]{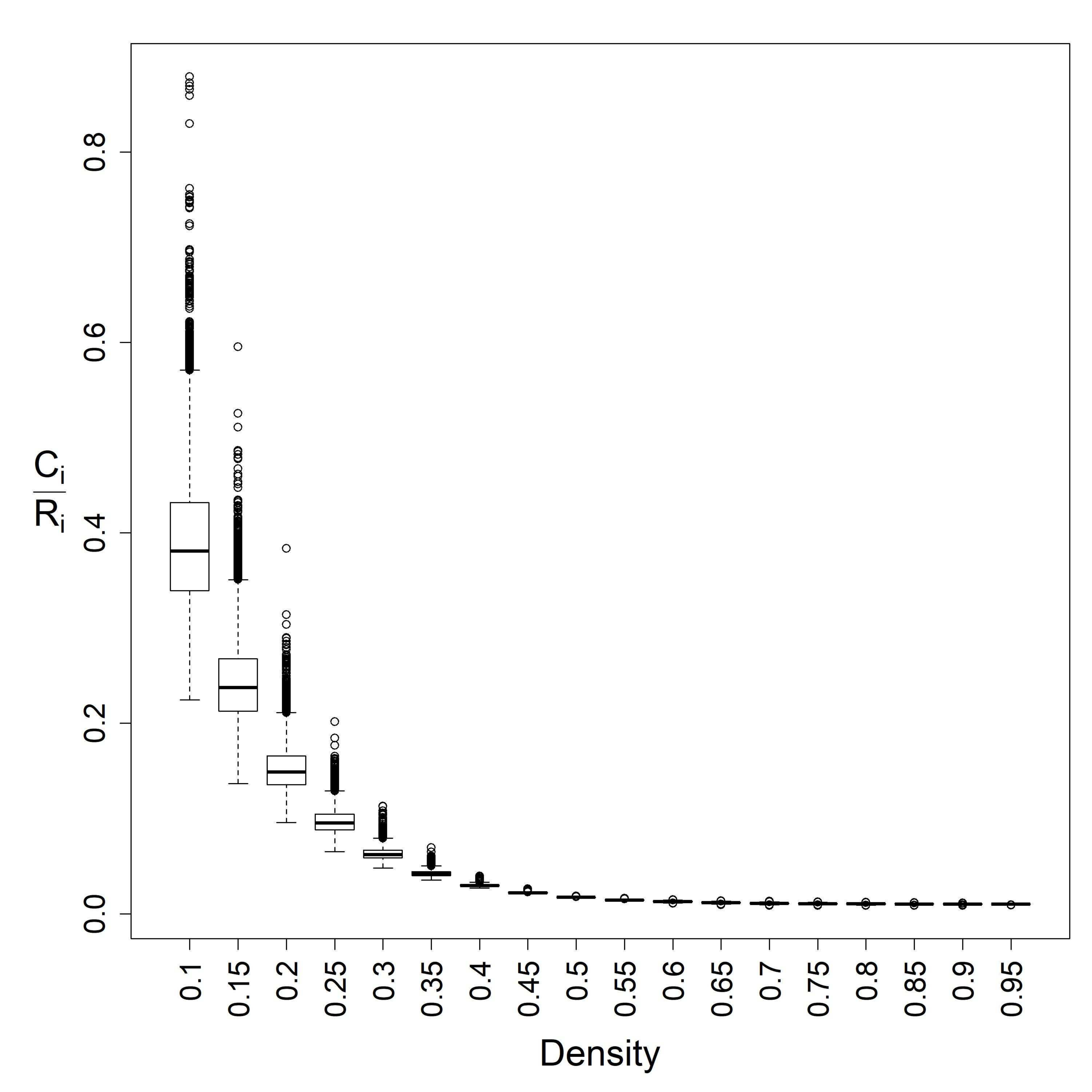}
		
	}
	
	\caption{Figures a), b), c) and d) display the distributions of ratios $\frac{{\mathscr{R}}_{i}}{{\mathbb E}\left({\mathscr{R}}_{i}\right)}$,
		$\frac{{\mathscr{C}}_{i}}{{\mathbb E}\left({\mathscr{C}}_{i}\right)}$, $\frac{{\mathscr{T}}_{i}}{{\mathbb E}\left({\mathscr{T}}_{i}\right)}$ and $\frac{{\mathscr{C}}_{i}}{{\mathscr{R}}_{i}}$ respectively, computed in case of a low external risk ($\zeta=0.1$). All Figures are based on 1000 randomly generated ER networks $\Gamma_{ER}(n;p)$ with a density varying between 0.1 and 0.9.}
	\label{fig:fig2}
\end{figure}

Lastly, in \cref{fig:fig3}, we focus on the ratio $\frac{{\mathscr{C}}_{i}}{{\mathscr{R}}_{i}}$ and we report the incidence of the circulability on the risk-dependent centrality as a function of the external risk $\zeta$. In case of sparse networks (\cref{fig:fig3}(a)), when the external risk is low, we have
that the infection remains in larger part circulating in a loopy way around the nodes, while only a lower proportion of risk tends to be transmitted to other nodes. 
This is due to the fact that, for $A$ sparse and $\zeta$ small, the matrix
$e^{\zeta A} = I + \zeta A + \frac{\zeta^2}{2} A^2 + O(\zeta ^3)$
is strongly diagonally dominant. When the external risk is high, as already observed, we have an average incidence of the circulability $\mathscr{C}_{i}$ on the risk-dependent centrality around $\frac{1}{n}$. On the contrary, when a very dense network is considered, the ratio $\frac{\mathscr{C}_{i}}{\mathscr{R}_{i}}$ is very little
affected by the external risk. In this case, both $\mathscr{C}_{i}$ and $\mathscr{R}_{i}$ increase on average at the same rate when $\zeta$ increases. However, the decreasing behavior of $\frac{\mathscr{C}_{i}}{\mathscr{R}_{i}}$ is noticeable for very low values of $\zeta$.

\begin{figure}[H]
	\subfloat[]{\includegraphics[width=0.45\textwidth]{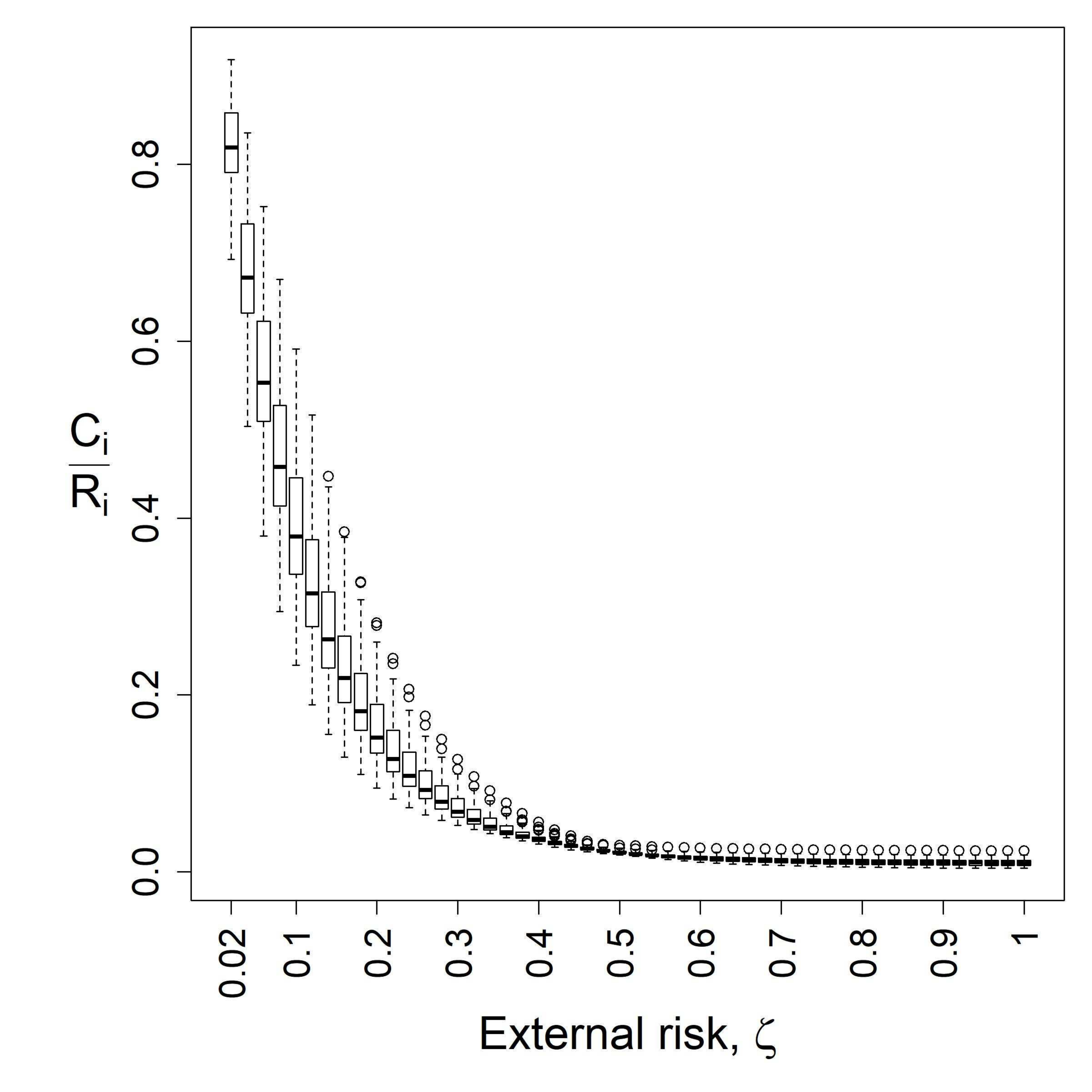}
		
	}\subfloat[]{\includegraphics[width=0.45\textwidth]{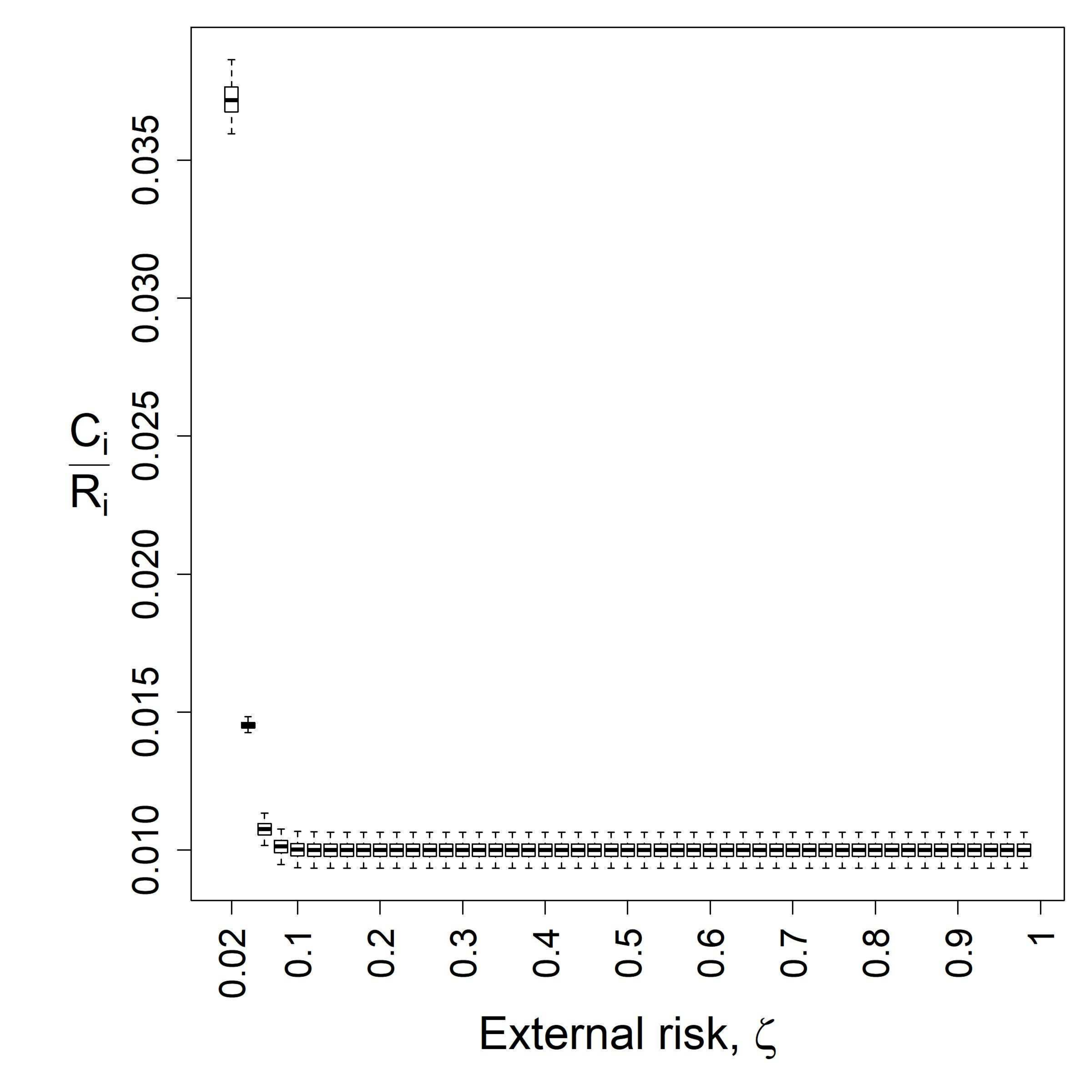}
		
	}
	
	\caption{Figures report the distribution of the ratios between the circulability
		${\mathscr{C}}_{i}$ and the risk-dependent centrality of each node
		${\mathscr{R}}_{i}$, computed for different $\zeta$ and by using
		generated ER graphs with a density equal to 0.1 (Figure a) and 0.9
		(Figure b), respectively. Both Figures are based on 1000 randomly generated ER networks $\Gamma_{ER}(n;p)$.}
	\label{fig:fig3}
\end{figure}

In what follows we provide an exhaustive proof of the behaviors observed so far.
Let us start with the pattern of the ratio $\frac{{\mathscr{C}}_{i}}{\mathscr{R}_{i}}$ at high density (see Figures \ref{fig:fig1}(d), \ref{fig:fig2}(d) and \ref{fig:fig3}(b)).

The asymptotic behavior of this ratio can be explained as a consequence of Theorem \ref{thm3} in the Appendix, where we derive the close expressions of the three risk-dependent centrality measures for a complete graph. In fact, as $\delta  \to 1$, the ER network approaches a complete network and, for $\zeta$ increasing, the ratio $\frac{{\mathscr{C}}_{i}}{\mathscr{R}_{i}}$ approaches $1/n$, as shown in \ref{ratio}.

Nonetheless, this result can be generalized. In fact, for an ER network which is dense enough, the following property holds for any $\zeta$.

\begin{theorem}\label{thm4}
	Let $\Gamma_{ER}(n;p)$ be an Erd\H{o}s-Rényi random graph with $n$
	nodes and probability $p$. If the edge density of the graph
	is $\delta>\left(\log n\right)^{6}/n$ and $p\left(1-p\right)>\left(\log n\right){}^{4}/n$, then for any node $i$
	
	\begin{equation}
	\underset{n\rightarrow\infty}{\lim}\, \dfrac{n{\mathscr{C}}_{i}}{\mathscr{R}_{i}}=1, 
	\end{equation}
	independently of $\zeta$.
\end{theorem}

\begin{proof}
	Let us consider as usual that $\lambda_{1}>\lambda{}_{2}\geq\cdots\geq\lambda_{n}$
	in a connected graph. It is known that in an ER graph the spectral gap $\left(\lambda_{1}-\lambda_{2}\right)\gg0$.
	Indeed, as proved in \cite{Janson}, $\lim_{n\rightarrow\infty}\frac{\lambda_{1}}{np}=1$,
	while $\lambda_{2}$ and $\lambda_{n}$ grow more slowly as $\lim_{n\rightarrow\infty}\frac{\lambda_{2}}{n^{\varepsilon}}=0$
	and $\lim_{n\rightarrow\infty}\frac{\lambda_{n}}{n^{\varepsilon}}=0$ for
	every $\varepsilon>0.5$, respectively.
	
	Then, if $np\left(1-p\right)>\left(\log n\right){}^{4}$,
	all but the largest eigenvalue lie with high probability in the interval
	$\sqrt{np\left(1-p\right)}\left[-2+o\left(1\right),+2+o\left(1\right)\right]$
	(see \cite{Vu2005} and \cite{Knowles2017}). Therefore,
	
	\begin{equation}
	\underset{n\rightarrow\infty}{\lim}\, \dfrac{\mathcal{\mathscr{C}}_{i}}{\mathscr{R}_{i}}=\underset{n\rightarrow\infty}{\lim}\, \dfrac{\psi_{1,i}^{2}e^{\zeta\lambda_{1}}+\sum_{k=2}^{n}\psi_{k,i}^{2}e^{\zeta\lambda_{k}}}{\psi_{1,i}\left(\vec{\psi}_{1}^{T}\vec{1}\right)e^{\zeta\lambda_{1}}+\sum_{k=2}^{n}\psi_{k,i}\left(\vec{\psi}_{k}^{T}\vec{1}\right)e^{\zeta\lambda_{k}}}=\dfrac{\psi_{1,i}}{\sum_{j=1}^{n}\psi_{1,j}}.
	\end{equation}
	
	The edge density of an ER graph is $\delta=p.$ In \cite{Erdos2013}, it was proved
	that for $np>\left(\log n\right)^{6}$, there exists a positive
	constants $C$ such that the following inequality holds
	
	\begin{equation}
	\left\Vert \vec{\psi}_{1}-\dfrac{1}{\sqrt{n}}\vec{1}\right\Vert _{\infty}<C\dfrac{1}{\sqrt{n}}\dfrac{\log n}{\log\left(np\right)}\sqrt{\dfrac{\log n}{np}},
	\end{equation}
	which in plain words means that an ER graph of density $\delta>\left(\log n\right)^{6}/n$ is ``almost'' regular when $n\rightarrow\infty$. That is $\lim_{n\to \infty}\sqrt{n}\psi_{1,i}=1$ for every node $i$. Thus, the result
	immediately follows.
\end{proof}

It is worth pointing out that,  when the density of an ER network is very low, the standard deviation of the ratio $\dfrac{{\mathscr{C}}_{i}}{\mathscr{R}_{i}}$
is very large with respect to that of ER networks with large densities (as shown in Figure \ref{fig:fig1}(d)). As we have proved before, the convergence of this ratio to the value $n^{-1}$ takes place only when the density of the graph is relatively large. Let us now analyze what happens when the edge density is very small for large graphs. In this case, we observe a slower decay of the ratio $\dfrac{{\mathscr{C}}_{i}}{\mathscr{R}_{i}}$ as a function of the external risk in the range $[0,1]$ (see Figure \ref{fig:fig3}(a)). This fact can be easily proven as follows. In general, both the numerator and denominator of this ratio can be expressed as infinite series of the type:

\begin{equation*}
{\mathscr{C\left(\zeta\right)}}_{i}=Q\left(\zeta\right)=1+a_{2}\zeta^{2}+\cdots+a_{k}\zeta^{k}+\cdots,
\end{equation*}

\begin{equation*}
{\mathscr{R\left(\zeta\right)}}_{i}=H\left(\zeta\right)  =1+b_{1}\zeta+\left(a_{2}+b_{2}\right)\zeta^{2}+\cdots+\left(a_{k}+b_{k}\right)\zeta^{k}+\cdots\\
=  Q\left(\zeta\right)+L\left(\zeta\right),
\end{equation*}

where $a_{k}$ counts the number of closed walks of length $k$ starting
and ending at node $i$ and $b_{k}$ counts all the open walks of
length $k$ starting at $i$ and ending at any node $j\neq i$. Let
us consider

\begin{equation*}
\begin{split}& \dfrac{d}{d\zeta}\left(\frac{Q\left(\zeta\right)}{Q\left(\zeta\right)+L\left(\zeta\right)}\right) =\frac{L\left(\zeta\right)Q'\left(\zeta\right)-L'\left(\zeta\right)Q\left(\zeta\right)}{\left[Q\left(\zeta\right)+L\left(\zeta\right)\right]^{2}}\\
& =\frac{\left(2a_{2}b_{1}\zeta^{2}+\cdots+2a_{2}b_{k}\zeta^{k+1}+\cdots\right)-\left(b_{1}+2b_{2}\zeta+a_{2}b_{1}\zeta^{2}+\cdots+b_{1}a_{k}\zeta^{k}+\cdots\right)}{\left[Q\left(\zeta\right)+L\left(\zeta\right)\right]^{2}}
\end{split}
\end{equation*}

Then, for certain $\zeta<1$ the numerator of the previous expression
is negative, which means that the ratio $\dfrac{\mathcal{\mathscr{C}}_{i}\left(\zeta\right)}{\mathscr{R}_{i}\left(\zeta\right)}$
is monotonically decreasing with $\zeta$. For instance, let us make
a second order approximation to the polynomials $Q\left(\zeta\right)$
and $H\left(\zeta\right)$. Then, we have

\[
\frac{Q\left(\zeta\right)}{H\left(\zeta\right)}=\frac{1+\frac{1}{2}\zeta^{2}k_{i}}{1+\zeta k_{i}+\frac{1}{2}\zeta^{2}\left(k_{i}+P_{2,i}\right)},
\]
where $P_{2,i}$ is the number of paths of length 2 (wedges) starting
at node $i$. In an ER graph $\mathbb{E}\left(k_{i}\right)=\left(n-1\right)p$
and $\mathbb{E}\left(P_{2,i}\right)=\left(n-1\right)^{2}p^{2}-\left(n-1\right)p$.
Thus,

\[
\frac{Q\left(\zeta\right)}{H\left(\zeta\right)}\approx\frac{1+\frac{1}{2}\zeta^{2}\left(n-1\right)p}{1+\zeta\left(n-1\right)p+\frac{1}{2}\zeta^{2}\left(n-1\right)^{2}p^{2}}=\frac{1+\cfrac{\bar{k}}{2}\zeta^{2}}{1+\bar{k}\zeta+\cfrac{\bar{k}^{2}}{2}\zeta^{2}},
\]
where $\bar{k}=\left(n-1\right)p$ is the mean degree. The first derivative
of this rational function is

\[
\dfrac{d}{d\zeta}\left(\frac{Q\left(\zeta\right)}{H\left(\zeta\right)}\right)=\frac{2\bar{k}^{2}\zeta^{2}-\left(4\bar{k}\left(\bar{k}-1\right)\zeta+4\bar{k}\right)}{\left(2+2\bar{k}\zeta+\bar{k}^{2}\zeta^{2}\right)^{2}},
\]
which is always negative for any $\bar{k}\geq1$ and $0\leq\zeta\leq1$
as can be seen in Figure \ref{plot}. Moreover, the absolute value of this derivative increases as $\bar{k}$ decreases, implying a slower decay in the function $\dfrac{\mathcal{\mathscr{C}}_{i}\left(\zeta\right)}{\mathscr{R}_{i}\left(\zeta\right)}$ for lower densities.

\begin{figure}
	\begin{centering}
		\includegraphics[width=0.6\textwidth]{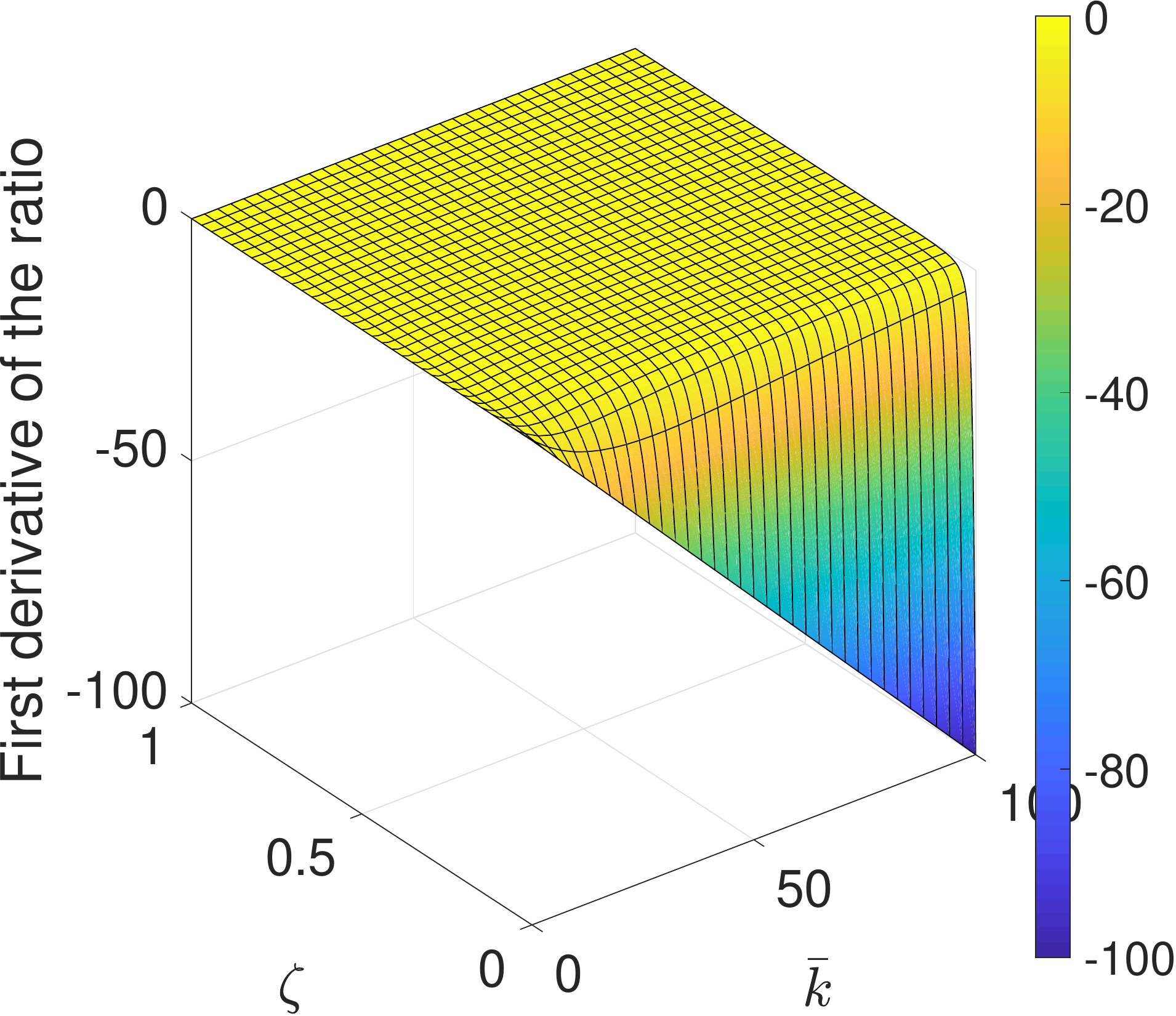}
		\par\end{centering}
	\caption{Illustration of the behavior of the derivative of the ratio $\dfrac{\mathcal{\mathscr{C}}_{i}\left(\zeta\right)}{\mathscr{R}_{i}\left(\zeta\right)}$
		for values of $0\protect\leq\zeta\protect\leq1$ and for the parameter
		$\bar{k}\protect\geq1.$}
	
	\label{plot}
\end{figure}
\color{black}

To conclude this section, we want to focus on the rankings produced by the two main centrality measures ${\mathscr{R}}_{i}$ and ${\mathscr{C}}_{i}$ and on the similarities between them. In particular, we are interested in determining if, or for what type of networks, the different centrality measures provide similar rankings. To this
end, we display in \cref{tab:RankERnew} the Spearman correlation
coefficient between the risk dependent centrality ${\mathscr{R}}_{i}$ and the circulability ${\mathscr{C}}_{i}$ for different graph densities and for various values of $\zeta$.
On average, we observe a strong positive monotonic dependence between
the two centrality measures. As expected, the two measures tend towards
the perfect monotonicity as the density arises. It is noteworthy the
behavior with respect to $\zeta$. The higher dependence is observed
in a low-risk framework ($\zeta=0.1$), while a slight reduction is
noticeable when higher risk contexts are analyzed, providing again
an empirical evidence of the fact that differences between nodes are
increased in stressed conditions. Furthermore, this result is in line
with the higher incidence of ${\mathscr{C}}_{i}$ on ${\mathscr{R}}_{i}$
as $\zeta$ vanishes, discussed in the previous lines. For the sake of brevity, we do not report the Spearman correlation between $\mathscr{R}_{i}$ and $\mathscr{T}_{i}$.
However, in all cases, the coefficient is larger than 0.9999.

\begin{table}[H]
	\centering{}%
	\caption{Spearman correlation coefficients between $\mathscr{C}_{i}$
		and $\mathscr{R}_{i}$ in ER graphs with 100 vertices at different
		densities and different values of $\zeta$.}
	\begin{tabular}{|lc|ccccc|}
		\hline 
		&  & \multicolumn{5}{c|}{\bf Density}\tabularnewline
		&  & 0.1 & 0.3 & 0.5 & 0.7 & 0.9\tabularnewline
		\hline 
		& 0.1 & 0.9947 & 0.9967 & 0.9971 & 0.9994 & 0.9998\tabularnewline
		$\zeta$ & 0.5 & 0.9844 & 0.9950 & 0.9966 & 0.9994 & 0.9998\tabularnewline
		& 1.0 & 0.9813 & 0.9950 & 0.9966 & 0.9994 & 0.9998\tabularnewline
		\hline 
	\end{tabular}
	\label{tab:RankERnew}
\end{table}

\section{Analysis of real-world financial networks}
\label{sec:RWfin}
In this section, we perform some empirical studies in order to
assess the effectiveness of the proposed approaches. We consider two
different networks. In the first one, we collected daily returns of a dataset
referred to the time-period ranging from January 2001 to December
2017, that includes $102$ leading U.S. stocks constituents of the
$S\&P$ 100 index at the end of 2017. Data have been downloaded from
Bloomberg. Returns have been split by using monthly stepped six-months
windows. It means that the data of the first in-sample window of width
six-month are used to build the first network. The process is repeated
rolling the window one month forward until the end of the dataset
is reached, obtaining a total of 199 networks. The first network,
denoted as \lq\lq1-2001\rq\rq covers the period $1^{st}$ of January
2001 to $30^{th}$ of June 2001. The latter one (\lq\lq 7-2017\rq\rq)
covers the period $1^{st}$ of July 2017 to $31^{th}$ of December
2017.

Hence, for each window, we have a network $\Gamma_{t}=(V_{t},E_{t})$ (with
$t=1,...,199$), where assets are nodes and links are weighted by
computing the correlation coefficient $_{t}\rho_{i,j}$ between the
empirical returns of each couple of assets. Notice that the number
of assets can vary over time. Indeed, as mentioned, we have considered the 102 assets
constituents of the $S\&P$ 100 index at the end of 2017. Some of
these assets have no information available for some specific time
periods. Therefore, in each window, we have considered only assets,
whose observations are sufficiently large to assure a significant
estimation of the correlation coefficient. However, it is not the
aim of this paper to deal with the effects of alternative estimation
methods. As a consequence, the number of nodes in the 199 networks
varies from 83 to 102 during the time-period.

Then, we follow the methodology proposed in \cite{Mantegna,Onnela}
and we use the non-linear transformation, based on distances $_{t}d_{i,j}$:
$_{t}d_{i,j}=\sqrt{2(1-{}_{t}\rho_{i,j})}$. The distance matrix $D_{t}=[{}_{t}d_{i,j}]_{i,j\in V_{t}}$,
with elements $0\leq{}_{t}d_{i,j}\leq2$, becomes the weighted adjacency
matrix of the graph $\Gamma_{t}$. As proposed in \cite{Onnela}, we extract
the minimum spanning tree $T_{t}$. This is a simple connected graph
that connects all $n_{t}$ nodes of the graph with $n_{t}-1$ edges
such that the sum of all edge weights $\sum_{_{t}d_{i,j}\in T_{t}}{}_{t}d_{i,j}$
is minimum. As shown in \cite{Onnela}, this minimum spanning tree,
as a strongly reduced representative of the whole correlation matrix,
bears the essential information about asset correlations. Furthermore,
the study of the centrality of nodes and the analysis of the evolution
of the tree over time are two critical issues in portfolio selection
problem (see \cite{Onnela,Peralta,Pozzi}).

The second dataset consists of a network of the top corporates in
US in 1999 according to Forbes magazine. The network is constructed
as follows. First we consider a bipartite network in which one set
of nodes consists of companies and the other of Chief Executive Officer
(CEO)'s of such companies. As one CEO can be in more than one company,
we make a projection of this bipartite graph into the company-company
space. In this way, the nodes represent corporations and two corporations
are joined by an edge if they share at least one director. We consider
two versions of this network, in the first we use the number of directors
shared by two companies as an edge weight, and in the second we use
the binary version of the first. We will refer to these as to the weighted
and binary network, respectively. The network has $824$ nodes, made
up of one giant component of $814$ nodes. We selected the giant component,
with its binary and weighted adjacency matrices. For a comprehensive description of this network see, for instance, \cite{Davis2003}. \\
Networks, derived by both datasets, have been studied by computing
the total communicability, circulability and transmissibility for
each node with $\zeta$ varying in $(0,1]$ with step $0.01$.

\subsection{Network of assets}

Starting from the asset trees $T_{t}$, we measure the relevance of
each node by using the risk-dependent centrality $\mathscr{R}_{i}$
and by testing different values of $\zeta$. We consider in \cref{fig:figR2001} the rankings' distribution of each asset. Different
outcomes of each distribution have been obtained by computing the
rankings based on $\mathscr{R}_{i}$ for alternative values of $\zeta$
in the interval $(0,1]$ with step $0.01$. These results regard the
first network \lq\lq 1-2001\rq\rq, namely, the network based on
data that cover the period $1^{st}$ of January 2001 to $30^{th}$
of June 2001. We observe that some nodes show a significant variability
according to different values of $\zeta$. Indeed, some assets have
climbed more than 20 positions in the ranking when $\zeta$ increases.
For instance, Amazon (node 7 in \cref{fig:figR2001}) moved
from position 66 to 41 in case of low and high risk, respectively.
Vice versa, Exelon Group (node 32 in \cref{fig:figR2001}) lowered
its ranking from 15 to 46. On the other hand, the most central nodes
in the network remain very central also when external risk is very
high. We have indeed that the top 6 is quite stable for different
values of $\zeta$. Top assets only exchange a bit their position,
preserving their central role. For instance, United Technologies Corporation
(node number 79 in \cref{fig:figR2001}) is at the top of the
ranking, independent of $\zeta$.

\begin{figure}[H]
	\begin{centering}
		\includegraphics[width=0.9\textwidth]{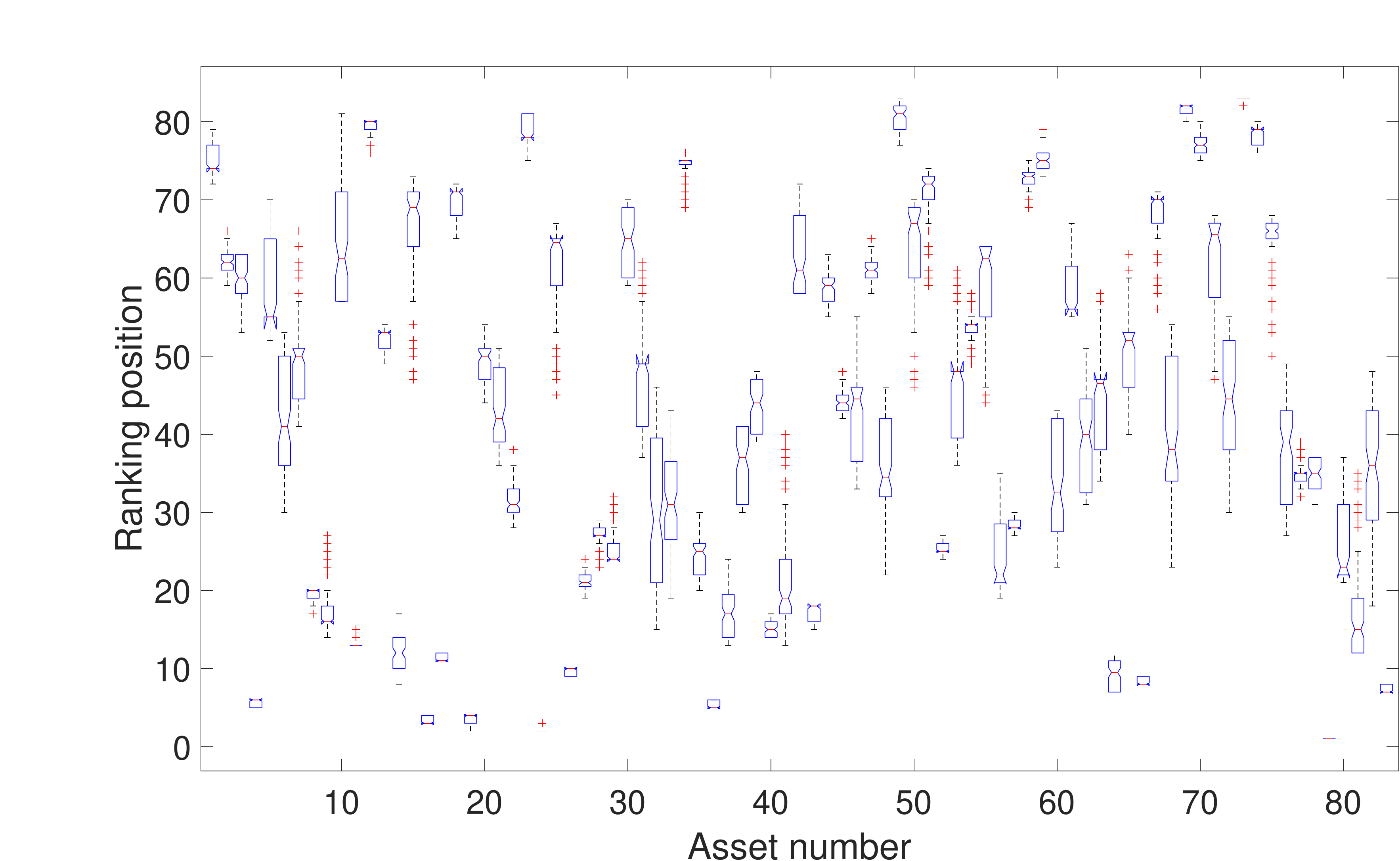}
		\par\end{centering}
	\caption{Figure reports the distribution of nodes' rankings based on $\mathscr{R}_{i}$
		with respect to $\zeta$. For each distribution, the set of outcomes
		is given by the rankings of $\mathscr{R}_{i}$ computed for alternative
		values of $\zeta$. Results regard the network $T_{1}$, i.e. the
		asset-tree in the first window $1-2001$.}
	
	\label{fig:figR2001}
\end{figure}

If we consider the period of the global financial crisis of 2007-2008
(see \cref{fig:figR2007} and \ref{fig:figR2008}), we observe
an increase in the rankings' volatility. In shock periods, centrality
of nodes is more affected by the value of $\zeta$. 
In particular, to catch rankings' volatility, we report in Figure \ref{fig:figSD} the standard deviations of rankings of each asset computed varying $\zeta$. In shocks periods, results confirm  higher average volatility as well as positive skewed distributions because of a greater number of assets whose ranking is highly affected by the value of $\zeta$. We also tested that differences in average volatility are significant by means of a paired t-test, useful for comparing the same sample of assets at different time periods. When the network 1-2001 is compared with the two networks covering period of crisis (End 2007 or End 2008), we obtain $p$-values around $10^{-5}$ and $10^{-8}$ that confirm strong evidence against the null hypothesis that the average difference between the two samples is zero. As expected, the test is not statistically significant ($p$-value is 0.31) when networks covering period of the global financial crisis are compared.

\begin{figure}[H]
	\begin{centering}
		\includegraphics[width=0.9\textwidth]{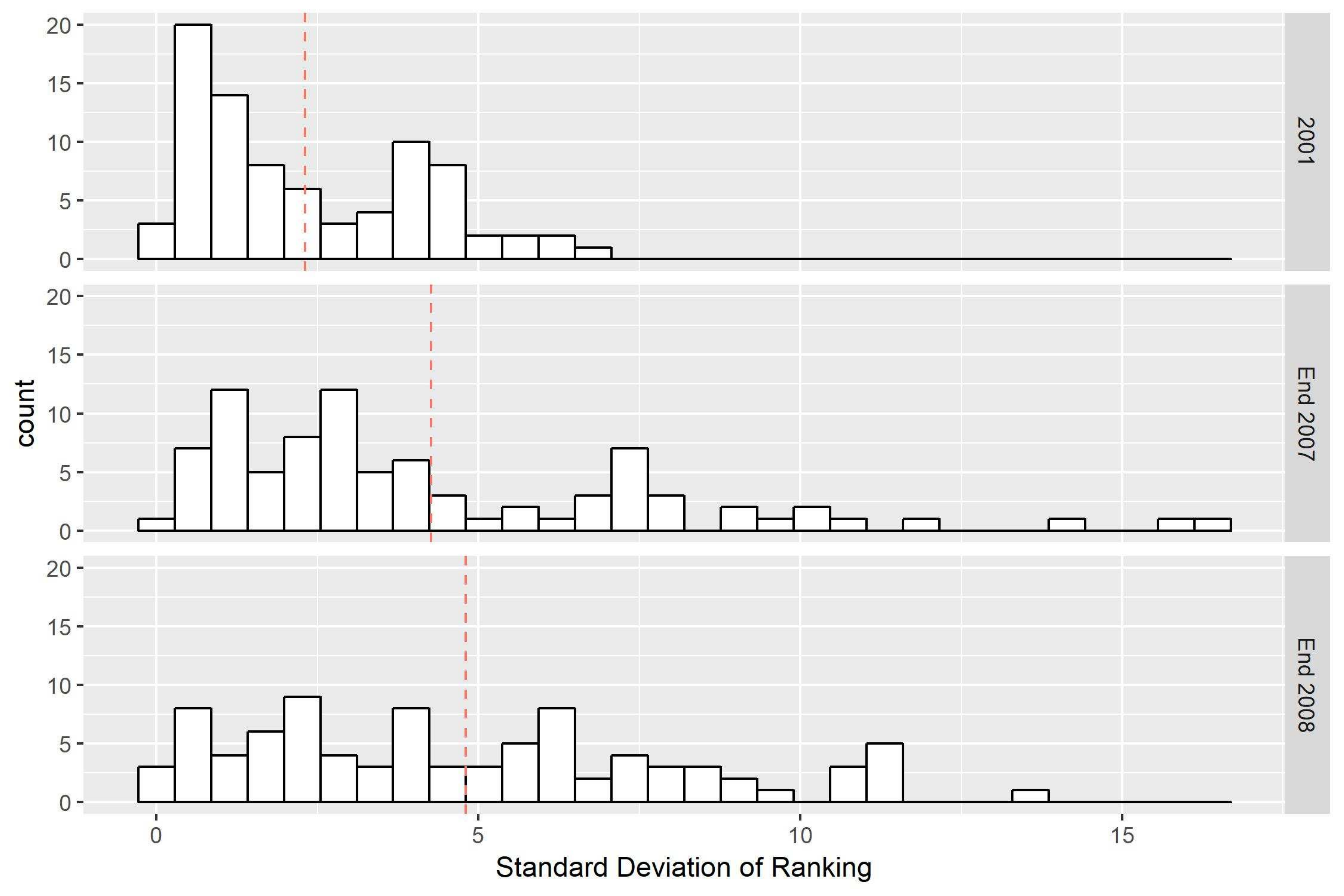}
		\par\end{centering}
	\caption{Figure reports the distribution of standard deviations of nodes' rankings based on $\mathscr{R}_{i}$
		with respect to $\zeta$. For each distribution, the set of outcomes
		is given by the standard deviation of rankings of $\mathscr{R}_{i}$ computed for alternative
		values of $\zeta$. Results regard respectively the network in the first window $1-2001$, at the end of $2007$ and at the end of $2008$. The dotted red lines indicate the average standard deviation: values are equal to 2.31, 4.27 and 4.80, respectively. }
	
	\label{fig:figSD}
\end{figure}

\color{black}

Concerning the behavior of specific assets, we observe, for
instance, that some assets move down by approximately 60 positions
from a low risk to an high risk framework. Two examples are represented
by Danaher Corporation and Honeywell International (assets 28 and
43, respectively, in \cref{fig:figR2007}). Instead, Accenture
PLC (node 3 in \cref{fig:figR2007}) increased its ranking from
position 61 to 11. \\
Even top central nodes are affected by $\zeta$ as the volatilities
of their rankings show. It is instead confirmed the relevance of United
Technologies Corporation (node number 82 in \cref{fig:figR2007}
and 83 in \cref{fig:figR2008}) that is again at the top of
the ranking at the end of 2017, independent of $\zeta$. At the
end of 2008, the centrality of this asset is also confirmed, although,
a bit of variability in the ranking is observed for this firm.

\begin{figure}[H]
	\begin{centering}
		\includegraphics[width=0.9\textwidth]{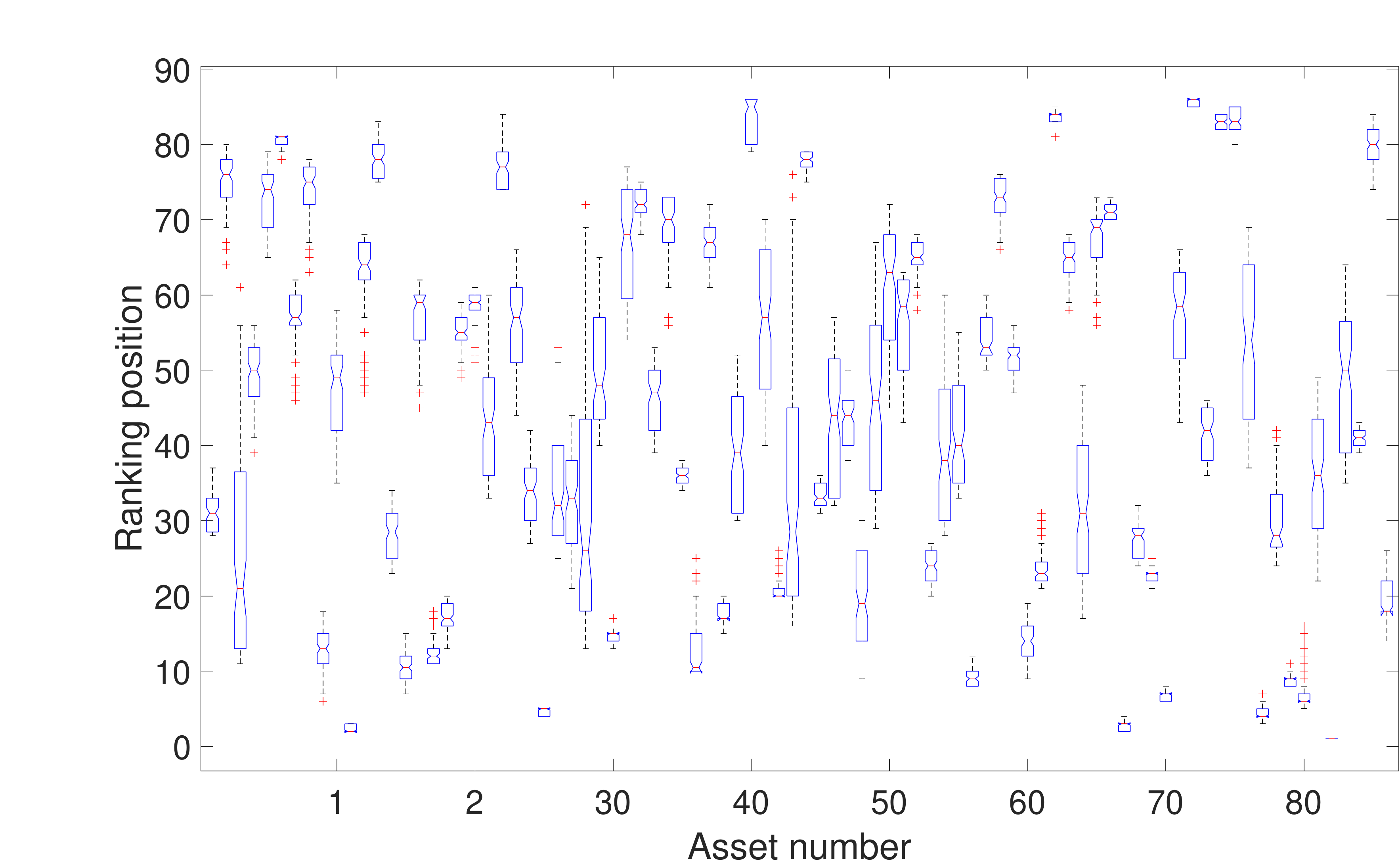}
		\par\end{centering}
	\caption{Figure reports the distribution of nodes' rankings based on $\mathscr{R}_{i}$
		with respect to $\zeta$. For each distribution, the set of outcomes
		is given by the rankings of $\mathscr{R}_{i}$ computed for alternative
		values of $\zeta$. Results regard the asset-tree at the end of $2007$.}
	
	\label{fig:figR2007}
\end{figure}

\medskip{}

\begin{figure}[H]
	\begin{centering}
		\includegraphics[width=0.9\textwidth]{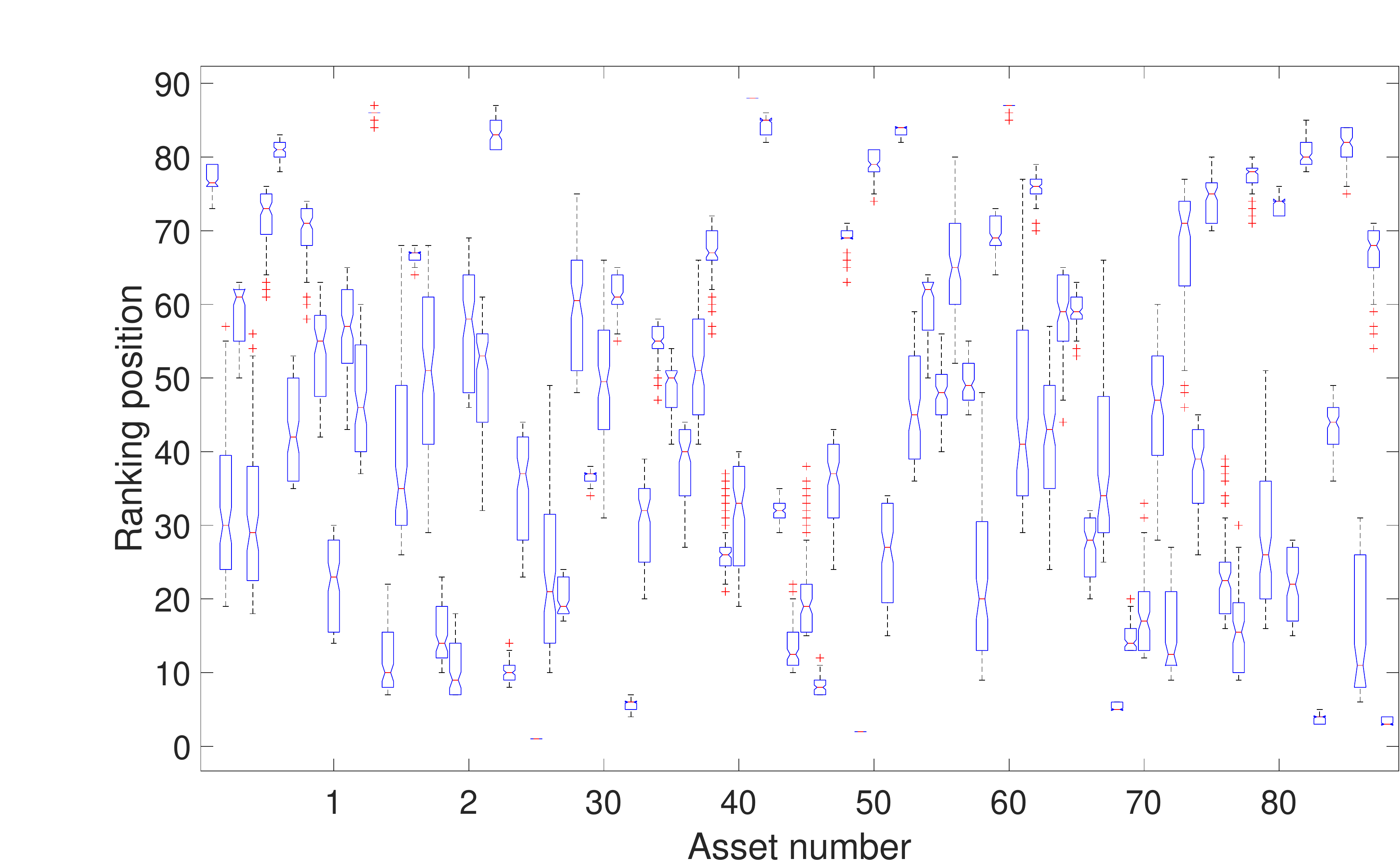}
		\par\end{centering}
	\caption{Figure reports the distribution of nodes' rankings based on $\mathscr{R}_{i}$
		with respect to $\zeta$. For each distribution, the set of outcomes
		is given by the rankings of $\mathscr{R}_{i}$ computed for alternative
		values of $\zeta$. Results regard the asset-tree at the end of $2008$.}
	
	\label{fig:figR2008}
\end{figure}

\subsection{US corporate network}\label{sec:US}

We now analyze the network of US top corporates in 1999 according to Forbes magazine. Before starting our analysis let us explain the importance of studying a spreading dynamics on this network. According to this network, the board of directors of a given corporation is formed by a few members, some of which are also present in the board of other corporations. Then, such directors serving on more than one board can act as spreaders of information between the corresponding corporations. Such information can be about future (favorable or unfavorable) economic situations, alarms, market opportunities, or anything that could be of interest to the companies in which the director is. Due to the global connectivity of the system, such ``information'' can be spread across the whole network ``infecting'' all the corporations in a relatively short time. As we have mentioned before, epidemiological models have also being used for modeling such propagation dynamics (see Section 1.1.). \\
Hence, we devote this section to the investigation about whether a significant increase of the risk-dependent centrality is a proxy of the vulnerability of the corporate to financial infections propagating on the network. At first, we should remark the fact that the network we are considering here was built based on data corresponding to year 1999. At this year the level of stress of the international economic system was relatively
high due to the fact that the East Asian financial crisis occurred in the years 1997--1998, which was also followed by the Russian default of 1998. The two aforementioned financial crises had a ripple effect on the US market. In the literature, for instance, it is well-documented the so-called ``fire-sale'' FDI (Foreign Direct Investment)
phenomenon, that is, the surge of massive foreign acquisitions of domestic firms during a financial crisis (\cite{Alquist2013,Krugman2008}).Thus, the level of stress and infectability of the system for the next few years after 1999 (we will eventually see that these correspond to the period 2000-2002) is expected
to be significantly larger than in the subsequent years when the effects of these crises gradually relaxed. Therefore, we proceed our analysis by considering that the level of infectability in 1999 is high and
we investigate the effects of relaxing such a condition to lower levels of stress. That is, we start by assuming that in 1999 the external market turmoil could be represented by a value of $\zeta=1$ and we want to find out how the companies change their ranking positions in term of risk-dependent centrality\footnote{The analysis has been also developed for circulability and transmissibility, but, since the significantly high rank correlation between $\mathscr{R}_{i}$ and $\mathscr{T}_{i}$ (with Spearman correlation coefficients larger than 0.99), we focus here only on $\mathscr{R}_{i}$.} $\mathscr{R}_{i}$ as $\zeta$ vanishes. To this purpose, we set up different initial conditions in the contagion model described by \cref{eq:310}, assigning to each year a different value of the infectability parameter $\gamma$, according to the environmental conditions of the market. Therefore, we let $\zeta$ factors reduce year by year in order to reflect a reduction in the overall stress on the network. In particular, we decrease $\zeta$ linearly from 1 to 0 in the period 1999-2003. Therefore, rankings based on the risk-dependent centrality computed for $\zeta=1$ allow to assess the relevance of each corporate in $1999$. Lowering  $\zeta$, we test how the positions of firms vary over time when the external risk reduces. It is noteworthy that the connection between this parameter and the risk could be quite loose but as provided by the following analysis the model seems to work quite well in describing firms that reduce their SVC in the period. \\
The variation of rankings is then compared with the pattern of the shareholder value creation (SVC) over time. According to the \textit{OECD Principles of Corporate Governance}, corporations should be run, first and foremost, in the interests of shareholders (OECD 1999). Therefore, companies should work to increase their shareholder values. Increasing shareholders value cannot be done without risk. It is known \cite{Lazonick_2} that in the shareholder value model, companies usually take more risk than needed in order to maximize SVC. As a consequence of this additional risk, companies acquire debts which could make them unstable and more exposed to the risk of bankruptcy. Acquiring large debts is seen as conductive to increasing shareholder value, due to the potential of the company to increase value when it has started from a low baseline. Thus, there is a relation between SVC and risk, because in searching for large SVC the companies increase their risks to attract more investors and increasing potential value gain, but, at the same time, the risk also puts the company in a more vulnerable position to bankruptcy and collapse. \\
To support our interpretation, we make use of SVCs of the companies\footnote{In particular, we use a sample of 337 companies in our network whose SVCs are made available in the dataset available in \cite{Fernandez}.} in the S\&P500 for the period 1999-2003, that have been collected by Fernández and Reinoso (see \cite{Fernandez}). 
Hence, we use SVC as a proxy for risk. Indeed, the global average of SVC reflects very well what happens for the period 1999-2003. After the financial crisis of 1998 the world was at a higher level of risk which is reflected by a dramatic drop of the SVC in year 2000 from a positive value in 1999 to a negative one in 2000. This situations remained until 2002, but eventually recovered to positive in 2003 (see Figure \ref{fig:SVCTot}).
It is noteworthy that the data for SVC was reported by Fernandez and Reinoso for the years 1993-2003. From this long period we select the segment 1999-2003 which contains exactly the valley produced from the financial crisis of 1998 and also because the data used for building the corporate network is of 1999. That is, it corresponds to a segment in time in which the world economy drop due to a crisis and then eventually recovered from it.

\begin{figure}[H]
	\centering
	\includegraphics[width=0.6\textwidth]{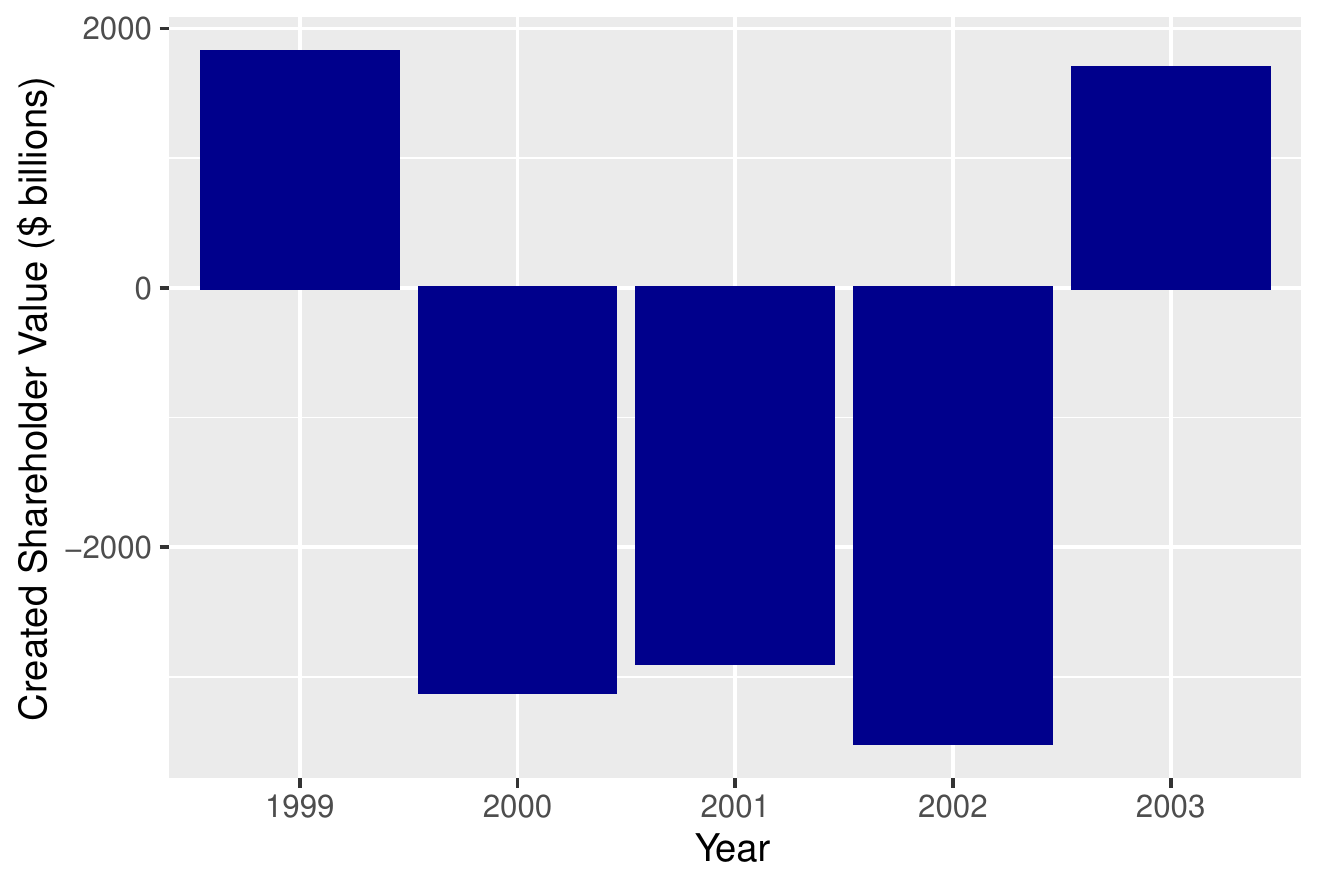}
	\caption{Total created shareholder value (\$ billion) of firms constituent the index S\&P500 for the period 1999-2003 (data taken from \cite{Fernandez}). } \label{fig:SVCTot} 
\end{figure}

We focus our statistical analysis on the predictability of the risk-dependent centrality on the evolution of the SVC. We consider the evolution of the SVC of a company for the period 1999-2003, which is the period immediately after the network of corporate elite in US was built. As a proxy for the evolution of the SVC of a company we consider the Pearson correlation coefficient $\rho$ of the ranking position of the company based on SVC versus the reciprocal of the year. In this case, a negative (positive) value of $\rho$ indicates that the corresponding company decreases (increases) its SVC from 1999 to 2003 when the global external infectability decreases.
Therefore, we apply a Linear Discriminant Analysis (LDA) to classify the companies into two groups: (i) those with negative trend in the SVC for this period, and (ii) those with a positive one. The only predictor used for this classification is the parameter $\Delta{\rm Rank}({\mathscr R}_{i})$. This parameter is the difference between the ranking position of the company $i$ when $\zeta=1$ and the ranking position of the same company when $\zeta=0.01$. In other words, a negative (positive) value of $\Delta{\rm Rank}({\mathscr R}_{i})$ means that the company dropped (increased) its exposure to risk when the infectability of the system is lower. \\
Before proceeding with the application of the LDA on the whole sample at disposal, we eliminate a few companies whose correlation coefficient between SVC and the reciprocal of the year is marginal (i.e., close to zero). We test empirically the effect produced by the removal of companies for which $ |\rho| <a$ for different values of the threshold $a$, e.g., $a=0.01, 0.025, 0.05, 0.075, 0.1$. The best classification of the companies into the two groups analyzed is obtained by eliminating those companies for which $| \rho |<0.05$. 
In this case the total accuracy of the LDA model is 60.5\%. That is, 200 out of 332 of the companies are classified correctly in their respective groups representing their trends in shrinking SVC or expanding it. In particular, the fitted LDA model is $\hat{Y}_{i}=-0.3177+0.0102\Delta{\rm Rank}({\mathscr R}_{i})$, where $\hat{Y}_{i}$ is the predicted response variable of our analysis that allows to classify companies in their respective group. The positive coefficient of the variable $\Delta{\rm Rank}({\mathscr R}_{i})$ indicates that: (i) increasing the exposure to risk ($\Delta{\rm Rank}({\mathscr R}_{i})>0$) tends to expand the SVC of the company, and (ii) decreasing the exposure to risk ($\Delta{\rm Rank}({\mathscr R}_{i})<0$) tends to shrink the SVC of the company.
For both groups, we report in Figure \ref{fig:SVCR}(a) a comparison between the predicted value with the LDA and the observed value for each firm. Red squares below the line and blue circles over the line are well classified, while blue circles below the line and red squares over the line are wrongly classified. Furthermore, in Figure \ref{fig:SVCR}(b) we report the related confusion plot,where the number of true negative and true positive are on the anti-diagonal (bottom and upper parts, respectively) and the number of false negative and false positive are on the main diagonal (bottom and upper parts, respectively). It is noticeable the low classification performance of the model, when only companies that expand their SVC are considered (in this regard, see in Figure \ref{fig:SVCR}(b) companies that belong to the observed class denoted with the sign “+”). For instance, from 147 companies in the network which increase their SVC in the period 1999-2003 only 36 are correctly predicted by $\Delta{\rm Rank}({\mathscr R}_{i})$ in their class. On the contrary, from the 175 companies that shrink their SVC in the period 1999-2003, the variable $\Delta{\rm Rank}({\mathscr R}_{i})$ correctly predicts 157 companies in this class. That is, the risk-dependent centrality of the companies clearly identifies about 90\% of the companies which will shrink their SVC in the period 1999-2003, using only data referring to the year 1999. In plain words, our results indicate that diminishing the exposure to risk when the external conditions of infectability are low, with high probability, reduces the SVC of a company.


\begin{figure}[H]
	\centering
	\subfloat[]{\includegraphics[width=0.45\textwidth]{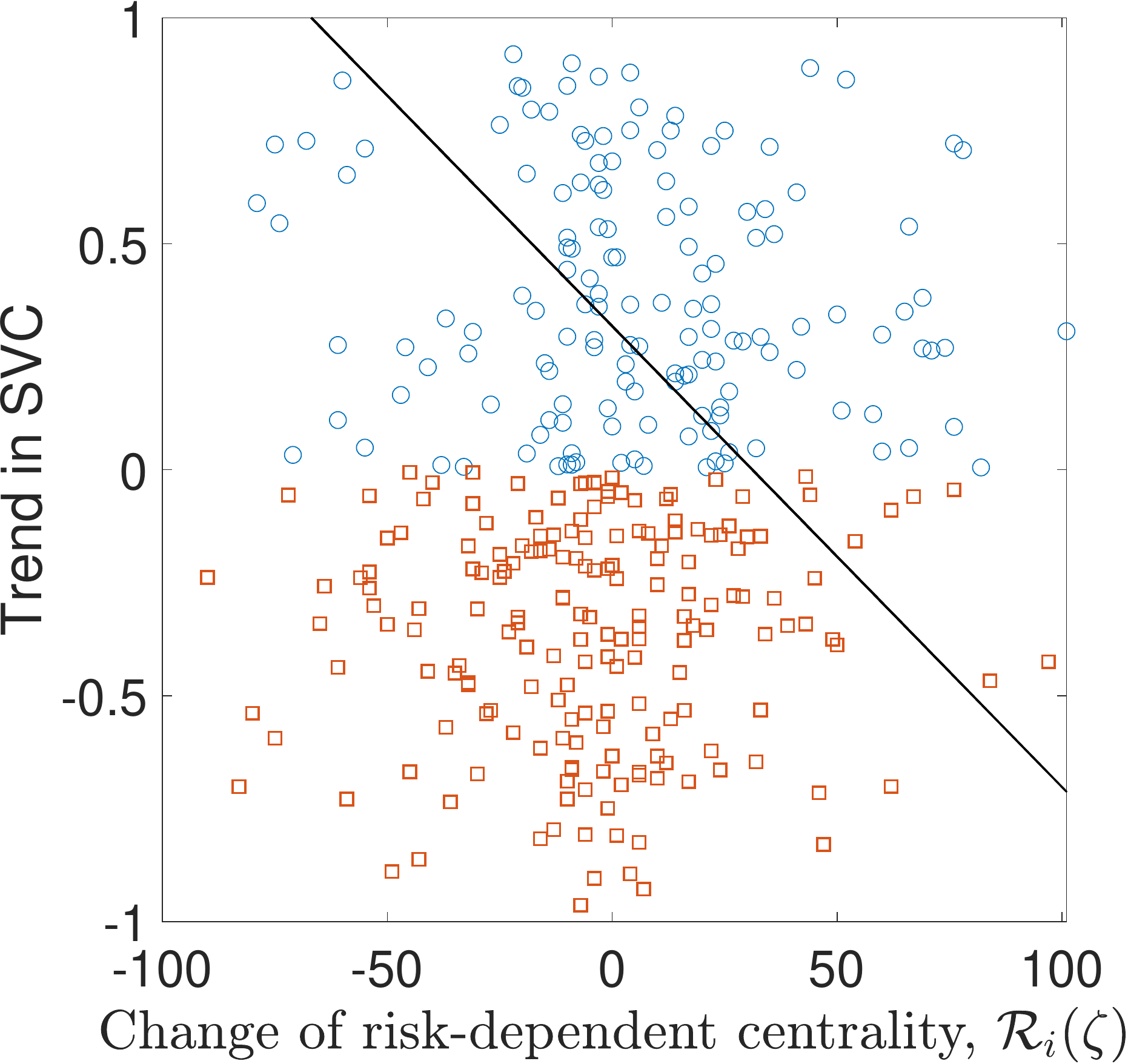}
		
	}\subfloat[]{\includegraphics[width=0.45\textwidth]{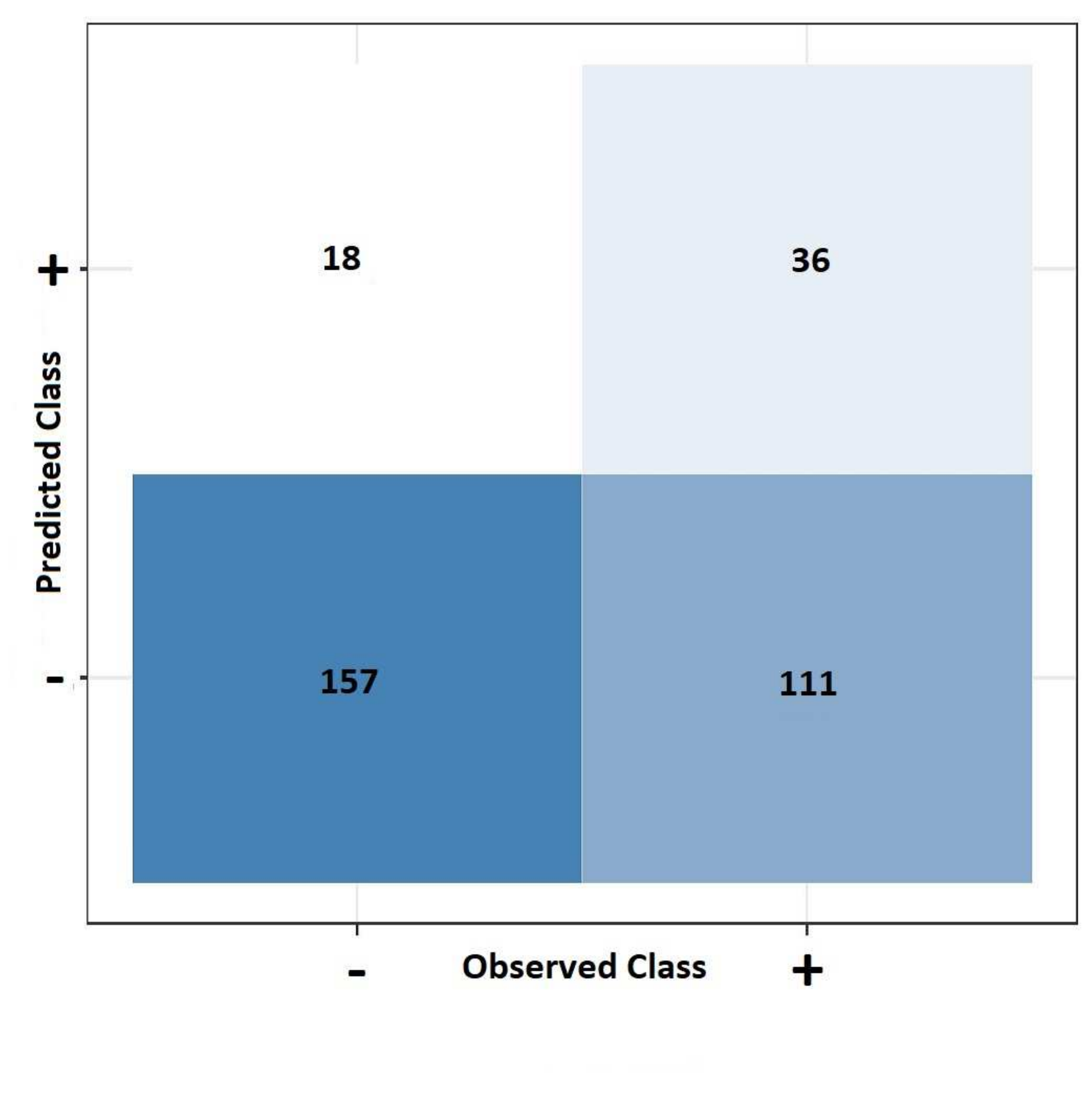}
		
	}
	\caption{(a) Illustration of the linear discriminant analysis (LDA) model classifying
		the trend of corporations into those shrinking their SVC (red squares)
		and those expanding it (blue circles). The black line represents the
		LDA model based on the change of $\mathscr{R}_{i}\left(\zeta\right)$
		for the values of $\zeta=0.01$ and $\zeta=1.0$ to predict the trend
		in the SVC. Red squares below the line and blue circles over the line
		are well classified, while blue circles below the line and red squares
		over the line are wrongly classified. The LDA classifies
		correctly about 90\% of all companies who shrank their SVC (red squares). (b) Plot of the confusion matrix. On the $x$-axis we report the true class (Observed Class), on the $y$-axis the predicted class (Output Class). The number of true negative and true positive cases are on the anti-diagonal (bottom and upper parts, respectively) and the number of false negative and false positive cases are on the main diagonal of the matrix (bottom and upper parts, respectively).  In the class “minus” (“plus”) we consider companies with negative (positive) trend in the SVC for the period 1999-2003} \label{fig:SVCR} 
\end{figure}

Let us conclude with the following remark. Even if ``good'' companies increase their risk-centrality ranking as $\zeta$ vanishes, it is worth noting that this occurs when the global stress in the market is very low. When the infectability rate is very low, the absolute probability of getting infected also remains very low for both ``good''
and ``bad'' companies. To show this fact, let us consider that, according to our model, the probability that a given corporate is not affected by a crisis propagating inside the network is given by $1-x_{i}(t)=\alpha  e^{-\frac{\beta}{\alpha}\left(\mathscr{R}_{i}-1\right)}$, where again $\beta$ and $\alpha=1-\beta$ are the initial probabilities to have infected and not-infected nodes, respectively. Hence, the ratio between the probabilities of two nodes $i$ and $j$ to pass successfully through a crisis is given by $e^{\frac{\beta}{\alpha}\left(\mathscr{R}_{j}-\mathscr{R}_{i}\right)}$. We compute these ratios for different couples of corporates operating in a similar sector, a ``good'' one and a ``bad'' one (\cref{fig:ProbRat}).

\begin{figure}[H]
	\centering
	\subfloat[]{\includegraphics[width=0.45\textwidth]{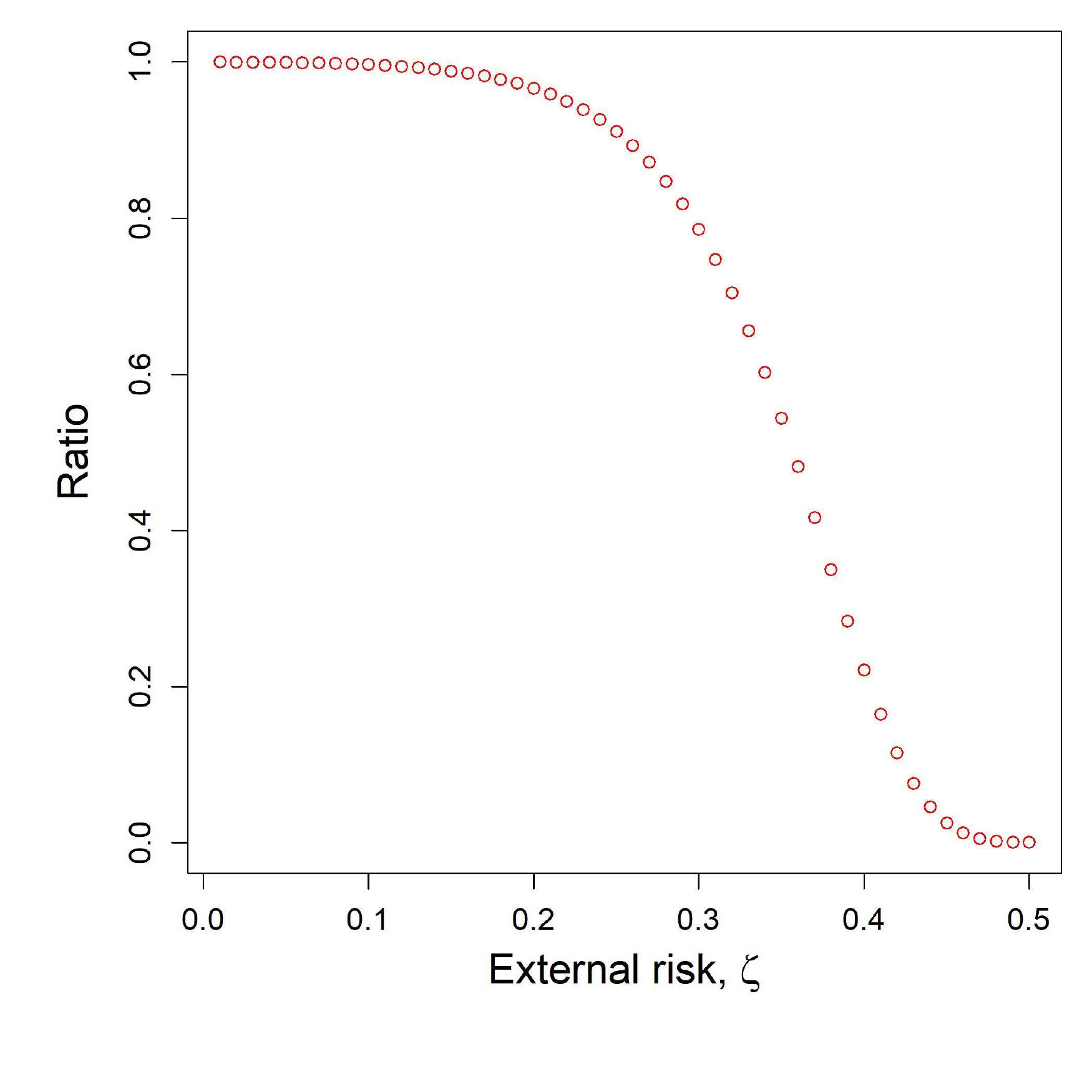}
		
	}\subfloat[]{\includegraphics[width=0.45\textwidth]{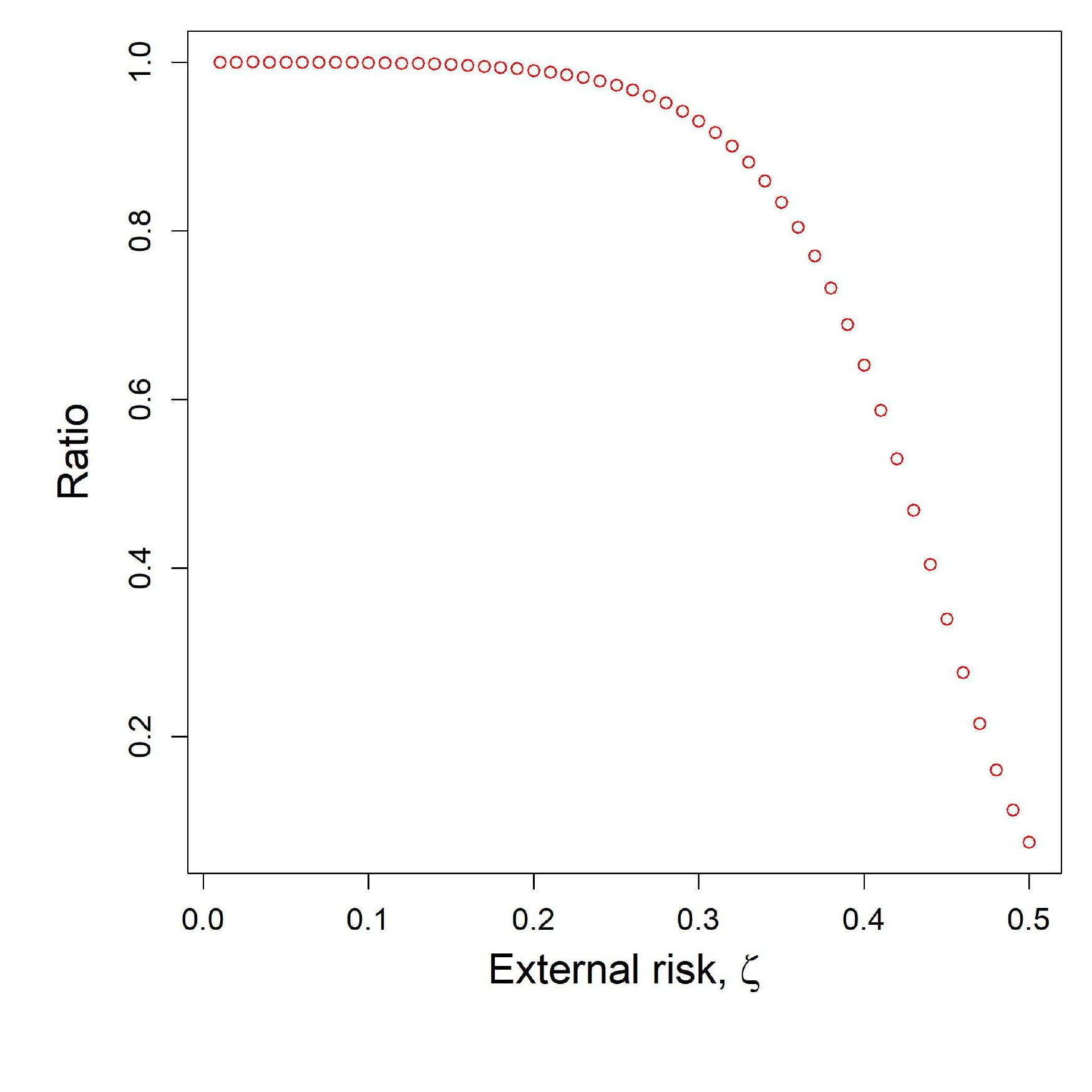}
		
	}
	
	\caption{Figures display the ratios between the probabilities of not being
		infected by a crisis for two different couples of Corporates: a) Lucent
		Technologies Inc. over General Electric Co. b) Morgan Stanley Co.
		over Bank One Corp.  It is noteworthy that Lucent
		Technologies Inc. and Morgan Stanley Co. reduced their rankings over time, while General Electric Co. and  Bank One Corp. increased their rankings.} 	\label{fig:ProbRat} 
\end{figure}

As expected, at low $\zeta$ the probability of not being infected by a crisis is the same for both high and low risk-centrality companies. But this ratio decreases very quickly as $\zeta$ increases and this means that for companies that reduced their risk (e.g., Lucent Technologies, Morgan Stanley, Union Carbide and American Express) the probabilities to stay safe during a crisis are very small if compared with the analogous probabilities for companies that increased their risk (e.g., General Electric, Bank One, Ashland and Bank of America).

\section{Ranking interlacement}
\label{sec:Ranking}
During the analysis of the two real-world networks studied above,
we have noticed that with the change of $\zeta$ some nodes vary
their ranking significantly, to the point of changing their  positions relative
to each other. For instance, in \cref{Interlacement} we illustrate
six pairs of corporates that interlace their positions with the change
of the global infectability in the network. In the first pair, \cref{Interlacement}(a), we see that at low levels of infectability, i.e., $\zeta\rightarrow0$, J.P. Morgan\&Co Inc. (red) occupies a position in the ranking of $\mathscr{C}_{i}$ more
at the bottom than Bank of America Corp. (blue). That is, at low global infectability
J.P. Morgan\&Co is exposed to less risk than Bank of America. However,
when the global infectability in the network increases ($\zeta\rightarrow1$),
Bank of America is exposed to less risk than J.P. Morgan\&Co. A similar
interlacement is observed between the other couples in \cref{Interlacement}. For instance, in \cref{Interlacement}(f), the interlacement between rankings for General Motors Corp. (red) and Boeing Co. (blue) occurs at
a smaller value of $\zeta$ than for the previous cases. Before proceeding
with the analysis of this phenomenon, we would like to remark that
the existence of ranking interlacement means that the ranking of the
nodes in a network based on the risk-dependent centralities is not
unique and fixed as in the case of other classical centrality measures,
e.g., degree, eigenvector, closeness, betweenness. Here instead the
ranking of nodes depends on the global external conditions to which
the network is submitted.

\begin{figure}[H]
	\begin{centering}
		\subfloat[]{\begin{centering}
				\includegraphics[scale=0.33]{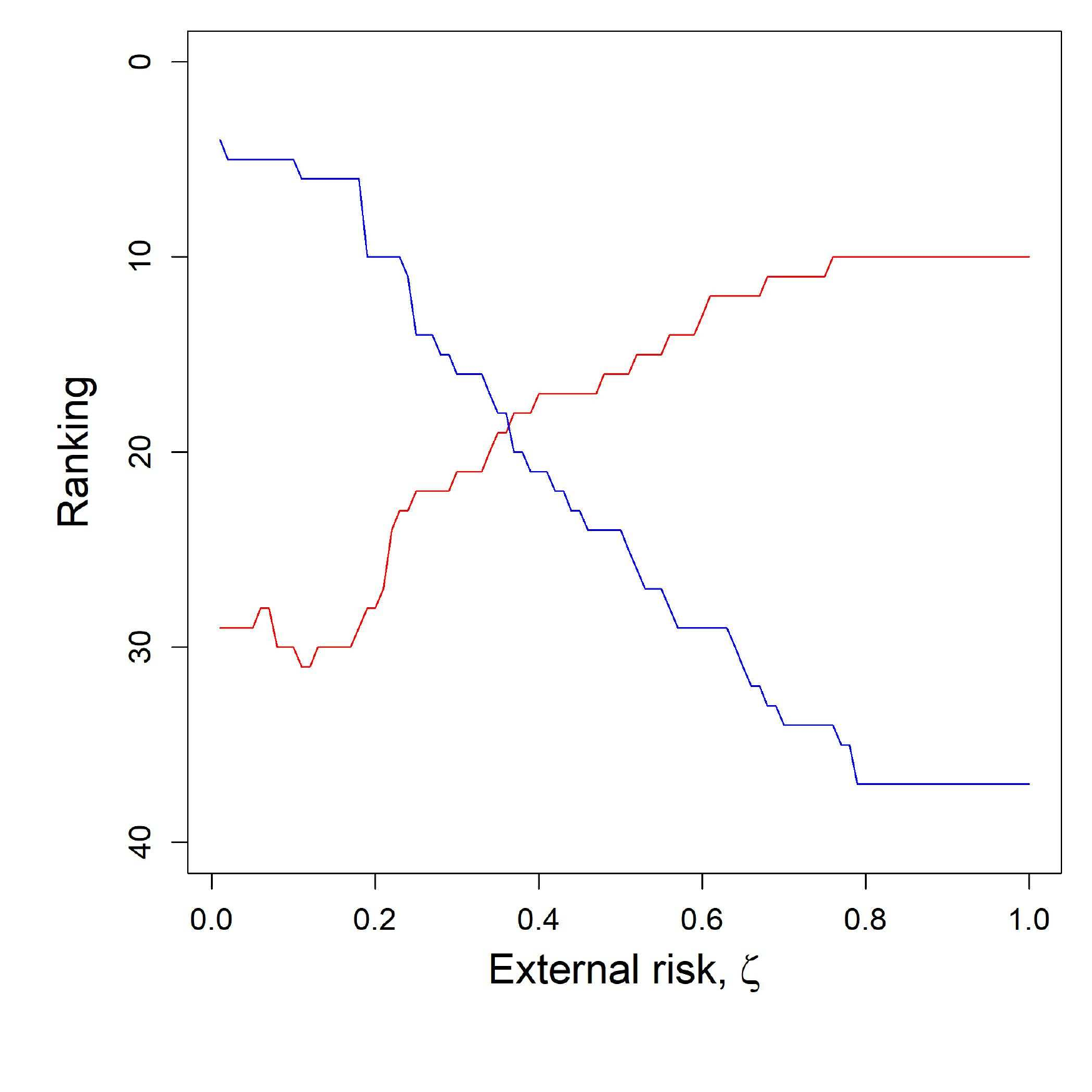}
				\par\end{centering}
			
		}\subfloat[]{\begin{centering}
				\includegraphics[scale=0.33]{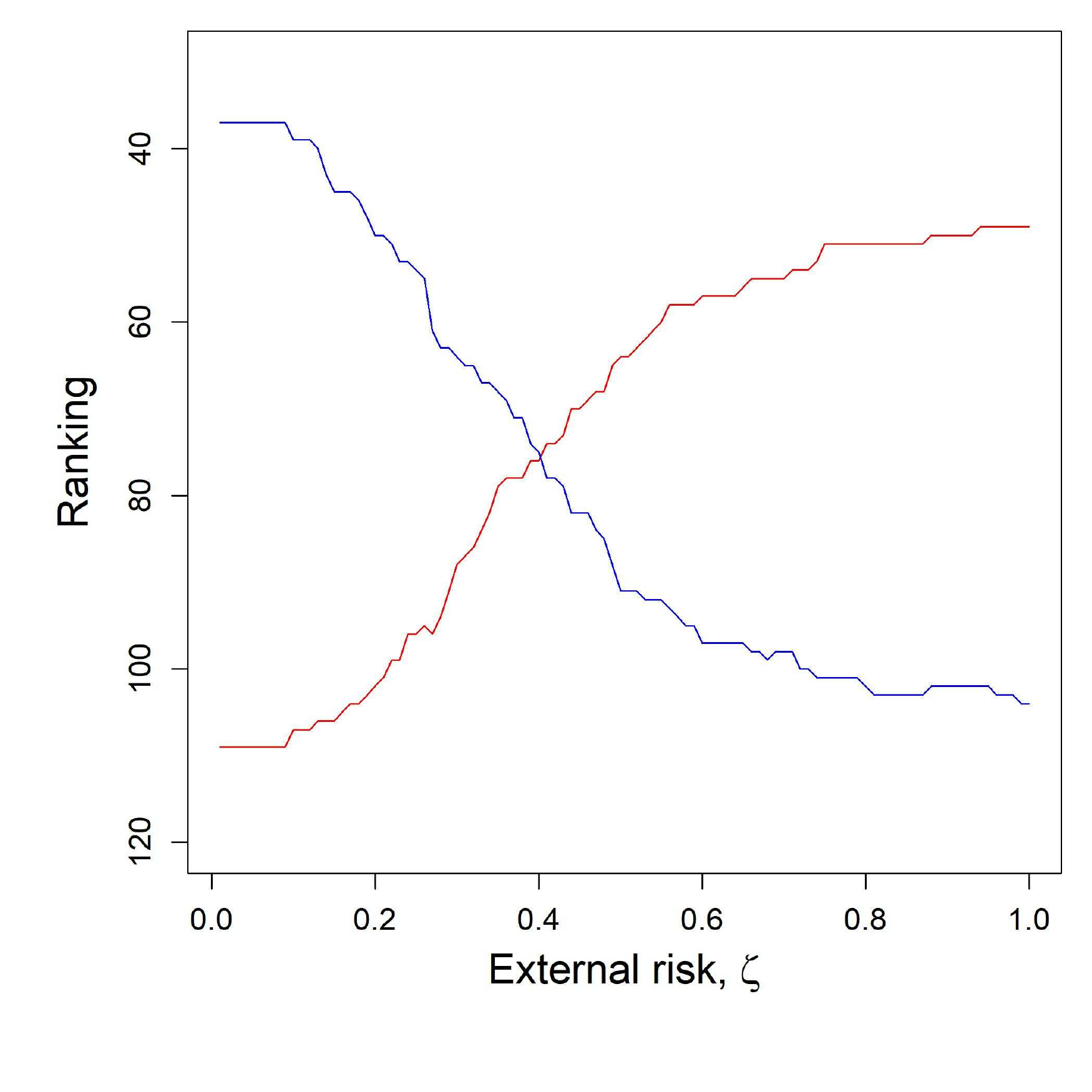}
				\par\end{centering}
		}
		\par\end{centering}
	\begin{centering}
		\subfloat[]{\begin{centering}
				\includegraphics[scale=0.33]{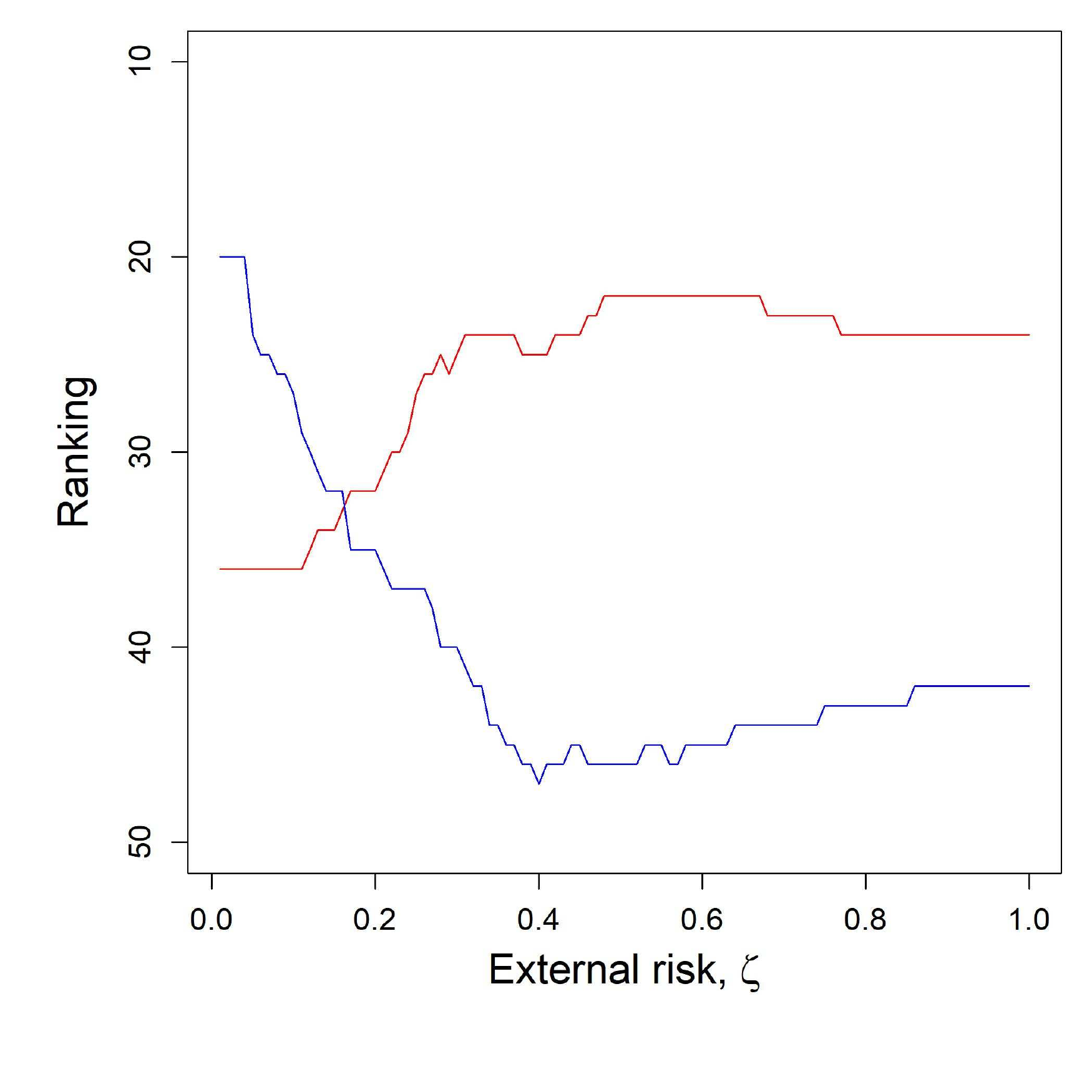}
				\par\end{centering}
		}\subfloat[]{\begin{centering}
				\includegraphics[scale=0.33]{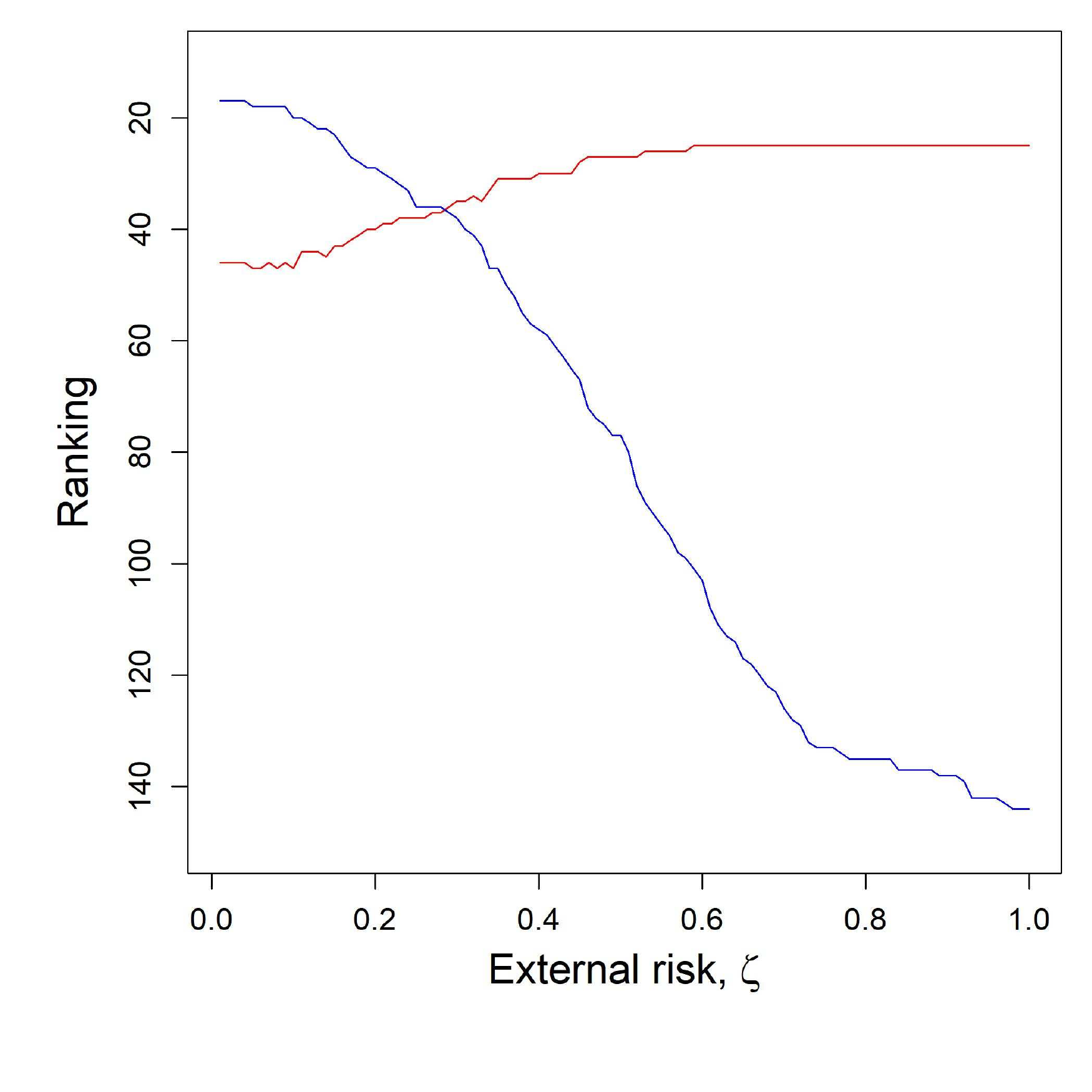}
				\par\end{centering}
		}
		\par\end{centering}
	\begin{centering}
		\subfloat[]{\begin{centering}
				\includegraphics[scale=0.33]{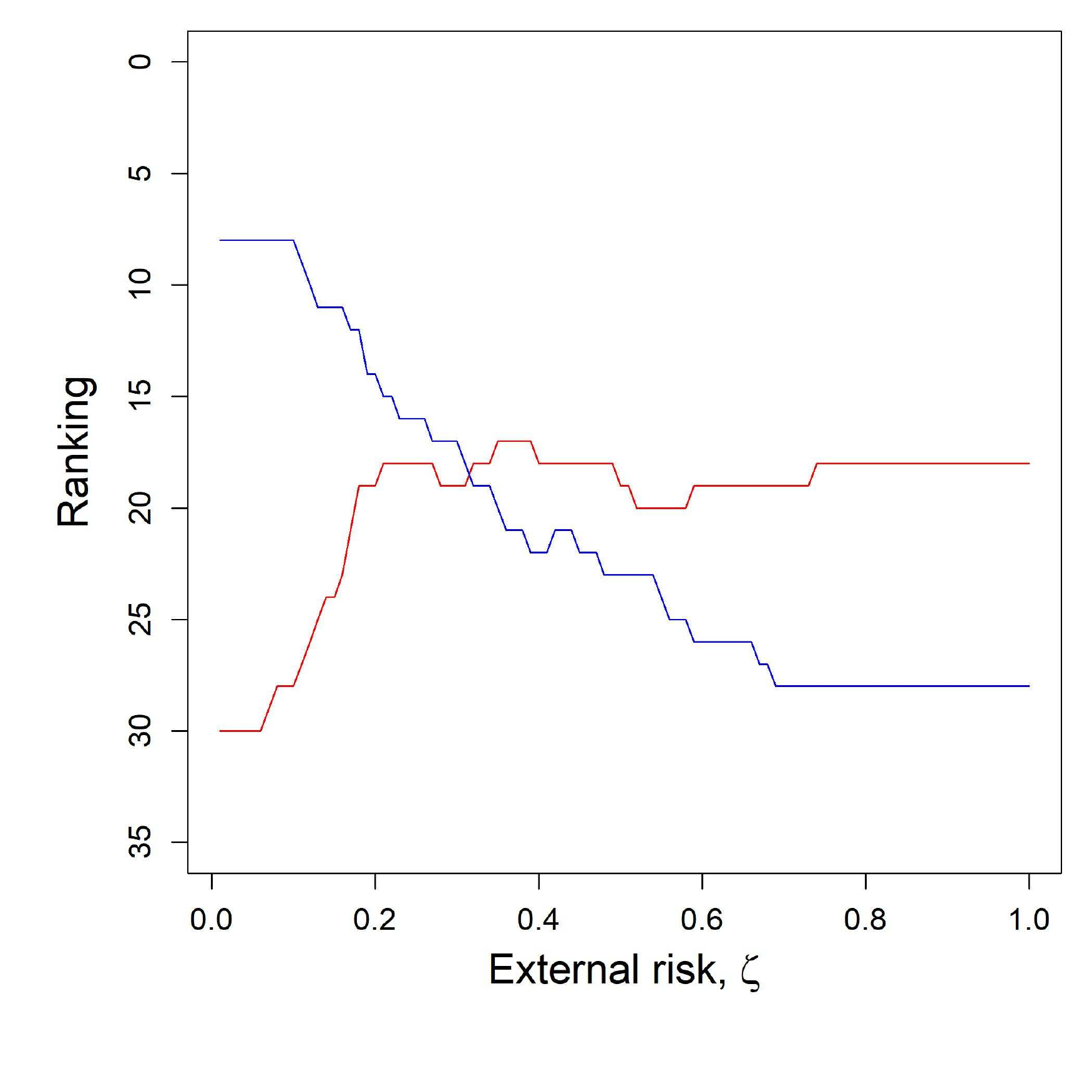}
				\par\end{centering}
		}\subfloat[]{\begin{centering}
				\includegraphics[scale=0.33]{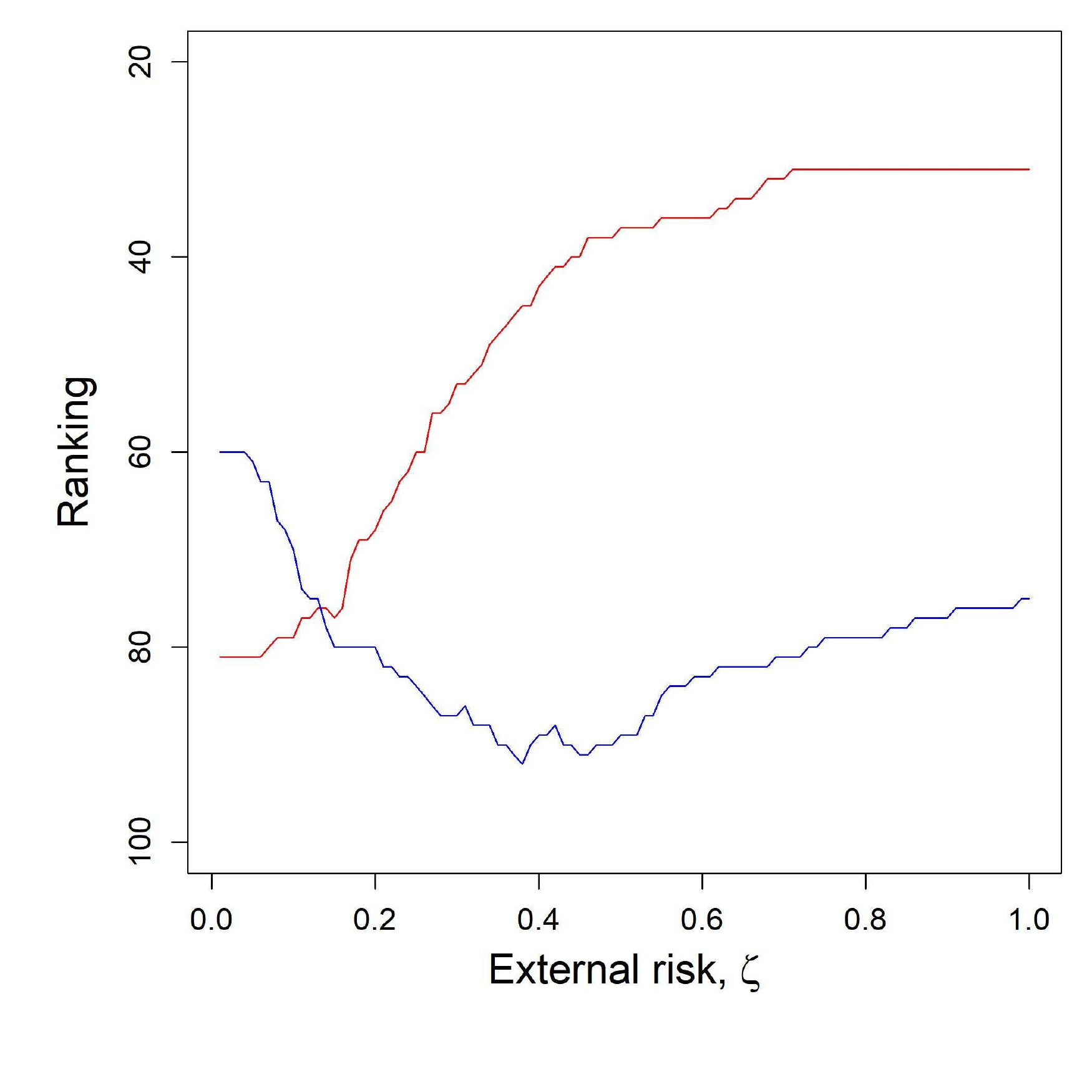}
				\par\end{centering}
		}
		\par\end{centering}
	\caption{Illustration of the Circulability Ranking Interlacement for a) J.P.
		Morgan\&Co Inc. (red) and Bank of America Corp. (blue) b) Pfizer Inc.
		(red) and Ashland Inc (blue) c) Morgan Stanley \& Co. (red) and Bank
		One Corp. (blue) d) AT\&T Corp. (red) and Airtouch Communications
		Inc. (blue) e) Union Carbide Corp. New (red) and AON Corp. (blue)
		f) General Motors Corp. (red) and Boeing Co. (blue)}
	
	\label{Interlacement}
\end{figure}

In order to shed light on the issue of ranking interlacement we will
make use of different representations of the risk-dependent total
communicability ${\mathscr{R}}_{i}(\zeta)$ and circulability ${\mathscr{C}}_{i}(\zeta)$
measures (the transmissibility is obtained as the difference of these
two and can be treated accordingly). First, expanding the matrix exponential
in a power series gives the representation 
\begin{equation}
	{\mathscr{R}}_{i}(\zeta)=\left(e^{\zeta A}{\vec 1}\right)_{i}=\sum_{k=0}^{\infty}\frac{\zeta^{k}}{k!}w_{i}^{(k)},\label{series1}
\end{equation}
where $w_{i}^{(k)}=\left(A^{k}{\vec 1}\right)_{i}$ denotes the number
of walks of length $k$ starting from the node $i$, with $w_{i}^{(0)}=1$.
In particular, $w_{i}^{(1)}=k_{i}$, the degree of node $i$. Similarly,
\begin{equation}
	{\mathscr{C}}_{i}(\zeta)=\left(e^{\zeta A}\right)_{ii}=\sum_{k=0}^{\infty}\frac{\zeta^{k}}{k!}w_{i,i}^{(k)},\label{series2}
\end{equation}
where now $w_{i,i}^{(k)}=\left(A^{k}\right)_{ii}$ is the number of
closed walks of length $k$ through node $i$; in particular, $w_{i,i}^{(0)}=1$,
$w_{i,i}^{(1)}=0$, $w_{i,i}^{(2)}=k_{i}$, and $w_{i,i}^{(3)}=2t_{i}$,
where $t_{i}$ is the number of triangles node $i$ participates in.

Second, we recall that the spectral theorem yields the formulas 
\begin{equation}
	{\mathscr{R}}_{i}(\zeta)=\sum_{k=1}^{n}e^{\zeta\lambda_{k}}\left(\psi_{k}^{T}{\vec 1}\right)\psi_{k,i},\quad{\mathcal{C}}_{i}(\zeta)=\sum_{k=1}^{n}e^{\zeta\lambda_{k}}\left(\psi_{k,i}\right)^{2}.\label{spec}
\end{equation}

Using \cref{series1}-\cref{series2}, we readily see that both functions
of $\zeta$ are {\em absolutely monotonic} for $\zeta>0$, i.e.
they are positive and infinitely differentiable on $(0,\infty)$,
with all the derivatives being nonnegative. In particular, both functions
are strictly increasing and strictly convex.
\begin{definition}
	\label{Def1} We say that the rankings of node $i$ and node $j$
	based on the circulability \textit{interlace at }$\zeta^{*}>0$ if
	${\mathscr{C}}_{i}(\zeta^{*})={\mathscr{C}}_{j}(\zeta^{*})$ and there
	exists an $\varepsilon>0$ such that ${\mathscr{C}}_{i}(\zeta)-{\mathscr{C}}_{j}(\zeta)$
	changes sign exactly once in $(\zeta^{*}-\varepsilon,\zeta^{*}+\varepsilon)$.
\end{definition}

In other words, nodes $i$ and $j$ interlace at $\zeta^* > 0$ if the plots of ${\mathscr C}_i (\zeta)$ and ${\mathscr C}_j (\zeta)$ cross for $\zeta = \zeta^*$. We note that, in principle, it is possible to have ${\mathscr{C}}_{i}(\zeta^{*})={\mathscr{C}}_{j}(\zeta^{*})$ for some value of $\zeta^*$ without interlacing taking place. Two cases are possible: in the first one, the two curves touch at the isolated point $\zeta^*$ (without crossing), and in the second one the two functions are identical on an open neighborhood of $\zeta^{*}$ and, therefore, for all $\zeta$ since they are analytic functions. In practice, either scenario is very unlikely to occur, at least for real world networks. Note that points of tangency must satisfy the additional condition ${\mathscr{C}}_{i}'(\zeta^{*})={\mathscr{C}}_{j}'(\zeta^{*})$.

An analogous definition can be given for the ranking based on other
$\zeta$-dependent measures, like the total communicability ${\mathscr{R}}_{i}(\zeta)$.
In the following we limit our discussion to the interlacing of rankings
according to the circulability, but analogous observations hold for
the total communicability and transmissibility functions.

Identifying the interlacing points (if they exist) requires to find
the roots of the transcendental equation
${\mathscr{C}}_{i}(\zeta)-{\mathscr{C}}_{j}(\zeta)=0$,
or\
\[
\Psi(\zeta):=\sum_{k=1}^{n}e^{\zeta\lambda_{k}}\left[\psi_{k,i}^{2}-\psi_{k,j}^{2}\right]=0.
\]

Even if we knew the eigenvalues and eigenvectors of $A$ explicitly,
there is no general closed form expression for the roots of the transcendental
function $\Psi$. Of course one could resort to numerical root-finding
techniques, but this would be impractical for large networks. Here
and below we give a qualitative discussion followed by a heuristic
approach that yields approximations that seem to work well in practice.

\noindent We begin with the following result. It applies to both circulability
and total communicability-based rankings, and in fact for a much larger
class of parameter-dependent centrality ranking functions, including
Katz centrality \cite{Katz}. We remind the reader that we restrict
the risk rate $\zeta$ to positive values.
\begin{theorem}
	\label{thm_finite} Let $i$ and $j$ be two nodes with different eigenvector centrality: $\psi_{1,i} \ne \psi_{1,j}$. Then the number of interlacing points for $i$ and $j$ is  necessarily finite (possibly zero). 
\end{theorem}

\begin{proof}
	Let us assume that there is at least one pair of nodes, $i$ and $j$,
	whose rankings interlace, so that $\Psi(\zeta)=0$ has at least
	one positive root. Observe that the ranking of node $i$ provided by ${\mathscr{C}}_{i}(\zeta)$ is identical to that obtained using
	\[
	\hat{\mathscr{C}}_{i}(\zeta)=e^{-\zeta\lambda_{1}}{\mathscr{C}}_{i}(\zeta)=\psi_{1,i}^{2}+\sum_{k=2}^{n}e^{\zeta(\lambda_{k}-\lambda_{1})}\psi_{k,i}^{2}.
	\]
	As this quantity tends monotonically to $\psi_{1,i}^{2}$ for $\zeta\to\infty$,
	there exists a $\bar{\zeta}$ such that no rank interlacing with node $j$ can occur for $\zeta>\bar{\zeta}$, since all the node rankings must stabilize
	on the eigenvector rankings in the large $\zeta$ limit. 
	Hence, all interlacing points must fall within
	the compact interval $[0,\bar{\zeta}]$. Suppose that the number of
	interlacing points is infinite. By the Bolzano-Weierstrass Theorem,
	this set has a point of accumulation. But since $\hat{\Psi}(\zeta):=e^{-\zeta\lambda_{1}}\Psi(\zeta)$
	is analytic, and zero on this set, it must be identically zero everywhere,
	which contradicts the assumption that there is at least one interlacing
	point in $(0,\infty)$.
\end{proof}

As a consequence:

\begin{corollary}
	If all nodes in the network have different eigenvector centralities, the
	total number of interlacing points is finite (possibly zero).
\end{corollary} 

\color{black}


A sufficient condition for the existence of at least one interlacing point for the pair of nodes $i$ and $j$ is that $k_{i}\ge k_{j}$ (or $k_{j}\ge k_{i}$) while $\psi_{1,i}<\psi_{1,j}$ (resp., $\psi_{1,i}>\psi_{1,j}$). This follows from Theorem \ref{thm2}: since ${\mathscr{C}}_{i}(\zeta)$ interpolates smoothly between degree centrality and eigenvector centrality, the only way that a node with higher degree
can have lower eigenvector centrality than another node is that the corresponding circulabilities interlace at some  value $\zeta^* > 0$. If more than one interlacing point exists, this number must be odd, for otherwise the node with higher degree would also have higher eigenvector centrality than the other node. That the above condition is not necessary is made clear considering the possibility of an even number of interlacing points. A necessary condition for the existence of at least one interlacing point is that there exist at least two values of $k$, say $k_{1}$ and $k_{2}$,
for which $(A^{k_{1}})_{ii}-(A^{k_{1}})_{jj}$ and $(A^{k_{2}})_{ii}-(A^{k_{2}})_{jj}$
have different sign. Indeed, it is obvious from \cref{series1}-\cref{series2}
that if (say) $(A^{k})_{ii}\ge(A^{k})_{jj}$ for all $k$, then no
rank interlacing point exists. That this condition may not be sufficient
is suggested by the fact that the series expansions contain an infinity
of terms.

We mention that the same problem has been studied, for a different centrality function (the Katz resolvent),
by \cite{Kloster2019} independently of us.

\subsection{A back of envelop approach}

We now consider heuristics based on truncated series expansions. Let
$k_{0}\ge3$ be the smallest value of $k$ such that the sequence
of values $\{(A^{k})_{ii}-(A^{k})_{jj}\}_{k\ge2}$ undergoes a sign
change (here zero is considered positive). If no such $k_{0}$ exists,
then no interlacing can take place, as we already observed. We consider
approximating ${\mathscr{C}}_{i}(\zeta)$ with its truncation to an
order $k\ge k_{0}$: 
\begin{equation}
	{\mathscr{C}}_{i}(\zeta)\approx1+\frac{1}{2!}\zeta^{2}w_{i,i}^{(2)}+\frac{1}{3!}\zeta^{3}w_{i,i}^{(3)}+\cdots+\frac{1}{k!}\zeta^{k}w_{i,i}^{(k)}=\tilde{\mathscr{C}}_{i}(\zeta),\label{approx}
\end{equation}
where we recall that $w_{i,i}^{(k)}=(A^{k})_{ii}$. We emphasize that
this polynomial approximation assumes that $\zeta$ is small, since
the error in it is $O(\zeta^{k+1})$. In alternative we can also use
as a surrogate for ${\mathscr{C}}_{i}$ the same polynomial shifted
by 1 and divided by $\zeta^{2}$: 
\[
\frac{\tilde{\mathscr{C}}_{i}(\zeta)-1}{\zeta^{2}}=\frac{1}{2!}w_{i,i}^{(2)}+\frac{1}{3!}\zeta w_{i,i}^{(3)}+\cdots+\frac{1}{k!}\zeta^{k-2}w_{i,i}^{(k)},
\]
where now the error is $O(\zeta^{k-1})$. We can now use these polynomial
approximations to try to locate, approximately, any interlacing points
sufficiently small in magnitude. This requires finding the (positive)
roots, if any, of the polynomial equation of degree $k-2$: 
\begin{equation}
	q(\zeta)=\frac{(w_{i,i}^{(k)}-w_{j,j}^{(k)})}{k!}\zeta^{k-2}+\frac{(w_{i,i}^{(k-1)}-w_{j,j}^{(k-1)})}{(k-1)!}\zeta^{k-3}+\cdots+\frac{(w_{i,i}^{(3)}-w_{j,j}^{(3)})}{3!}\zeta+\frac{(w_{i,i}^{(2)}-w_{j,j}^{(2)})}{2!}=0.\label{poly2}
\end{equation}
It is well known that for degree greater than or equal to 5 there
is no closed form expression of the solutions of an algebraic equation
involving only arithmetic operations and root extractions, so in general
if $k\ge7$ we will have to resort to numerical methods for solving
\cref{poly2}. Evaluation of the coefficients requires computing
the diagonal entries of powers of the adjacency matrix $A$, which
can be expensive for very large graphs and large values of $k$.

As the simplest possible example, we consider the case where $w_{i,i}^{(2)}>w_{j,j}^{(2)}$
and $w_{i,i}^{(3)}<w_{j,j}^{(3)}$ (or vice-versa), i.e., $k_{0}=3$.
Taking $k=k_{0}$, equation \cref{poly2} becomes the linear equation
\[
\frac{(w_{i,i}^{(3)}-w_{j,j}^{(3)})}{3!}\zeta+\frac{(w_{i,i}^{(2)}-w_{j,j}^{(2)})}{2!}=0,
\]
which admits the unique solution $\zeta^{*}=\frac{3(w_{i,i}^{(2)}-w_{j,j}^{(2)})}{w_{i,i}^{(3)}-w_{j,j}^{(3)}}$,
which is of course positive. In terms of the degree of the nodes and
the number of triangles in which they take place, this can be written
in the form:

\begin{equation}
	\zeta^{*}=\frac{3}{2}\left|\frac{k_{i}-k_{j}}{t_{i}-t_{j}}\right|.\label{eq:Bound-1}
\end{equation}

In the case of weighted networks, the degree is replaced by the weighted
degree or strength, and the number of triangles is replaced by the
weighted number of cycles of length 3, i.e., the weight of a cycle
of length 3 is the product of the weights at its three edges. A priori,
there is no reason to expect that this value is close to an actual
interlacing point (assuming it even exists), since the behavior of
higher order terms may more than offset the influence of the negative
term involving $t_{i}-t_{j}$. Better approximations might be obtained
by considering higher order approximations; for example using $k=4$
leads to an easily solved quadratic equation in $\zeta$, $k=5$ leads
to a cubic, and so forth. In any case, these are heuristics whose
usefulness can only be assessed experimentally on concrete examples.
We emphasize that the use of power series truncation requires knowledge
of $k_{0}$, since truncating the series at orders lower than $k_{0}$
would lead to an equation devoid of positive solutions and therefore
to concluding that no interlacing points exist for a given pair of
nodes, even if such points do exist.

It is also worth recalling Descartes's Rule of Signs, according to
which the number of positive real roots of a polynomial (counted with
their multiplicities) is equal to the number of sign changes in the
(nonzero) coefficients or less than that by an even whole number,
when the powers are ordered in descending order. If, moreover, the polynomial
is known to have only real roots (as in the case of a symmetric adjacency matrix,
i.e., of undirected networks) then the number of sign changes is exactly equal
to the number of positive roots.
It is then obvious
that if the power series is truncated at order $k_{0}$, i.e., as
soon as we observe the first sign change in the coefficients, then
there will be exactly one positive root and therefore only one (approximate)
interlacing point can be found by this method. 
A polynomial truncation of higher
degree $k>k_{0}$ may have more than one positive root, depending
on the number of changes in the coefficients (assuming the network is undirected).
We will come back to this case shortly.

To exemplify the previous finding let us consider a pair of nodes
with a small difference in their degree, e.g., $k_{i}-k_{j}=2$, then
$-\left(k_{i}-2\right)^{2}\leq\left(t_{i}-t_{j}\right)\leq k_{i}^{2}$,
such that if, for instance, $k_{i}\leq10$ and we let $\zeta$ vary from $0$ to $0.1$
we obtain the plot given in \cref{surfaces}(a). As can be seen
there are certain values of $\varDelta=t_{i}-t_{j}<0$ for which we
can obtain positive and negative values of $\mathscr{C}_{i}-\mathscr{C}_{j}$.
This is illustrated in \cref{surfaces}(b) where we can see
that when $-100\leq\varDelta\leq-40$ there are both positive and
negative values of $\mathscr{C}_{i}-\mathscr{C}_{j}$. In other words,
it is possible to find pairs of nodes for which  $\mathscr{C}_{i}\left(\zeta_{1}\right)>\mathscr{C}_{j}\left(\zeta_{1}\right)$
and then $\mathscr{C}_{i}\left(\zeta_{2}\right)<\mathscr{C}_{j}\left(\zeta_{2}\right)$,
which means that these nodes will change their ranking position in
terms of the risk-dependent centrality when the values of $\zeta$
change even for a relatively narrow window. Notice that if $k_{i}-k_{j}=2$,
and $\varDelta\geq-30$ such change is not observed for the corresponding
range of $\zeta$ analyzed.

\begin{figure}[H]
	\begin{centering}
		\subfloat[]{\begin{centering}
				\includegraphics[width=0.45\textwidth]{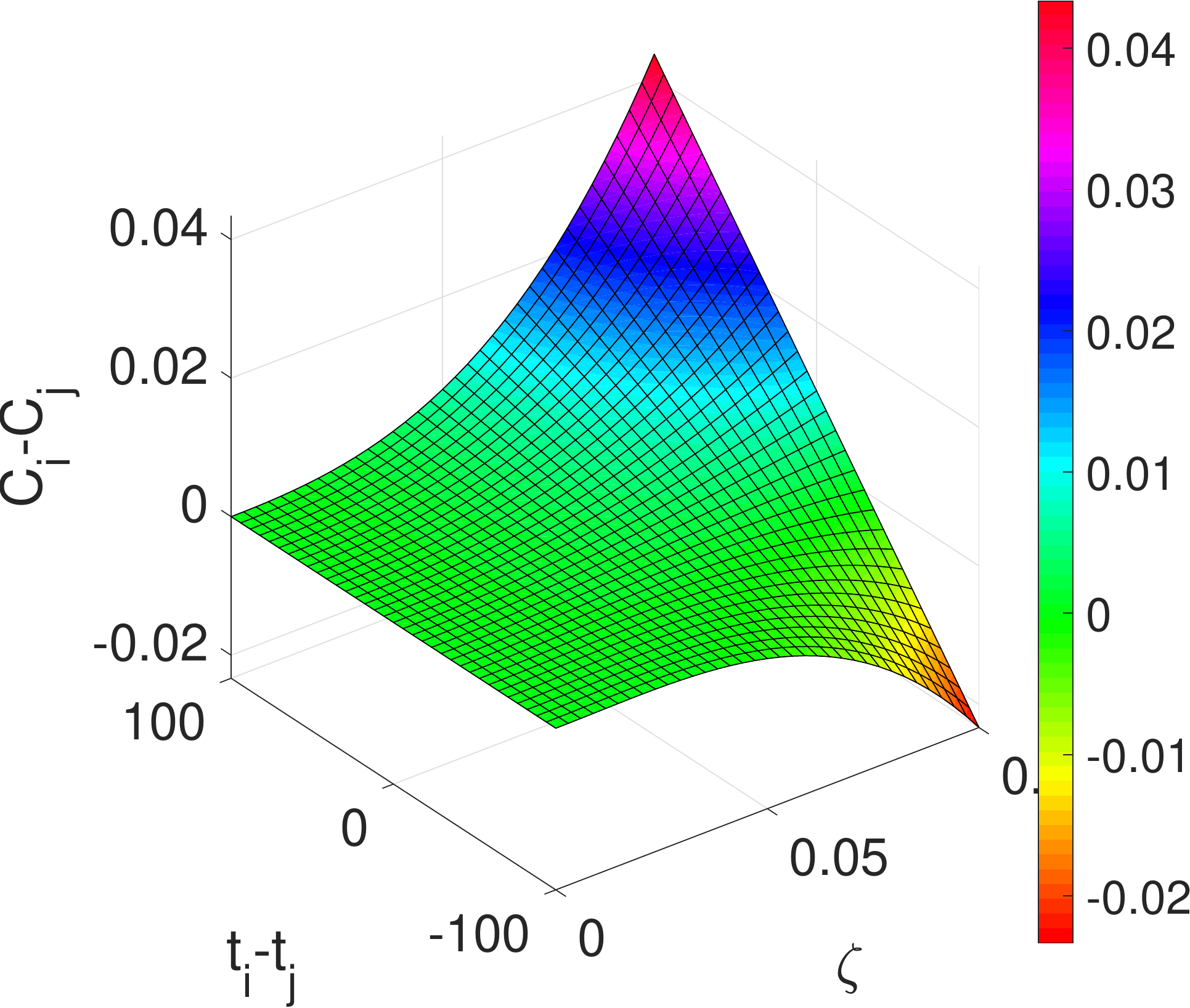}
				\par\end{centering}
		}\subfloat[]{\begin{centering}
				\includegraphics[width=0.45\textwidth]{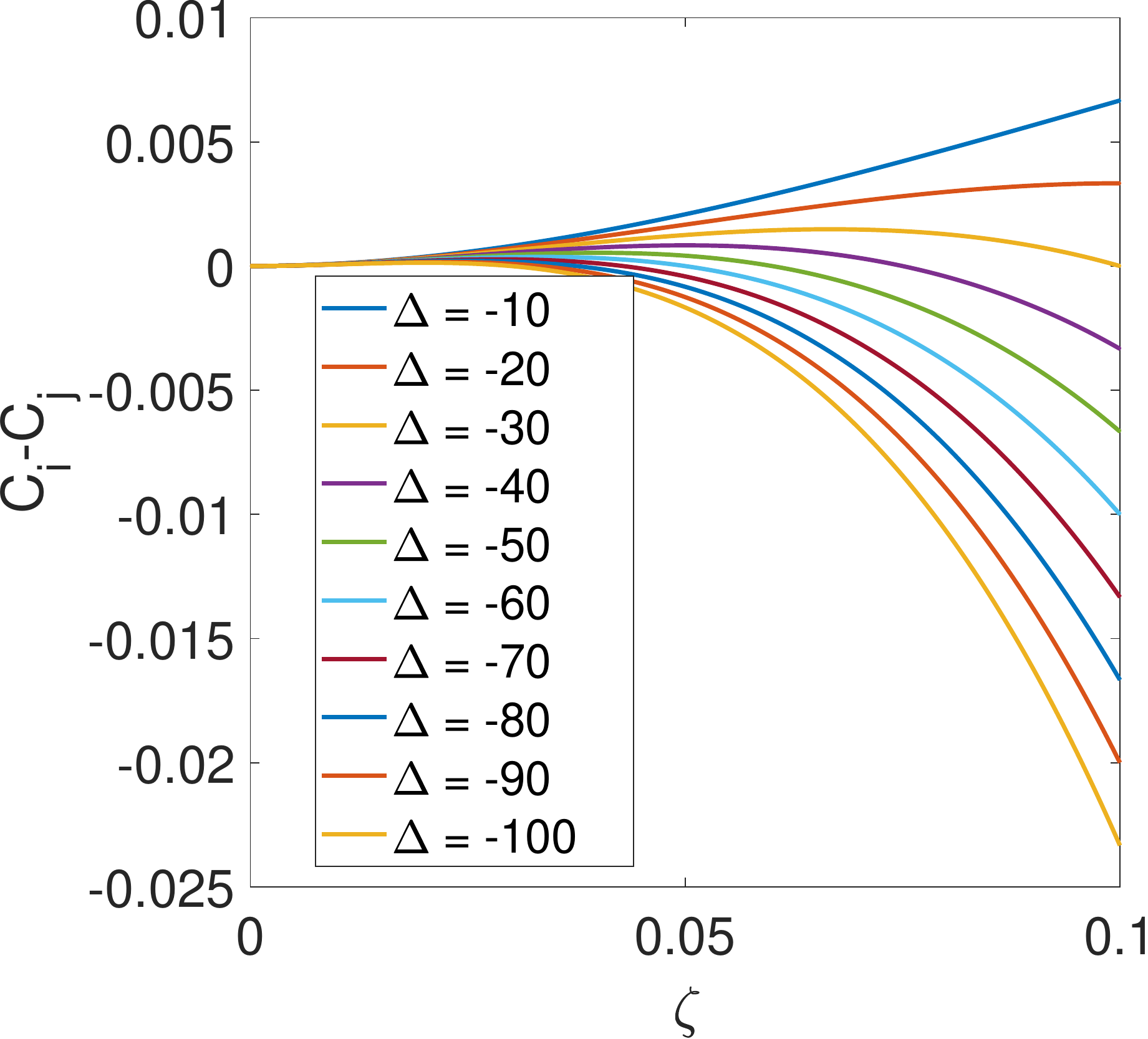}
				\par\end{centering}
		}
		\par\end{centering}
	\caption{(a) Illustration of the change in the difference in the risk-dependent
		centrality of nodes having a small difference in degrees, $k_{i}-k_{j}=2$,
		as a function of the difference in the number of triangles, $t_{i}-t_{j}$,
		and of the network infectivity risk $\zeta$. (b) Some of the curves
		obtained for $k_{i}-k_{j}=2$ and a given value of $\varDelta=t_{i}-t_{j}$
		as a function of $\zeta$.}
	
	\label{surfaces}
\end{figure}

If we now consider a large difference in the node degrees, e.g., $k_{i}-k_{j}=100$,
and the same range of change for the difference in the number of triangles,
e.g., $-100\leq\varDelta\leq 100$ we do not observe any variation
in the ranking of pairs of nodes as can be seen in \cref{surfaces_2}(a).
In this case the range of $\varDelta$ must be increased dramatically
to obtain inversions in the ranking of pairs of nodes (see \cref{surfaces_2}(b)).

\begin{figure}[H]
	\begin{centering}
		\subfloat[]{\begin{centering}
				\includegraphics[width=0.45\textwidth]{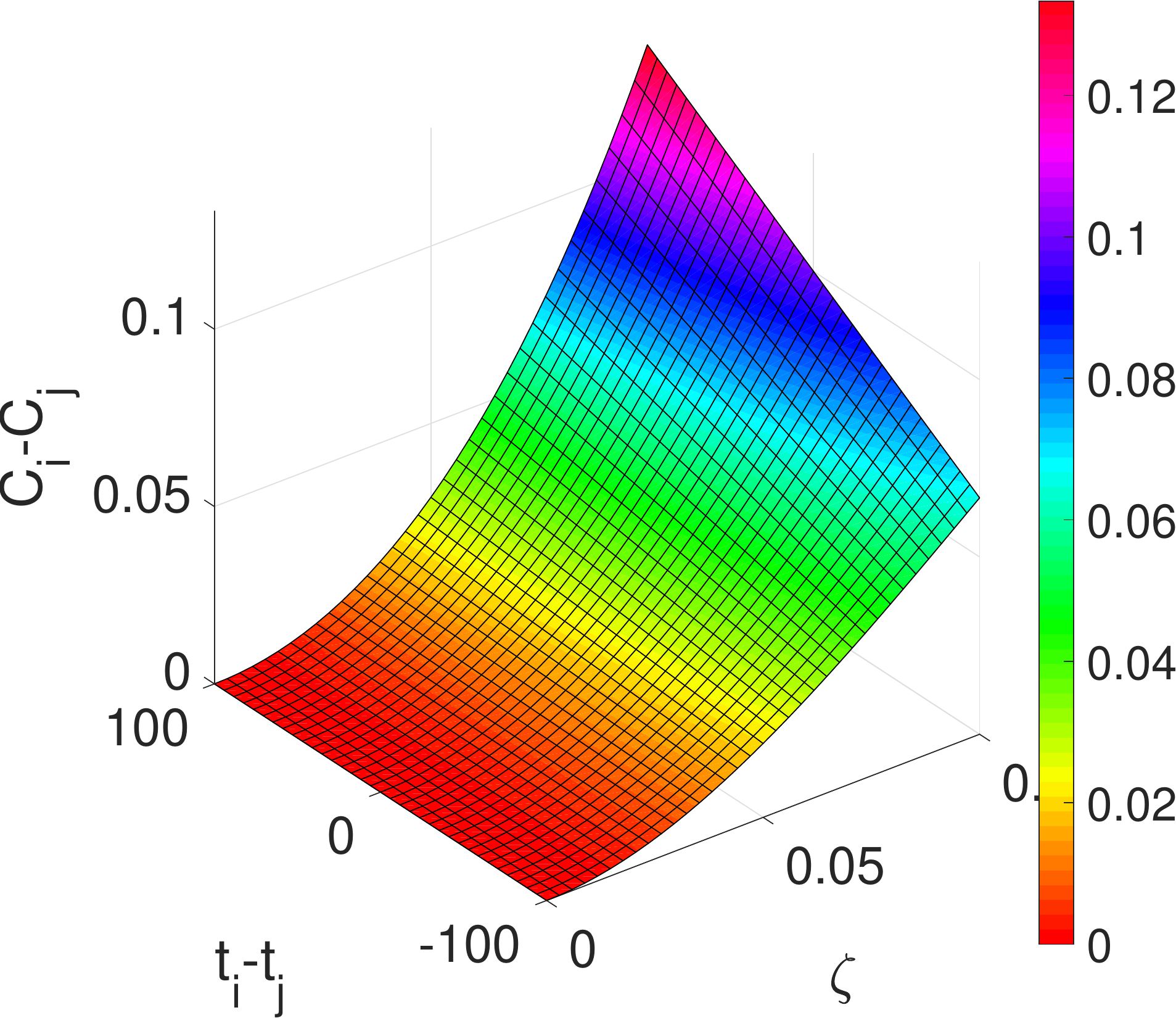}
				\par\end{centering}
		}\subfloat[]{\begin{centering}
				\includegraphics[width=0.45\textwidth]{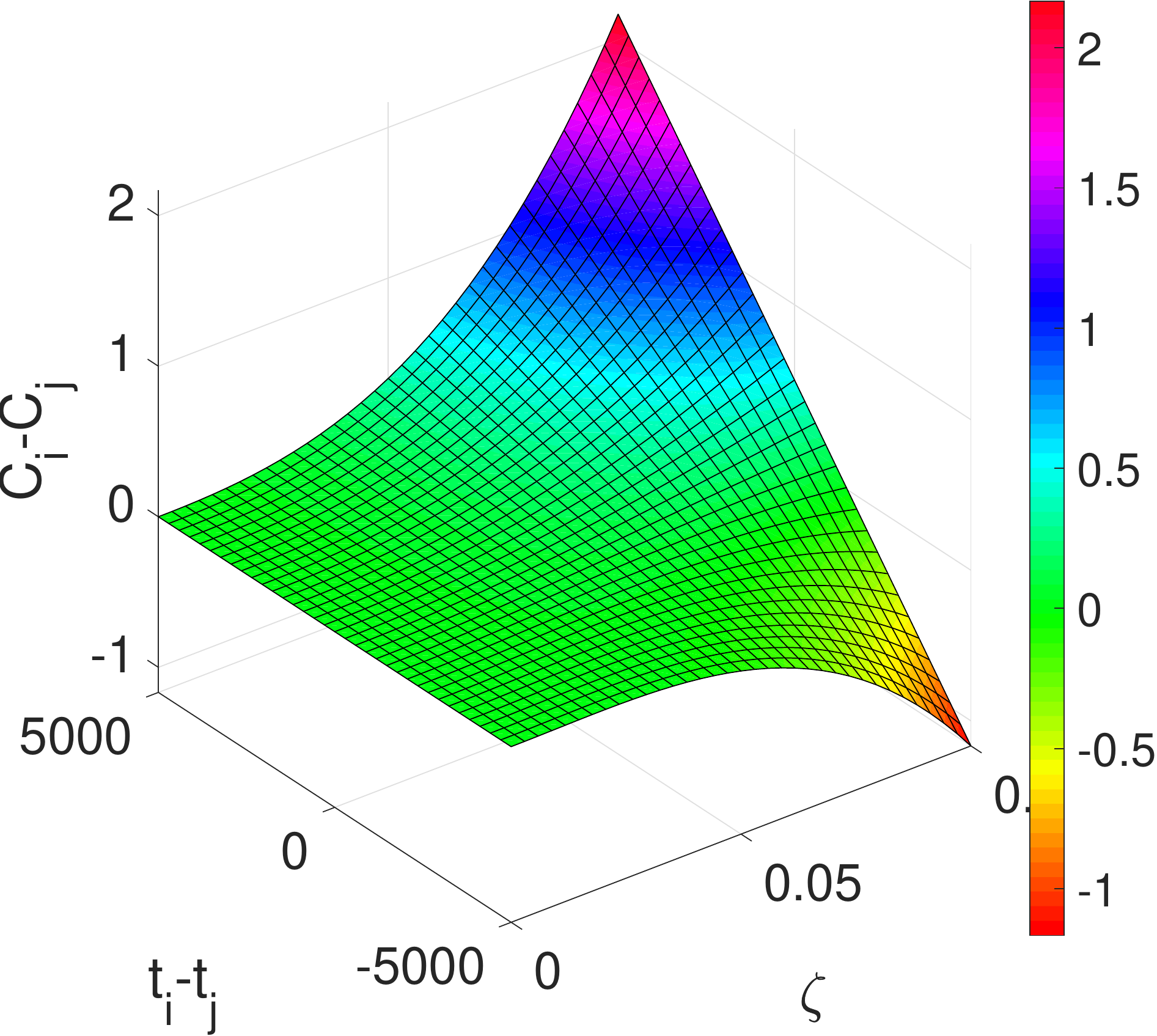}
				\par\end{centering}
		}
		\par\end{centering}
	\caption{Illustration of the change in the difference in the risk-dependent
		centrality of nodes having a small difference in degrees, $k_{i}-k_{j}=100$,
		as a function of the difference in the number of triangles, $-100\protect\leq\varDelta\protect\leq 100$
		(a) and $-5000\protect\leq\varDelta\protect\leq 5000$ (b), and of
		the network infectivity (risk) $\zeta$.}
	
	\label{surfaces_2}
\end{figure}

To illustrate how well the estimate  \cref{eq:Bound-1} performs, we
use it for approximating the interlacement point for several pairs of
corporates and compare them with the observed values in \cref{TabR-1}
for the weighted version of the US corporate network.

\begin{table}[!h]
	\small
	\caption{Calculation of the crossing point of ranking interlacement
		for several pairs of corporates in the US corporates network of 1999
		as well as the observed values at which such interlacements occur.}
	\centering{} %
	\begin{tabular}{|c|l|l|c|c|}
		\hline 
		\textbf{Plot} & \textbf{Corporate 1} & \textbf{Corporate 2} & \textbf{$\zeta^{*}$ calculated} & \textbf{$\zeta^{*}$ observed}\tabularnewline
		\hline 
		\hline 
		(a) & J.P. Morgan\&Co Inc. & Bank of America Corp. & 0.375 & 0.37\tabularnewline
		\hline 
		(b) & Pfizer Inc. & Ashland Inc. & 0.441 & 0.41\tabularnewline
		\hline 
		(c) & Morgan Stanley \& Co. & Bank One Corp. & 0.176 & 0.17\tabularnewline
		\hline 
		(d) & AT\&T Corp. & Airtouch Communications  & 0.273 & 0.27\tabularnewline
		\hline 
		(e) & Union Carbide Corp. New & AON Corp. & 0.353 & 0.32\tabularnewline
		\hline 
		(f) & General Motors Corp. & Boeing Co. & 0.214 & 0.14\tabularnewline
		\hline 
	\end{tabular}
	\label{TabR-1}
\end{table}

A few more general considerations on the validity of the power series
truncation heuristic can be made. The size of the interval containing
any interlacing points is dictated to a large extent by how quickly
the rankings based on the measures ${\mathscr{C}}_{i}(\zeta)$ (or
$\tilde{\mathscr{C}}_{i}(\zeta)$) stabilize near the rankings obtained
using eigenvector centrality. This, in turn, depends on the spectral
gap $\lambda_{1}-\lambda_{2}$: the larger the gap, the faster the
eigenvector centrality rankings are approached for increasing values
of $\zeta$. Hence, in the case of relatively large gaps, we expect
any interlacing values to occur for fairly small values of $\zeta$.
In this case, the heuristics based on polynomial approximations may
be justified, since interlacing is likely to occur already for small
values of $\zeta$. As is well known, however, it is not easy to determine
when the spectral gap is ``sufficiently large''. On the other hand,
when the spectral gap is tiny, then the interval $[0,\bar{\zeta}]$
is going to be larger and therefore there is ``more room'' for the
occurrence of interlacing. Unfortunately, in this case it is not clear
that polynomial truncation will be effective in approximately locating
the interlacing points. In this case, a possible solution is to expand
the functions ${\mathscr{C}}_{i}(\zeta)$ not around the value $\zeta=0$,
but also around a few values $\zeta_{0}>0$. This strategy can also
be used to find a possible second point of interlacing after having
found a first such point $\zeta^{*}$. Expanding around $\zeta^{*}$
leads to 
\begin{align}
& \Psi(\zeta^{*}+\eta)={\mathscr{C}}_{i}(\zeta^{*}+\eta)-{\mathscr{C}}_{j}(\zeta^{*}+\eta)= \nonumber \\
& \frac{1}{2!}(w_{i,i}^{(2)}-w_{j,j}^{(2)})\eta^{2}+\frac{1}{3!}(w_{i,i}^{(3)}-w_{j,j}^{(3)})\eta^{3}+\cdots+\frac{1}{k!}(w_{i,i}^{(k)}-w_{j,j}^{(k)})\eta^{k}+O(\eta^{k+1}). \nonumber 
\end{align}
Dividing by $\eta^{2}$ and setting the result equal to zero leads
to an algebraic equation of degree $k-2$ for $\eta$; the smallest
positive root $\eta^{*}$ of this equation, if there are any, leads
to the approximation $\zeta^{*}+\eta^{*}$ for the next interlace
point, and so forth.

Completely analogous considerations apply to the approximation of
interlacing points when the ranking of nodes is done according to
the risk-based total communicability measure ${\mathscr{R}}_{i}(\zeta)$.
In this case the transcendental equation to be solved is given by
\[
\chi(\zeta)={\mathscr{R}}_{i}(\zeta)-{\mathscr{R}}_{j}(\zeta)=\sum_{k=1}^{n}e^{\zeta\lambda_{k}}\left(\psi_{k}^{T}{\vec 1}\right)\left[\psi_{k,i}-\psi_{k,j}\right]=0.
\]
Let $w_{i}^{(k)}=(A^{k}{\vec 1})_{i}$. Then, truncating the series
expansion \cref{series1} and dividing by $\zeta>0$ leads to the
approximation 
\begin{equation}
	\frac{(w_{i}^{(k)}-w_{j}^{(k)})}{k!}\zeta^{k-1}+\cdots+\frac{(w_{i}^{(2)}-w_{j}^{(2)})}{2!}\zeta+(w_{i,i}^{(2)}-w_{j,j}^{(2)})=0\label{poly3}
\end{equation}
for the equation whose smallest positive solution approximates the
first interlacement value for the rankings of nodes $i$ and $j$,
assuming it exists; here again $k\ge k_{0}$ where now $k_{0}\ge2$
is the smallest integer value for which the sequence $\{w_{i}^{(k)}-w_{i}^{(k)}\}_{k}$
changes sign. The simplest possible case is when $k=k_{0}=2$, which
occurs when $w_{i,i}^{(2)}-w_{j,j}^{(2)}$ and $w_{i}^{(2)}-w_{j}^{(2)}=(A^{2}{\vec 1})_{i}-(A^{2}{\vec 1})_{j}$
have different sign. In this case \cref{poly3} reduces to the linear
equation 
\[
\frac{(w_{i}^{(2)}-w_{j}^{(2)})}{2}\zeta+(w_{i,i}^{(2)}-w_{j,j}^{(2)})=0,
\]
with the unique root 
\[
\zeta^{*}=2\frac{w_{i,i}^{(2)}-w_{j,j}^{(2)}}{w_{j}^{(2)}-w_{i}^{(2)}}>0.
\]

\section{Risk prediction and COVID-19}\label{sec:COV}
Since we submitted the first versions of this work, a pandemic has
been expanding from the city of Wuhan, Hubei province of China \cite{SARS-CoV-2,COVID-19} starting on December 2019.
This disease is produced by a new coronavirus
named SARS-CoV-2 \cite{SARS-CoV-2_1} and has affected in about three
months more than 200 countries around the world. The main problem
right now is of a medical nature, but as stated by Balwing and Weder
di Mauro this coronavirus is ``as contagious economically as it is
medically'' \cite{Book}. One of the most important characteristics
of this pandemic in comparison with recent ones is that it is hitting
very strongly the most important economies in the world: China, USA,
Germany, Italy, Spain. There are some preliminaries studies about
the macroeconomic impacts of this pandemic (see for instance \cite{Book}).
However, it is important to apply mathematical and computational techniques
to forecast, at regional, national and international level, the impact
of this crisis on financial institutions, corporations and small companies.
All of them are highly interconnected in a globally dependent economy,
forming series of complex networks. In this new scenario the current
work represents an opportunity for modelers to advance predictions
on the potential risks to which different institutions are submitted
to in the current situation. This modeling scenario consists of the
networks of interactions between the institutions under analysis assuming
a high infectability in the network. Using the transmissibility and
circulability measures defined here, the modeler can understand how
at risk of transmitting the crisis to others or, respectively, of staying in a cycle
of repeated economic difficulties, a company is. At
the same time, the current work allows to model how different palliative
measures taken by regional or global financial institutions in the
European Union, USA or China can impact these companies. In this
case, the modeler should drop the infectivity of the system and analyze
how the ranking of risk for the different companies changes to gain
insights about their potential recovery or bankruptcy.

\section{Conclusions}
\label{sec:conclusions}
In general, node centrality in networks are of either of two types:
(i) node centrality in networks of time-invariant topology \cite{Estrada book},
or (ii) node centrality in networks of time-dependent topology (aka
temporal networks) \cite{temporal networks}. In this work we have
developed a new concept of node centrality, depending on both the topology of the network
and the external conditions to which the network as a whole is submitted. In particular, we have
focused on global risk as the external factor by which an economic and financial network is affected.
We started by considering the ``Susceptible-Infected'' model and its
connection to the communicability functions of nodes and edges in
a network. Then, we developed a few centrality measures which depend
not only on the local and global topological environment of a node
but also on the level of infectivity stressing the system as a
whole. In this way we have been able to make predictions in financial
and economic systems about the changes in the risk-dependent centralities
of nodes as a function of the global level of infectivity in the system.
We observe that without altering the topology of the network, i.e.,
without varying any connection between the nodes, the ranking of
the nodes, according to these new centrality measures, changes significantly as the infectivity rate changes. In the real-world
networks studied here we have been able to associate those changes in the risk-dependent centrality of nodes 
with events of
the real financial and economic worlds in which these networks are
embedded. In closing, we provide here both theoretical, computational
and empirical evidences that the node centrality is not a static
function even when the topology of the system is not varying at all.
This new paradigm is expected to play a fundamental role in assessing
the robustness of financial and economic systems to the variation of the external conditions which they are submitted to.


\bibliographystyle{siamplain}
\bibliography{references}

\newpage
\appendix

	\section{Risk-dependent centrality measures for complete graphs}
	The following theorem provides a close expression for $\mathscr{R}_{i}$,
	$\mathscr{C}_{i}$ and $\mathscr{T}_{i}$ for a complete network.
	\begin{theorem}\label{thm3} The risk-dependency ${\mathscr{R}}_{i}$ for each
		node in a complete graph is given by
		\[
		{\mathscr{R}}_{i}=e^{(n-1)\zeta}
		\]
		and the circulability and transmissibility are given by
		\[
		{\mathscr{C}}_{i}(\zeta)=\frac{n-1}{n}\left[\frac{e^{(n-1)\zeta}}{n-1}+\frac{1}{e^{\zeta}}\right]\qquad{\mathscr{T}}_{i}(\zeta)=\frac{n-1}{n}\left[e^{(n-1)\zeta}-\frac{1}{e^{\zeta}}\right]
		\]
	\end{theorem}
	\begin{proof} For a complete graph, ${\psi}_{j}^{T}\cdot\vec{1}=0,\ j\neq1$,
		because of the mutual orthogonality between ${\psi}_{j},\ j\neq1$
		and the principal eigenvector ${\psi}_{1}$ of constant components.
		That is, ${\mathscr{R}}_{i}$ is completely determined by the eigenvector
		centralities $\psi_{1,i}$ which of course are equal for every node
		and equal to $\psi_{1,i}=\frac{1}{\sqrt{n}}$. Since $\lambda_{1}=n-1$,
		we obtain:
		\[
		{\mathscr{R}}_{i}=e^{\zeta\lambda_{1}}\left({\psi}_{1}^{T}\cdot\vec{1}\right)\psi_{1,i}+0=e^{(n-1)\zeta}\left(\frac{1}{\sqrt{n}}\cdot n\right)\frac{1}{\sqrt{n}}=e^{(n-1)\zeta}
		\]
		Subgraph centrality close expression for a complete graph is provided in \cite{Estrada05}:
		\[
		{\mathscr{C}}_{i}(1)=SC(i)=\frac{1}{n}\left[e^{n-1}+\frac{n-1}{e}\right]
		\]
		Multiplying each entry in $A$ by $\zeta$ and summing up the power
		series, we get
		\[
		{\mathscr{C}}_{i}(\zeta)=\frac{n-1}{n}\left[\frac{e^{(n-1)\zeta}}{n-1}+\frac{1}{e^{\zeta}}\right]
		\]
		By difference, we get ${\mathscr{T}}_{i}(\zeta)$.
	\end{proof}
	An important remark concerns the ratio $\frac{{\mathscr{C}}_{i}}{{\mathscr{R}}_{i}}$. Indeed:
	\begin{equation}
	\underset{\zeta\rightarrow+\infty}{\lim}\frac{\mathscr{C}_{i}}{\mathscr{R}_{i}}=\underset{\zeta\rightarrow+\infty}{\lim}\frac{\frac{n-1}{n}\left[\frac{e^{(n-1)\zeta}}{n-1}+\frac{1}{e^{\zeta}}\right]}{e^{(n-1)\zeta}}=\underset{\zeta\rightarrow+\infty}{\lim}\left[\frac{1}{n}+\frac{n-1}{n}\frac{1}{e^{n\zeta}}\right]=\frac{1}{n}. \label{ratio}
	\end{equation}
	Similarly,
	\[
	\underset{\zeta\rightarrow+\infty}{\lim}\frac{\mathscr{C}_{i}}{\mathscr{T}_{i}}=\frac{\frac{e^{(n-1)\zeta}}{n-1}+\frac{1}{e^{\zeta}}}{e^{(n-1)\zeta}-\frac{1}{e^{\zeta}}}=\frac{1}{n-1}.
	\]

\end{document}